\documentclass[12pt]{amsart}

\usepackage[all]{xy}
\usepackage{amsmath}
\usepackage{amsfonts}
\usepackage{amssymb}
\usepackage{amscd}
\usepackage{amsthm}
\usepackage{latexsym}
\usepackage{amsbsy}
\usepackage{graphicx}
\usepackage{epstopdf}
\usepackage{shuffle}

\input xypic

\usepackage{color}

\newtheorem{thm}{Theorem}[section]
\newtheorem{prop}[thm]{Proposition}
\newtheorem{cor}[thm]{Corollary}
\newtheorem{lem}[thm]{Lemma}

\newtheorem{defn}[thm]{Definition}
\newtheorem{rem}[thm]{Remark}
\newtheorem{ex}[thm]{Example}

\numberwithin{equation}{section}

\def\bE{{\mathbb E}}

\def\bK{{\mathbb K}}

\def\bS{{\mathbb S}}

\def\C{{\mathbb C}}

\def\N{{\mathbb N}}

\def\Q{{\mathbb Q}}
\def\R{{\mathbb R}}
\def\Z{{\mathbb Z}}
\def\K{{\mathbb K}}

\def\cA{{\mathcal A}}

\def\cC{{\mathcal C}}
\def\cD{{\mathcal D}}

\def\cH{{\mathcal H}}

\def\cL{{\mathcal L}}
\def\cM{{\mathcal M}}

\def\cP{{\mathcal P}}
\def\cQ{{\mathcal Q}}
\def\cR{{\mathcal R}}
\def\cS{{\mathcal S}}

\def\Hom{{\rm Hom}}

\def\Spec{{\rm Spec}}

\def\Tr{{\rm Tr}}

\title[Spectral Gravity on (Multifractal) Robertson-Walker Cosmologies]{Bell polynomials and Brownian bridge in Spectral Gravity models on multifractal Robertson--Walker cosmologies}
\author{Farzad Fathizadeh, Yeorgia Kafkoulis, Matilde Marcolli }
\address{Department of Mathematics, Computational Foundry, 
Swansea University,  Bay Campus, SA1 8EN, Swansea, United Kingdom \newline
\indent Max Planck Institute for Biological Cybernetics, 72076 T\"ubingen, Germany}
\email{farzad.fathizadeh@swansea.ac.uk}
\address{Division of Physics, Mathematics, and Astronomy, California Institute of Technology, Pasadena, CA 91125, USA}
\email{Yeorgia.Kafkoulis@caltech.edu}
\address{Department of Mathematics, University of Toronto,  ON  M5S 2E4, Canada \newline
\indent Perimeter Institute for Theoretical Physics, Waterloo, ON N2L 2Y5, Canada \newline \indent
Division of Physics, Mathematics, and Astronomy, California Institute of Technology, Pasadena, CA 91125, USA}
\email{matilde@math.utoronto.ca}
\email{mmarcolli@perimeterinstitute.ca}
\email{matilde@caltech.edu}

\begin{document}
\maketitle

\begin{abstract}
We obtain an explicit formula for the full expansion of the spectral action on
Robertson--Walker spacetimes, expressed in terms of Bell polynomials, using Brownian
bridge integrals and the Feynman--Kac formula. We then apply this result to the case of
multifractal Packed Swiss Cheese Cosmology models obtained from an arrangement
of Robertson--Walker spacetimes along an Apollonian sphere packing. Using Mellin
transforms, we show that the asymptotic expansion of the spectral action contains
the same terms as in the case of a single Robertson--Walker spacetime, but with 
zeta-regularized coefficients, given by values at integers of the zeta function of the 
fractal string of the radii of the sphere packing, as well as 
additional log-periodic correction terms arising from the poles (off the real line) of this
zeta function. 
\end{abstract}

\tableofcontents

\section{Introduction}\label{IntroSec}

The spectral action was proposed in the '90s by Chamseddine and Connes \cite{CCact} as
a possible action functional for gravity coupled to matter that extends to
noncommutative spaces. It was successfully applied to the construction of
particle physics models \cite{ConnesSM}, where its asymptotic expansion reconstructs the
Lagrangian of the Standard Model with right handed neutrinos and Majorana
masses, \cite{CCM}, see also \cite{WvS}.  It was also shown in \cite{CCM} that, in the gravity sector,
the asymptotic expansion of the spectral action gives rise to a modified gravity
model that includes, in addition to the Einstein--Hilbert action and the cosmological
term of General Relativity, also a conformal gravity term (Weyl curvature) and
a Gauss--Bonnet gravity term (which is non-dynamical and topological in dimension four). 

\smallskip

It was shown in \cite{CaMaTe}, \cite{MaPieTe1}, \cite{MaPieTe2} that one can
incorporate in the spectral action functional a scalar field, seen as a perturbation
of the Dirac operator. The action functional then determines a potential for
this scalar field that has the shape of a slow--roll potential, suitable for a
cosmological inflation scenario. This slow--roll potential was used in 
\cite{CaMaTe}, \cite{MaPieTe1}, \cite{MaPieTe2} to study how an action
functional model of gravity can address the cosmic topology question. For
an overview of spectral action models in cosmology see \cite{Mar-book}. 

\smallskip

The spectral action, with its full asymptotic expansion, has been computed explicitly 
for various types of solutions of the Einstein equations, including Robertson--Walker metrics
\cite{CC-RW}, \cite{FGK}, \cite{FM} and Bianchi IX gravitational instantons, 
\cite{FFM1}, \cite{FFM2}, \cite{FFM3}.

\smallskip

The starting point of this paper is the computation of \cite{CC-RW} of the 
spectral action for (Euclidean) Robertson--Walker metrics and the results 
of \cite{FGK}, showing that the coefficients of the asymptotic expansion 
of the spectral action for these metrics are recursively given by rational
functions of the scaling factor of the metric and its derivatives, 
with $\Q$-coefficients. 

\smallskip

We obtain here a different and more explicit derivation of the full expansion
of the spectral action for Robertson--Walker metrics. As in \cite{CC-RW}
it is based on Brownian bridge integrals, but a more convenient choice
of variables leads to more easily computable integrals and to a completely
explicit form of the coefficients in terms of Bell polynomials. 

\smallskip

This explicit form suggests that a deeper algebraic and combinatorial
structure is present in the asymptotic expansion of the spectral action,
at least for very regular geometries like the Robertson--Walker metrics.
This structure is closely related to the Hopf algebra structure of
renormalization in quantum field theories, manifested here through the
Fa\`a di Bruno Hopf algebra and its relation to the Bell polynomials. 

\smallskip

We then consider the case where, instead of a single Robertson--Walker 
cosmology with spatial sections given by a sphere $S^3$, we have a
multifractal arrangement in the form of a Robertson--Walker 
cosmology over an Apollonian packing of spheres. This type of
multifractal cosmology models are known as ``Packed Swiss Cheese
Cosmology", \cite{MuDy}. They model spacetimes that are isotropic
but non-homogeneous, based on a construction originally introduced
in \cite{ReeSci}. A model of the spectral action for Packed Swiss Cheese
Cosmologies was developed in \cite{BaMa}, based on a simplified static
model with constant scaling factor. We extend the results here to the
full model with an arbitrary underlying Robertson--Walker metric. 

\smallskip

We consider an Apollonian packing of spheres $S^3$ with a sequence of radii $\cL=\{ a_{n,k} \}$.
We endow each $4$-dimensional spacetime $\R \times S^3$ with a Robertson--Walker metric
$d^2+ a(t)^2 d\sigma^2$, scaled by the corresponding radius $a_{n,k}^2$ in two possible
ways, see \eqref{RWmetricank} and \eqref{RWmetricank2}. For a particular choice of a
scaling factor of the form $a(t)=\sin(t)$ this general setting includes the case of packings 
of four-spheres. 

\smallskip

We first illustrate a lower dimensional example based on a special class of
Apollonian circle packings, the Ford circles, where we show explicitly the
terms arising in the spectral action that detect the fractal structure, which
are expressible in terms of zeros of the Riemann zeta function. 
This example also illustrates the fact that the very restrictive condition on
the sphere packing used in the simplified model of \cite{BaMa}, based on
an approximation by self-similar fractal strings with lattice property, is too
strong for the general setting we need to consider here. The method we
use in this paper to obtain the full asymptotic expansion of the spectral
action is independent of this approximation assumption and only requires a milder
condition on the fractal string $\cL=\{ a_{n,k} \}$ of the sphere packing,
namely the property that the zeta function $\zeta_{\cL}(z)$ admits analytic
continuation to a meromorphic function on $\C$ with simple poles located
away from the set of integers less than or equal to $4$. 

\smallskip

We obtain the full expansion of the spectral action on these
multifractal Packed Swiss Cheese Cosmologies in terms of
the expansion for a single Robertson--Walker metric obtained
in the first part of the paper, using a Mellin transform argument.
The resulting expansion has two series of term, one that
corresponds to the terms in the expansion of a single 
Robertson--Walker metric, where the coefficients are modified
by a zeta regularized sum of powers of the packing radii, so
that the coefficients are no longer rational numbers but they
contain the zeta values $\zeta_{\cL}(4-2M)$, for $M\in \N$. 
The second series of terms corresponds to the poles of the
fractal string zeta function $\zeta_{\cL}(z)$ and give rise
to a series of log-periodic terms as already observed in the
simpler model of \cite{BaMa}. In this case the coefficients
are values of the zeta function of the Dirac operator of
the underlying model Robertson--Walker metric. 

\bigskip

\noindent {\bf Acknowledgment.}  The first author acknowledges the support from the 
Marie Curie/SER Cymru II Cofund Research Fellowship 663830-SU-008, and thanks 
the Perimeter Institute for Theoretical Physics for their hospitality in July 2018 and their 
excellent environment where this work was partially carried out.
The second author was partially supported by a Gantvoort Scholarship, 
a Mr.~and~Mrs.~Robert C.~Loschke Summer Undergraduate Research Fellowship, 
and a Taussky-Todd Prize. The third author was partially supported by 
NSF grant DMS-1707882, by NSERC Discovery Grant RGPIN-2018-04937 and Accelerator Supplement grant RGPAS-2018-522593, and by the Perimeter Institute for Theoretical Physics.

\bigskip
\section{Spectral gravity and Robertson--Walker metrics}

We discuss here some general preliminary facts about Robertson--Walker metrics
and the spectral action functional, which will be useful in the following sections.

\medskip
\subsection{Spectral gravity}

The spectral action functional can be defined on
ordinary manifolds or more generally on noncommutative geometries, described
in terms of the spectral triple formalism \cite{CoS3}. The functional is defined in
terms of a regularized trace of a Dirac operator
\begin{equation}\label{SAdef}
\cS_\Lambda= \Tr(f(\frac{D}{\Lambda})) =\sum_{\lambda\in \Spec(D)}  f(\frac{\lambda}{\Lambda}),
\end{equation}
where $\Lambda\in \R^*_+$ is an energy scale and $f(x)$ is a smooth test function (which one can think
of as a smooth approximation to a cutoff function. In the commutative case, one assumes
that the underlying manifold is Riemannian and compact, so that the Dirac operator has compact resolvent,
hence the series in \eqref{SAdef} makes sense.
The general spectral triple axioms in the noncommutative case \cite{CoS3} are modelled on the analytic properties of
Dirac operators on compact Riemannian spin manifolds. While the definition \eqref{SAdef} does not
directly extend to Lorentzian geometries, it is often the case that the local terms in the asymptotic 
expansion of the spectral action may admit Wick rotations to Lorentzian signature and can be used
in gravity and cosmology models. 

\smallskip

Throughout this paper we will work only with spacetimes with
Euclidean signature, and in particular, with Robertson--Walker metrics of the form $dt^2 + a(t)^2 d\sigma^2$
on $\R \times S^3$,  
where the geometry is given as a warped product of the flat metric on $\R$ 
and the round metric $d\sigma^2$ on the $3$-dimensional sphere $S^3$ of radius one. 
Here the real line $\R$ is used to parametrize the cosmic time $t$, and the spatial 
section of the universe is a $3$-sphere, expanded by the scaling factor $a(t)$. 

\smallskip

The asymptotic expansion of the spectral action is obtained from the heat kernel expansion
for the square $D^2$ of the Dirac operator.  Suppose that the heat kernel has a small time 
asymptotic expansion at $\tau\to 0+$ of the form
$$ \Tr(e^{-\tau D^2}) \sim \sum_\alpha c_\alpha \tau^\alpha, $$
where we assume that the terms $\alpha>0$ are integers,
then, using a test function of the form $f(x)=\int_0^\infty e^{-\tau x^2} d\mu(\tau)$ for some
measure $\mu$ with normalization $f(0)=\int_0^\infty d\mu(\tau)$, we obtain an expansion
(see \cite{CCact}, \cite{CC-RW}, \cite{WvS})
\begin{equation}\label{SAexpand}
 \Tr(f(D/\Lambda)) \sim \sum_{\alpha <0} f_\alpha \, c_\alpha \, 
\Lambda^{-\alpha} + a_0 f(0) + \sum_{\alpha>0} f_\alpha\, c_\alpha\, \Lambda^{-\alpha}, 
\end{equation}
where the coefficients $f_\alpha$ are given by
\begin{equation}\label{SAcoeffs}
 f_\alpha = \left\{ \begin{array}{ll} \int_0^\infty f(v) v^{-\alpha -1} dv & \alpha <0 \\
(-1)^\alpha f^{(\alpha)}(0) & \alpha >0, \, \alpha \in \N \end{array}\right. 
\end{equation}
Thus, the problem of computing the full spectral action expansion amounts to
computing the heat kernel expansion coefficients. 

\smallskip

In the case of $\R \times S^3$ with a Riemannian Robertson--Walker metric 
of the form 
$$ dt^2 + a(t)^2 d\sigma^2  = dt^2 +a(t)^2 (d\chi^2 \sin^2\chi (d\theta^2 +\sin^2\theta \, d\phi^2), $$ 
the square $D^2$ of the Dirac operator has the explicit form (\cite{CC-RW})
$$ D^2 = - (\frac{\partial}{\partial t} + \frac{3 a'(t)}{2 a(t)} )^2 + \frac{1}{a(t)^2} (\gamma^0 D_3)^2 -
\frac{a'(t)}{a^2(t)} \gamma^0 D_3, $$
with
$$ D_3 =\gamma^1 (\frac{\partial}{\partial \chi} +{\rm cot} \chi ) + \gamma^2 \frac{1}{\sin \chi} (\frac{\partial}{\partial\theta} + \frac{1}{2} {\rm cot}\theta)  +\gamma^3 \frac{1}{\sin\chi \sin\theta} \frac{\partial}{\partial \phi}. $$
Using a basis of eigenfunctions of the Dirac operator on $S^3$, the operator $D^2$
was decomposed into a direct sum of operators of the form
$$ H_n^\pm = - ( \frac{d^2}{dt^2} - \frac{(n+\frac{3}{2})^2}{a^2} \pm \frac{(n+\frac{3}{2}) a'}{a^2} ),  $$
which are then used to compute the spectral action expansion via a Feynman--Kac formula. We
will discuss this setting more in detail in \S \ref{CombSec1}, where we present a different method
of computing these coefficients, also in terms of the Feynman--Kac formula and the Brownian
bridge integrals, but with a computationally simpler choice of coordinates. 

\medskip
\subsection{Pseudodifferential calculus}

We let $D$ be the Dirac operator of the Robertson-Walker metric with a 
general cosmic factor $a(t)$: 
\begin{equation} \label{RWmetric1}
ds^2 = dt^2 + a(t)^2 d \sigma^2.
\end{equation}

\smallskip

We can express the heat kernel as
\begin{equation}\label{heatkerRlambda}
 e^{-\tau D^2}=\frac{1}{2\pi i}  \int_{\gamma} e^{-\tau \lambda}(D^2-\lambda)^{-1}\,d\lambda, 
\end{equation} 
for $\gamma$ a contour in the complex plane traveling clockwise around the non-negative reals.
Since $D^2$ is an elliptic operator of order 2, we can approximate $(D^2-\lambda)^{-1}$ by a parametrix 
\begin{equation}\label{parametrix}
\sigma(R_\lambda) \sim \sum \limits_{j=0}^\infty {r_{j}(x,\xi,\lambda)} ,
\end{equation}
where each of the $r_{j}(x,\xi,\lambda)$ is a pseudodifferential symbol of order $-2-j$ so 
that $r_{j}(x,\tau \xi,\tau^2\lambda)=\tau^{-2-j}r_{j}(x,\xi,\lambda)$.  It is then possible to determine
recursively the homogeneous pseudodifferential symbols $r_j$ in the expansion of the parametrix,
with $r_{0}(x,\xi, \lambda)=(p_{2}(x, \xi)-\lambda)^{-1}$, and for any $n>1$
\begin{equation}\label{recursionparam}
r_{n}(x,\xi,\lambda)=-\sum \frac{1}{\alpha!}\partial^{\alpha}_{\xi}r_{j}(x,\xi,\lambda)\partial^{\alpha}_{x}p_{k}(x,\xi)r_{0}(x,\xi,\lambda),
\end{equation}
where the summation runs over all $\alpha \in \Z_+^4$, $j \in {0,1,...,n-1}$, $k \in \{0,1,2\}$ such that $|\alpha|+j+2-k=n$, see \cite{FGK},
\cite{FFM1}. 

\smallskip

The small time asymptotic expansion of the heat kernel
\begin{equation}\label{heatkerexp}
\Tr(e^{-\tau D^2}) \sim_{\tau \to 0^+} \sum \limits_{n=0}^\infty \frac{\tau^{(n-4)/2}}{16 \pi^4} 
\int{\textnormal{tr}(e_n(x) )\, dvol_g}
\end{equation}
and \eqref{heatkerRlambda} with the parametrix expansion \eqref{parametrix} give 
\begin{equation}\label{heatparam}
e_n(x) \cdot \sqrt{det(g)}=-\frac{1}{2\pi i}  \int_{\gamma} e^{-\lambda}r_{n}(x, \xi, \lambda)\,d\lambda \,d\xi,  
\end{equation}
and the coefficients $a_n$ of the heat kernel expansion can be written as in \cite{FGK} in the form
\begin{eqnarray}\label{anparam}
a_n &=& \frac{1}{16\pi^4} \int_{S^3_{\alpha}} \textnormal{tr}(e_n) \,dvol_g \\ \nonumber
&=& \frac{1}{16\pi^4} \int^{2\pi}_{0} \int^{2\pi}_{0} \int^{\pi/2}_{0} \textnormal{tr}(e_n) a^3(t) \sin(\eta) \cos(\eta) \,d\eta \,d\phi_1 \,d\phi_2  .
\end{eqnarray}
In fact, only the even coefficients $a_{2m}$ are nontrivial.
The coefficients $a_{2m}$ can then be determined in terms of the recursive formula \eqref{recursionparam} for the parametrix.
Using this method it is proved in \cite{FGK} that the coefficients $a_n$ satisfy a rationality phenomenon conjectured in \cite{CC-RW}.
Namely, if we denote by $a_{2m}(t)$ the coefficient $a_{2m}$ prior to time-integration, that is, $a_{2m}=\int a_{2m}(t)\, dt$, 
then the result proved in \cite{FGK} shows that each $a_{2m}(t)$ is described by a polynomial in several variables whose 
coefficients are rational numbers, 
\begin{equation}
a_{2m}(t) = \frac{Q_{2m} \left (a(t), a'(t), \dots, a^{(2m)}(t) \right )}{a(t)^{2m-3}}, 
\end{equation}
where $Q_{2m} \in \Q[x_0, x_1, \dots, x_{2m}]$. Moreover, the degree of each monomial 
appearing in $Q_{2m}$ is either $2m-2$ or $2m$. More concretely,  
\begin{equation}\label{Q2m}
Q_{2m} \left (x_0, x_1, \dots, x_{2m} \right ) = \sum_{k} c_{2m, k} \, x_0^{k_0} x_1^{k_1} \cdots x_{2m}^{k_{2m}}, 
\end{equation}
where $c_{2m, k} \in \Q$, and for each multi-index $k= (k_0, k_1, \dots, k_{2m})$ in the summation we have:
\begin{equation} \label{BellRels}
\textrm{either} \qquad \sum_{j=0}^{2m} k_j = \sum_{j=0}^{2m} j k_j =2m -2  \qquad 
\textrm{or} \qquad  \sum_{j=0}^{2m} k_j 
= \sum_{j=0}^{2m} j k_j =2m. 
\end{equation}

\smallskip

As we will see later, the structure \eqref{BellRels} of the summation in the polynomials \eqref{Q2m} is
reminiscent of the structure of a combinatorially very interesting family of polynomials, the Bell
polynomials, that describe the combinatorial structure of derivatives of composite functions. Indeed
we will prove in the next sections that the coefficients $a_{2m}(t)$ can be computed explicitly in
terms of Bell polynomials.

\medskip
\subsection{Physical examples: expansion models}

The first few coefficients  $a_0, a_2, a_4, a_6$ in the expansion of the spectral action for the
Robertson--Walker metric, as computed in \cite{CC-RW}, give the following expressions
(written without time integration)
\begin{align} \nonumber
a_0(t) &=\frac{1}{2} a^3(t), \\ \nonumber
a_2(t)&= \frac{a^3(t)}{4} \bigg ( \frac{a''(t)}{a(t)} + \frac{(a^\prime(t))^{2}-1}{a^2(t)} \bigg),  \\ \nonumber
a_4(t) &= \frac{1}{120}  \bigg ( 3a^2(t)a^{(4)}(t)+9a(t)a^{\prime}(t)a^{(3)}(t)+3a(t)(a'')^{2}(t)-4(a^\prime)^2(t)a''(t)-5a''(t)   \bigg),  \\
\nonumber
a_6(t) &= -\frac{a^{\prime}(t)^{2}a''(t)}{240a^{2}(t)}-\frac{a^{\prime}(t)^{4}a''(t)}{84a^{2}(t)}+\frac{a''(t)^{2}}{120a(t)}+\frac{a^{\prime}(t)^{2}a''(t)^2}{21a(t)}-
\frac{a''(t)^{3}}{90}+\frac{a^{\prime}(t)a^{(3)}(t)}{240a(t)} \\ \nonumber &
+\frac{a^{\prime}(t)a^{(3)}(t)}{84a(t)}-
\frac{a^{\prime}(t)a''(t)a^{(3)}(t)}{20}-
\frac{a(t)a^{(3)}(t)^2}{1680}-\frac{a^{(4)}(t)}{240}- 
\frac{a^{\prime}(t)^{2}a^{(4)}(t)}{120} \\ \nonumber & +\frac{a(t)a''(t)a^{(4)}(t)}{840}+ 
\frac{a(t)a^{\prime}(t)a^{(5)}(t)}{140}+\frac{a(t)^{2}a^{(6)}(t)}{560}.  
\end{align} 

While the spectral action is computed for a compact $4$-dimensional Riemannian manifold
(for example the sphere $S^4$ for which $a(t)=\sin(t)$) the expressions obtained above
for the coefficients as functions of the scaling factor $a(t)$ of the Robertson--Walker metric
continue to make sense for more realistic universe models where the scaling factor 
describes different phases of the expansion of the universe. In the case of an expanding
universe (as opposed to the expansion and contraction case of the sphere $S^4$) the
time integration of the expressions above may introduce divergences that requires cutoff
regularization. 

\smallskip

Throughout this paper we will not assume any fixed form for the scaling factor $a(t)$ and we
will work in complete generality for an arbitrary smooth function. We list here some cosmologically 
relevant examples of scaling factors of an expanding universe and the corresponding form 
of the first few coefficients of the spectral action expansion, computed using the pseudodifferential
calculus discussed above. 

\subsubsection{Inflation dominated universe}

For an inflation dominated universe model the scaling factor derived from the Friedmann 
equations and the Robertson-Walker metric  produces an exponentially expanding universe,
\cite{FRWUni}, with scaling factor $a(t)=e^{H t}$. This gives 
\begin{eqnarray*}
a_0(t) &= &\frac{1}{2}e^{2H t}, \\
a_2(t) &=& \frac{2 H^2e^{3H t}-e^{H t}}{4}, \\
a_4(t)&=& 11 H^4e^{3Ht}  -5 H^2e^{Ht}, \\
a_6(t) &=& \frac{-31}{2510}{H^6 e^{3H t}}+\frac{1}{240} H^4 e^{H t}.
\end{eqnarray*}

\subsubsection{Radiation dominated universe}

In the radiation-dominated phase of the universe expansion the scaling factor grows like
$a(t)=(2H t)^{1/2}$. This gives terms of the form
\begin{eqnarray*}
a_0(t) &= &  \sqrt{2}(H t)^{3/2},  \\
a_2(t) &=&     -\frac{\sqrt{2H t}}{4},   \\
a_4(t) &=&  -\frac{\sqrt{2H t}(-11H +30t)}{1440t^2}, \\
a_6(t) &=& \frac{-919\cdot 2^{1/6}H^2t+189\cdot 2^{1/6}H t^2+30\cdot 6^{1/3}H(Ht)^{5/6}}{20160\cdot 2^{2/3}t^5\sqrt{Ht}} \\
& & +\frac{21\cdot 6^{1/3}t(Ht)^{5/6}+126\cdot 3^{2/3}H (Ht)^{7/6}}{20160\cdot 2^{2/3}t^5\sqrt{Ht}}.
\end{eqnarray*}

\subsubsection{Matter-dominated universe}

In a matter-dominated universe model \cite{FRWUni} the scale factor has
an expansion rate of the form $a(t)=(\frac{3}{2}H t)^{2/3}$ and gives terms 
\begin{eqnarray*}
a_0(t) &= &  \frac{9}{8}H^2t^2,  \\
a_2(t) &=& \frac{H^2}{8}-\frac{1}{4} \bigg(\frac{3}{2} \bigg)^{2/3}(H t)^{2/3},      \\
a_4(t)&=&  \frac{1}{216} \frac{H^2}{t^2} +\frac{1}{72} \bigg(\frac{2}{3} \bigg)^{1/3} \frac{H^{2/3}}{t^{4/3}}, \\
a_6(t) &=& \frac{5}{2916}\frac{H^2}{t^4}+\frac{11}{810 \cdot 2^{2/3} \cdot 3^{1/3}} \frac{H^{2/3}}{t^{10/3}}.
\end{eqnarray*}

\subsubsection{Empty universe}

In an empty universe with scaling factor $a(t)=Ht$ the spectral action coefficients take the form
\begin{eqnarray*}
a_0 &= & \frac{1}{2}(H t)^3,\\
a_2 &=& \frac{H^3t-Ht}{4}. 
\end{eqnarray*}
with vanishing higher terms.

\medskip
\section{Combinatorial structures in the spectral action expansion}\label{CombSec1}

We work here and in the rest of the paper with a Robertson--Walker metric with
an arbitrary choice of the scaling factor $a(t)$.

\smallskip

In \cite{CC-RW} a method for computing the spectral action expansion on
Robertson--Walker metrics based on the Feynman-Kac formula and Brownian
bridge integrals was developed. The main steps of their argument are summarized
as follows. Let $D^2$ be the square of the Dirac operator on a Euclidean Robertson--Walker metric.
The spectral action, for a test function of the form $f(u)=e^{-su}$, is written as
$$ \Tr(f(D^2))\sim \sum_{n\geq 0} \mu(n) \, \Tr(f(H_n)), $$
with multiplicities $\mu(n)=4(n+1)(n+2)$ and with the operator $H_n$ of the form
$$ H_n =- \frac{d^2}{dt^2} + V_n(t),  $$
\begin{equation}\label{Vnt}
 V_n(t) = \frac{(n+\frac{3}{2})}{a(t)^2} ((n+\frac{3}{2}) - a'(t)). 
\end{equation} 
In order to evaluate the trace $\Tr(e^{-s H_n})$ one then uses the Feynman-Kac formula of
\cite{Simon}, Theorem~6.6,
\begin{equation}\label{FKformula}
 e^{-s H_n}(t,t) = \frac{1}{2\sqrt{\pi s}} \int \exp (-s \int_0^1 V_n(t+\sqrt{2s} \alpha(u)) du ) \, D[\alpha] . 
\end{equation} 
Here $D[\alpha]$ denotes the Brownian bridge integrals, \cite{Simon}, where the Brownian
bridge is the Gaussian process characterized by the covariance
\begin{equation}\label{BrownBridge}
 \bE(\alpha(s)\alpha(t))=s(1-t), \ \ \ \  0\leq s\leq t \leq 1. 
\end{equation} 
One then uses the Euler--Maclaurin formula to replace the summation 
$\sum_n \mu(n) e^{-s H_n}(t,t)$ by a continuous integration over $x\geq 3/2$ which gives
$$ \int_{3/2}^\infty k_s(x)\, dx + \frac{1}{2} k_s(3/2) - \frac{k'_s(3/2)}{12} + \cdots, $$
with the functions
$$ k_s(x)=(4x^2-1) \frac{1}{2\sqrt{\pi s}} \int e^{u(b-x)x}\, D[\alpha], $$
\begin{equation}\label{uCC}
 u=s \int_0^1 a^{-2} (t+\sqrt{2s}\alpha(v))\, dv,  
\end{equation} 
$$ ub =s \int_0^1 a^\prime a^{-2} (t+\sqrt{2s}\alpha(v))\, dv. $$
The asymptotic expansion is then obtained in \cite{CC-RW} using Taylor expansions of the form
$$ \int_0^1 F(t+\sqrt{2s}\alpha(v))\, dv =F(t)+\sum_k \frac{F^{(k)}(t)}{k!}(\sqrt{2s})^k x_k(\alpha), $$
where 
\begin{equation}\label{xkalpha}
 x_k(\alpha)=\int_0^1 \alpha(v)^k\, dv. 
\end{equation} 
In computational terms, the Laurent series expansion for $b$, which is given as the quotient
\begin{equation}\label{bCC}
 b= \frac{\int_0^1 a^{-2} (t+\sqrt{2s}\alpha(v))\, dv}{\int_0^1 a^\prime a^{-2} (t+\sqrt{2s}\alpha(v))\, dv}, 
\end{equation} 
introduces a considerable amount of complication that slows down the computation.
We argue here that a simpler choice of variables significantly simplifies the computational
complexity of the terms of this expansion and provides a more transparent description of 
the resulting terms of the asymptotic expansion, which reveals the presence of a richer
combinatorial structure.

\smallskip
\subsection{A convenient choice of variables for Brownian bridge integrals}\label{UVvarSec}

We set
\begin{equation}\label{AandB}
A(t)=1/a(t), \qquad B(t)=A(t)^2.
\end{equation}
We can then write the potential $V_n(t)$ of \eqref{Vnt} in the form
\begin{equation}\label{VntAB}
V_n(t)= x^2 A(t)^2 + x A'(t) = x^2 B(t)+ x A'(t), \qquad \text{with } x = n+3/2.
\end{equation}
We then write the integral in the Feynman--Kac formula \eqref{FKformula} as
\begin{equation}\label{FKintUV}
-s \int_0^1 V_n(t+\sqrt{2s} \, \alpha(v)) \, dv 
=-x^2 U - x V, 
\end{equation}
in terms of the expressions 
\begin{equation}\label{Uvar}
U = s \int_0^1 A^2 \left (t+\sqrt{2s} \, \alpha(v) \right ) \, dv 
= s \int_0^1 B \left (t+\sqrt{2s} \, \alpha(v) \right ) \, dv,  
\end{equation}
\begin{equation}\label{Vvar}
V= s \int_0^1 A' \left(t+\sqrt{2s} \, \alpha(v) \right) \, dv.
\end{equation}
In order to do the summation over $n$ using the Poisson summation formula (cf. \cite{CCuncanny}, \cite{CC-RW}), we set
\begin{equation}\label{fsx}
f_s(x):= \left(x^2 - \frac{1}{4} \right) e^{-x^2 U - xV}, 
\end{equation}
and we obtain
\begin{equation}\label{intfsx}
\int_{-\infty}^\infty f_s(x) \, dx = \frac{\sqrt{\pi } \, e^{\frac{V^2}{4 U}} \left(-U^2+2 U+V^2\right)}{4 U^{5/2}}.
\end{equation}

Considering the variables $U$ and $V$ given in the definitions \eqref{Uvar} and \eqref{Vvar} above, the
task of computing the terms of the asymptotic expansion becomes  
significantly easier (even for the computer) compared to the series for  
the specific functions \eqref{uCC} and \eqref{bCC} of the $u$ and $b$ defined in \cite{CC-RW}. It is in particular the 
Laurent series \eqref{bCC} of $b$ in \cite{CC-RW} which creates a great amount of unnecessary computational 
difficulties. 

\smallskip

Using \eqref{fsx}, we obtain the 
function that generates the full expansion in the form
\begin{equation}\label{genfunc}
\frac{1}{\sqrt{\pi s}} \frac{\sqrt{\pi } \, e^{\frac{V^2}{4 U}} \left(-U^2+2 U+V^2\right)}{4 U^{5/2}}
=\frac{1}{\sqrt{ s}} \frac{ \, e^{\frac{V^2}{4 U}} \left(-U^2+2 U+V^2\right)}{4 U^{5/2}}.
\end{equation}
We consider then the Laurent series expansion of 
the function given by \eqref{genfunc} in the variable $s$.

\medskip
\subsection{Laurent series expansion}

In order to keep the notation concise and more efficient we let $\tau = s^{1/2}$ 
and write: 
\begin{equation}\label{UVtausum}
U = \tau^2 \sum_{n=0}^\infty \frac{u_n}{n!} \,\tau^n,  \qquad 
V = \tau^2 \sum_{n=0}^\infty \frac{v_n}{n!} \,\tau^n,  
\end{equation}
where 
\[
u_n = 
B^{(n)}(t) \,2^{n/2} x_n(\alpha) 
 = \left ( \sum_{k=0}^n \binom{n}{k}A^{(k)}(t) A^{(n-k)}(t) \right )  2^{n/2} \, x_n(\alpha) ,
\]
\[
v_n= A^{(n+1)}(t)  2^{n/2}  \, x_n(\alpha)\, , \]
with $A$ and $B$ as in \eqref{AandB} and with $x_k(\alpha)$ as in \eqref{xkalpha}.

\smallskip

\begin{lem}\label{CrmLem}
For  $r \in \R$ and $m \in \Z_{\geq 0}$, we have 
\begin{equation}\label{CrmMser}
e^{\frac{V^2}{4U}}\, U^r \,V^m = 
\tau^{2(r+m)} \sum_{M=0}^\infty C^{(r,m)}_M \tau^M, 
\end{equation}
where 
\begin{equation}\label{CrmM}
C^{(r,m)}_M = 
\sum_{\substack{ 0 \leq k, p, N \leq M \\ 0 \leq n \leq M/2 \\ N+2n=M \\ 
1 \leq \ell_1, \dots, \ell_k, q_1, \dots, q_p\leq N \\ 
\ell_1 +\cdots + \ell_k  + q_1 +\dots+ q_p = N}}
\frac{\binom{-n+r}{k} \binom{2n+m}{p}}{4^n n!}  
 u_0^{-n+r-k} v_0^{2n+m-p}
\frac{u_{\ell_1} \cdots u_{\ell_k}v_{q_1} \cdots v_{q_p} }{\ell_1! \cdots \ell_k! \, q_1! \cdots q_p!}
\end{equation}
with the convention that when $k=0$ we have $u_{\ell_1} \cdots u_{\ell_k}=1$, and when $p=0$ 
we have $v_{q_1} \cdots v_{q_p} =1$. Note that it follows that when $k=p=0$ one considers 
$u_{\ell_1} \cdots u_{\ell_k} v_{q_1} \cdots v_{q_p} =1$.
\end{lem}

\begin{proof} By direct computation we find that
\begin{eqnarray*}
e^{\frac{V^2}{4U}}\, U^r \,V^m &=& 
\sum_{n=0}^\infty \frac{1}{4^{n} n!} \, U^{-n+r} \, V^{2n+m} =  
\end{eqnarray*}
\begin{eqnarray*}
\sum_{n=0}^\infty \frac{1}{4^{n} n!} \, 
u_0^{-n+r} \,\tau^{-2n+2r}  \sum_{k=0}^\infty\, \binom{-n+r}{k} 
\left ( \sum_{\ell=1}^\infty  \frac{u_\ell }{\ell! \,u_0} \tau^\ell \right )^k
\tau^{4n+2m} v_0^{2n+m}  
\end{eqnarray*}
\begin{eqnarray*} \times \sum_{p=0}^{2n+m} 
\binom{2n+m}{p} \left ( \sum_{q=1}^\infty  \frac{v_q}{q! \, v_0} \tau^q \right )^p
\end{eqnarray*}
\begin{eqnarray*}
= \tau^{2(r+m)}\sum_{n, k, p \geq 0 }
\frac{u_0^{-n+r} v_0^{2n+m}}{4^n n!} 
\binom{-n+r}{k} \binom{2n+m}{p} \, 
\tau^{2n} \,
\left ( \sum_{\ell=1}^\infty  \frac{u_\ell}{\ell! \, u_0} \tau^\ell \right )^k
 \left ( \sum_{q=1}^\infty  \frac{v_q}{q! \, v_0} \tau^q \right )^p
\end{eqnarray*}
\begin{eqnarray*}
= \tau^{2(r+m)}\sum_{n, k, p \geq 0 }
\frac{u_0^{-n+r} v_0^{2n+m}}{4^n n!} 
\binom{-n+r}{k} \binom{2n+m}{p} \, 
\tau^{2n} \,
\end{eqnarray*}
\begin{eqnarray*}
\times 
\sum_{N=0}^\infty \left ( 
\sum_{|\ell| + |q| = N}
\frac{u_{\ell_1} \cdots u_{\ell_k}}{\ell_1! \cdots \ell_k! \,u_0^k}\frac{v_{q_1} \cdots v_{q_p}}{q_1! \cdots q_p! \, v_0^p} 
\right ) 
\tau^N
\end{eqnarray*}
where the inner summation is over all $\ell_1, \dots, \ell_k \geq 1$ and 
$q_1, \dots, q_p \geq 1$ such that $\ell_1 +\cdots + \ell_k  + q_1 +\dots+ q_p = N$ 
with the convention that when $p=k=0$ the sum is equal to 1. For simplicity we 
denoted this condition by $|\ell| + |q| = N $. 
\end{proof}

We then obtain an explicit formula for the general term in the full expansion 
of the spectral action.

\begin{thm} 
\label{fullexpansion1}
As $\tau = s^{1/2} \to 0^+,$ we have: 
\[
\Tr(\exp(-\tau^2 D^2)) 
\sim \sum_{M=0}^\infty  \,\tau^{2M-4}   \int a_{2M}(t) \, dt,  
\]
where 
\begin{equation}\label{a0tC0}
 a_0(t) = \frac{1}{2}  C_0^{(-3/2,0)}, 
\end{equation}
and 
\begin{equation}\label{a2MtC2M}
a_{2M}(t) = 
\int \left ( 
\frac{1}{2} C_{2M}^{(-3/2,0)} 
+ \frac{1}{4} \left ( C_{2M-2}^{(-5/2,2)}  - C_{2M-2}^{(-1/2,0)}  \right ) 
\right ) D[\alpha], 
\qquad M \in \Z_{\geq1}. 
\end{equation}
\end{thm}

\proof
Using Lemma~\ref{CrmLem} we write the desired expansion for the function given by \eqref{genfunc} as
\begin{equation}\label{intfsexpand}
\begin{array}{c}
\displaystyle{ \frac{1}{\tau} \frac{ \, e^{\frac{V^2}{4 U}} \left(-U^2+2 U+V^2\right)}{4 U^{5/2}} 
=  }\\[3mm]
\displaystyle{  \frac{1}{4} \sum_{M=0}^\infty  \left ( C_M^{(-5/2,2)} - C_M^{(-1/2,0)}  \right ) \tau^{M-2}
+
\frac{1}{2} \sum_{M=0}^\infty   C_M^{(-3/2,0)}   \tau^{M-4}. }
\end{array}
\end{equation}
The statement of the theorem then follows directly from this expression
for the expansion of \eqref{genfunc}.
\endproof

\smallskip

The form \eqref{Q2m} with the relations \eqref{BellRels} of the terms in
the expansion of the spectral action for Robertson--Walker metrics,
obtained in \cite{FGK}, suggests that the explicit terms obtained above
should be expressible in terms of Bell polynomials. We show in the
next subsections that this is indeed the case. 

\medskip
\subsection{Bell polynomials}\label{BellSec}
Bell polynomials arise naturally in the Fa\`a di Bruno formula that expresses the derivatives of composite functions, \cite{Rota}, \cite{Riordan}
\begin{equation}\label{FdBeq}
\frac{d^n}{dt^n} f(g(t)) = \sum_{m=1}^n f^{(m)}(g(t)) \, B_{n,m}(g^\prime(t), g^{\prime\prime}(t), \ldots, g^{(n-m+1)}(t)).
\end{equation}

\smallskip

More precisely, 
the multivariable Bell polynomials are defined as
\begin{equation}\label{Bell}
B_{\beta,k}(x_1, \dots, x_{\beta-k+1}) = 
\sum_\lambda \frac{\beta!}{\lambda_1 ! \lambda_2 ! \cdots \lambda_{\beta-k+1}!}
\left ( \frac{x_1}{1!} \right )^{\lambda_1} \left ( \frac{x_2}{2!} \right )^{\lambda_2} 
\cdots \left ( \frac{x_{\beta-k+1}}{(\beta-k+1)!} \right )^{\lambda_{\beta-k+1}} 
\end{equation}
where the summation is over all sequences $\lambda=(\lambda_1, \lambda_2, \dots)$ of 
non-negative integers such that  
\[
|\lambda|' := \sum_{i=1}^\infty i \,\lambda_i= \beta,  
\qquad
|\lambda|:=\sum_{i=1}^\infty \lambda_i  = k. 
\]
Note that these conditions imply that 
$\lambda_ {\beta-k+2}=\lambda_ {\beta-k+3}=\cdots=0$. We shall use the following conventions: 
\[
B_{0,0}(x_1)=1,
\] 
\[
B_{\beta, 0}(x_1, \dots, x_{\beta+1}) = 0,  \qquad \beta >0,   
\]
\[
B_{\beta, k} =0,  \qquad 0 \leq \beta < k.
\]

\smallskip

The structure of the polynomials \eqref{Q2m} that arise in the spectral action expansion of the Robertson--Walker metric \cite{FGK} suggests that
the combinatorial structure of the asymptotic expansion may be describable in terms of Bell polynomial, arising from the time-derivatives
of expressions depending on the scaling factor $a(t)$ in the recursive formula \eqref{recursionparam}. 
Although it may be possible to see this directly at the level of the recursive formula
obtained by the pseudodifferential calculus, it seems difficult to control the terms explicitly by that method.
We shown instead that the explicit form of the terms of the asymptotic expansion obtained
in Theorem~\ref{fullexpansion1} above can indeed be expressed directly in terms of Bell polynomials. 

\smallskip 

\begin{prop} \label{CrmBellPoly}
For  $r \in \R$ and $m, M \in \Z_{\geq 0}$, we have: 
\[
C^{(r,m)}_{2M} =
\sum_{\substack{ 0 \leq k, p \leq2 M \\ 
0 \leq n \leq M  \\ 
0 \leq \beta \leq 2M-2n}}
\Bigg ( 
\frac{\binom{-n+r}{k} \binom{2n+m}{p} \binom{2M-2n}{\beta}\, k! \, p!}{4^n \,n! \,(2M-2n)!}  
 u_0^{-n+r-k} v_0^{2n+m-p} \times 
\]
\[
B_{\beta,k} \left (u_1, \dots, u_{\beta-k+1} \right )
B_{2M-2n-\beta,\, p} \left ( v_1, \dots, v_{2M-2n-\beta-p+1}\right ) 
\Bigg ). 
\]
\end{prop}

\begin{proof} We have
\[
C^{(r,m)}_{2M} = 
\sum_{\substack{ 0 \leq k, p, N \leq 2M \\ 0 \leq n \leq M \\ N+2n=2M \\ 
1 \leq \ell_1, \dots, \ell_k, q_1, \dots, q_p\leq N \\ 
\ell_1 +\cdots + \ell_k  + q_1 +\dots+ q_p = N}}
\frac{\binom{-n+r}{k} \binom{2n+m}{p}}{4^n n!}  
 u_0^{-n+r-k} v_0^{2n+m-p} 
 \frac{u_{\ell_1} \cdots u_{\ell_k}v_{q_1} \cdots v_{q_p} }{\ell_1! \cdots \ell_k! \, q_1! \cdots q_p!}
\]
\[
=
\sum_{\substack{ 0 \leq k, p \leq2 M \\ 
0 \leq n \leq M \\  
1 \leq \ell_1, \dots, \ell_k, q_1, \dots, q_p\leq 2M-2n \\ 
\ell_1 +\cdots + \ell_k  + q_1 +\dots+ q_p = 2M-2n}}
\frac{\binom{-n+r}{k} \binom{2n+m}{p}}{4^n n!}  
 u_0^{-n+r-k} v_0^{2n+m-p}
\frac{u_{\ell_1} \cdots u_{\ell_k}v_{q_1} \cdots v_{q_p} }{\ell_1! \cdots \ell_k! \, q_1! \cdots q_p!}
\]
\[
=
\sum_{\substack{ 0 \leq k, p \leq2 M \\ 
0 \leq n \leq M }}
\frac{\binom{-n+r}{k} \binom{2n+m}{p}}{4^n n!}  
 u_0^{-n+r-k} v_0^{2n+m-p} \]\[ \times \,\,
\sum_{\beta=0}^{2M-2n}
\sum_{\substack{ 1 \leq \ell_1, \dots, \ell_k, q_1, \dots, q_p\leq 2M-2n \\ 
\ell_1 +\cdots + \ell_k  = \beta \\
q_1 +\dots+ q_p =2M-2n-\beta}} 
\frac{u_{\ell_1} \cdots u_{\ell_k}v_{q_1} \cdots v_{q_p} }{\ell_1! \cdots \ell_k! \, q_1! \cdots q_p!}. 
\]
In the inner summation, assume the integers  $\ell_1, \dots, \ell_k$ 
consist of $\lambda_1$ copies of $1$, $\dots$, $\lambda_\beta$ copies 
of $\beta$, and the integers  $q_1, \dots, q_p$ 
consist of $\mu_1$ copies of $1$, $\dots$, $\mu_{2M-2n-\beta}$ copies 
of $2M-2n-\beta$. We then obtain  
\[
C^{(r,m)}_{2M} =
\]
\[
\sum_{\substack{ 0 \leq k, p \leq2 M \\ 
0 \leq n \leq M }}
\frac{\binom{-n+r}{k} \binom{2n+m}{p}}{4^n n!}  
 u_0^{-n+r-k} v_0^{2n+m-p}
\sum_{\beta=0}^{2M-2n}
\sum_{\lambda, \mu} 
\binom{k}{\lambda_1, \dots, \lambda_\beta}
\binom{p}{\mu_1, \dots, \mu_{2M-2n-\beta}} \times 
\]
\[
\frac{
u_{1}^{\lambda_1} \cdots u_{\beta}^{\lambda_\beta}v_{1}^{\mu_1} \cdots v_{2M-2n-\beta}^{\mu_{2M-2n-\beta}} 
}
{(1!)^{\lambda_1}\cdots (\beta!)^{\lambda_\beta} (1!)^{\mu_1}\cdots ((2M-2n-\beta)!)^{\mu_{2M-2n-\beta}} 
}=
\]
\[
\sum_{\substack{ 0 \leq k, p \leq2 M \\ 
0 \leq n \leq M }}
\frac{\binom{-n+r}{k} \binom{2n+m}{p}}{4^n n!}  
 u_0^{-n+r-k} v_0^{2n+m-p}
\Big ( \sum_{\beta=0}^{2M-2n}
\frac{k! \, p!}{\beta! \, (2M-2n-\beta)!} \times
\]
\[
\sum_{\lambda, \mu} 
\frac{\beta!}{\lambda_1! \cdots \lambda_\beta!}
\frac{(2M-2n-\beta)!}{\mu_1! \cdots (\mu_{2M-2n-\beta})!} 
\frac{
u_{1}^{\lambda_1} \cdots u_{\beta}^{\lambda_\beta}v_{1}^{\mu_1} \cdots v_{2M-2n-\beta}^{\mu_{2M-2n-\beta}} 
}
{(1!)^{\lambda_1}\cdots (\beta!)^{\lambda_\beta} (1!)^{\mu_1}\cdots ((2M-2n-\beta)!)^{\mu_{2M-2n-\beta}}
} \Big ), 
\]
where the last summation is over all sequences of non-negative integers 
$\lambda$ and $\mu$ such that 
\[
\lambda_1 + 2 \lambda_2 + \cdots + \beta \lambda_\beta = \beta, 
\qquad 
\lambda_1 +  \lambda_2 + \cdots + \lambda_\beta = k, 
\]
\[
\mu_1 + 2 \mu_2 + \cdots + (2M-2n-\beta) \mu_{2M-2n-\beta} = 2M-2n-\beta, 
\qquad 
\mu_1 +  \mu_2 + \cdots + \mu_{2M-2n-\beta} = p. 
\]
Now we can use the Bell polynomials to write: 
\[
C^{(r,m)}_{2M} =
\sum_{\substack{ 0 \leq k, p \leq2 M \\ 
0 \leq n \leq M }}
\frac{\binom{-n+r}{k} \binom{2n+m}{p} k! \, p!}{4^n\, n! \,(2M-2n)!}  
 u_0^{-n+r-k} v_0^{2n+m-p} \] \[ \times \,\, 
\Big ( 
\sum_{\beta=0}^{2M-2n} \binom{2M-2n}{\beta}
B_{\beta,k} \left (u_1, \dots, u_{\beta-k+1} \right ) \times 
B_{2M-2n-\beta,\, p} \left ( v_1, \dots, v_{2M-2n-\beta-p+1}\right ) \Big ) . 
\]
This gives the stated result.
\end{proof}

\section{Brownian bridge and combinatorial structure}

In this section we compute explicitly the Brownian bridge integrals and obtain the full
combinatorial structure of the spectral action expansion for Robertson--Walker metrics. 

\medskip
\subsection{Brownian bridge integrals}

We provide combinatorial formulas for the Brownian bridge integrals 
that we need in order to write combinatorial expressions for the integrals 
in Theorem~\ref{fullexpansion1} describing the coefficients 
$a_{2M}(t)$ in the full expansion of the spectral action for the 
Robertson-Walker metric. 

\smallskip

We need a preliminary result about integrals of monomials on the standard simplex.

\begin{lem} \label{monointegralsimplex} 
Let $\Delta^n$ denote the simplex
\[
\Delta^n = \{ (v_1, v_2, \dots, v_n) \in \mathbb{R}^n: 
0 \leq v_1 \leq v_2 \leq \cdots \leq v_n \leq 1\}. 
\]
Then the integral of a monomial $v_1^{k_1} v_2^{k_2} \cdots v_n^{k_n}$ is
given by the expression
\begin{equation}\label{intmonomial}
\int_{\Delta^n} v_1^{k_1} v_2^{k_2} \cdots v_n^{k_n} \, dv_1 \, dv_2 \cdots dv_n  
=
\frac{1}{(k_1+1)(k_1+k_2+2)\cdots (k_1+k_2+\cdots+k_n+n)}.
\end{equation}
\end{lem}

\smallskip

In particular, we have the following case.

\begin{cor}\label{intlinear}
If $1 \leq j_1 < j_2 < \cdots < j_k \leq n$ then 
\begin{equation}\label{intlin}
\int_{\Delta^n} v_{j_1} v_{j_2} \cdots v_{j_k} \, dv_1 \, dv_2 \cdots dv_n  
= \frac{j_1 (j_2 +1)(j_3+2) \cdots (j_{k}+k-1)}{(n+k)!}. 
\end{equation}
\end{cor}

\proof The case of \eqref{intlin} is a special case of Lemma~\ref{monointegralsimplex}.
The right-hand-side is directly obtained from the corresponding expression in \eqref{intmonomial}.
\endproof

\smallskip

We then consider the Brownian bridge integrals. The defining property \eqref{BrownBridge}
of the Brownian bridge implies the following.

\begin{lem} \label{BridgeIntegral1}
For $(v_1, v_2, \dots, v_{2n}) \in \Delta^{2n}$, $n \in  \mathbb{Z}_{\geq 0}$ 
the Brownian bridge integrals can be computed as 
\begin{equation}\label{bridgev1v2n}
\int \alpha(v_1) \alpha(v_2) \cdots \alpha(v_{2n}) \, D[\alpha] 
= \sum v_{i_1}(1-v_{j_1}) v_{i_2}(1-v_{j_2}) \cdots v_{i_n}(1-v_{j_n}), 
\end{equation}
where the summation is over all indices such that 
$i_1 < j_1$, $i_2 < j_2$, $\dots$, $i_n < j_n$, and 
$\{i_1, j_1, i_2, j_2, \dots, i_n, j_n\} = \{1, 2, \dots, 2n \}.$ 
\end{lem}

It is convenient to reformulate the expression \eqref{bridgev1v2n} in a slightly
different notation as follows.

\begin{cor}\label{BridgeIntegral2}
For $(v_1, v_2, \dots, v_{2n}) \in \Delta^{2n}$, $n \in  \mathbb{Z}_{\geq 0}$, we have
\begin{equation}\label{bridge2v1v2n}
\int \alpha(v_1) \alpha(v_2) \cdots \alpha(v_{2n}) \, D[\alpha] 
= \sum_{\sigma \in S^*_{2n}} 
v_{\sigma(1) }(1-v_{\sigma(2) }) v_{\sigma(3) }(1-v_{\sigma(4) }) 
\cdots v_{\sigma(2n-1)}(1-v_{\sigma(2n)}), 
\end{equation}
where $S^*_{2n}$ is the set of all permutations $\sigma$ in the 
symmetric group $S_{2n}$ such that 
$\sigma(1) < \sigma(2)$, $\sigma(3) < \sigma(4)$, \dots, $\sigma(2n-1) < \sigma(2n)$. 
\end{cor}

This then gives a reformulation that will be useful in the following.

\begin{lem} \label{BridgeIntegral3} 
For $(v_1, v_2, \dots, v_{2n}) \in \Delta^{2n}$, $n \in  \mathbb{Z}_{\geq 0}$ we have
\[
\int \alpha(v_1) \alpha(v_2) \cdots \alpha(v_{2n}) \, D[\alpha] 
=
\]
\[
\sum_{\sigma \in S^*_{2n}}
v_{\sigma(1)} v_{\sigma(3)} \cdots v_{\sigma(2n-1)}  
\left ( 1+ \sum_{k=1}^n \,\,\,
 (-1)^k 
\sum_{1 \leq j_1 < j_2 < \dots < j_k \leq n}  
v_{\sigma(2 j_1)} v_{\sigma(2 j_2)} \cdots v_{\sigma(2 j_k)}  
\right) .
\]
\end{lem}

\proof This follows directly from Lemma \ref{BridgeIntegral2}. \endproof

\begin{defn} \label{sigmaJdef}
For any  non-negative integer $n$,  
we let $\mathcal{J}_{0,n}$ be the set consisting of the 0-tuple. 
For any $k= 1, \dots, n$,  we let $\mathcal{J}_{k,n}$ be 
the set of all $k$-tuples of integers  
$J=(j_1, j_2, \dots, j_k)$ such that $1 \leq j_1 < j_2 < \dots < j_k \leq n$. 
For  any  $J \in \mathcal{J}_{k,n}$,  and 
$\sigma \in S^*_{2n}$,  we define $\sigma_J(1)$, $\sigma_J(2)$, $\dots,$ 
$\sigma_J(n+k)$ by the property that 
\[
\sigma_J(1) < \sigma_J(2) < \dots < \sigma_J(n+k),   
\]  
and that the set of such $\sigma_J$'s is given by 
\[
\{ \sigma_J(1) < \sigma_J(2) < \dots < \sigma_J(n+k) \}
=\{ \sigma(1), \sigma(3), \dots, \sigma(2n-1), \sigma(2j_1), \dots, \sigma(2j_k)\}. 
\]
\end{defn}

\smallskip

Using the notation \eqref{xkalpha} as in \cite{CC-RW}, 
\[
x_k(\alpha) = \int_{0}^1 \alpha(v)^k \, dv \, ,
\]
we can then write the Brownian bridge integrals of the $x_k(\alpha)$ in the following
form.

\begin{lem}\label{xkalphaint} 
We have 
\[
\int x_1(\alpha)^{2n} \, D[\alpha] 
= 
\int \Big (\int_0^1 \alpha(v) \, dv \Big)^{2n} \, D[\alpha]
=
\]
\[
(2n)! 
\sum_{\sigma \in S^*_{2n}}     
\sum_{k=0}^n \sum_{J \in \mathcal{J}_{k,n}} 
(-1)^k
\frac{\sigma_J(1) \left(\sigma_J(2)+1 \right) \cdots \left(\sigma_J(n+k)+n+k-1 \right)}{(3n+k)!}. 
\]
\end{lem}

\proof This is obtained using the expression of Lemma~\ref{BridgeIntegral3}  and
the notation as in Definition~\ref{sigmaJdef}. \endproof

We can then formulate the monomial Brownian bridge integrals as follows.

\begin{prop} \label{monomialbridgeintegral}
For $(v_1, v_2, \dots, v_{n}) \in \Delta^{n}$ and for 
$i_1, i_2, \dots, i_n \in \mathbb{Z}_{\geq 0}$ such that $i_1+ i_2 + \cdots+ i_n \in 2\mathbb{Z}_{\geq 0}$, 
we have
\begin{equation}\label{brownmonomials}
\begin{array}{c} \displaystyle{
\int   \alpha(v_1)^{i_1} \alpha(v_2)^{i_2} \cdots \alpha(v_n)^{i_n}  \, D[\alpha]  =} \\[3mm]
\displaystyle{
 \binom{|I|}{I}^{-1} \,\,
\frac{|I|!}{(\sqrt{-1})^{|I|}}
\frac{(-1/2)^{|I|/2}}{(|I|/2)!}  
\left ( 
\sum \binom{|I|/2}{k_{m,j}}
\sum_{r_1=0}^{K_1} \sum_{r_2=0}^{K_2} \cdots \sum_{r_n=0}^{K_n} 
\prod_{p=1}^n (-1)^{r_p} v_p^{i_p-r_p}
\right )},  \end{array}
\end{equation}
where $I=(i_1, i_2, \dots, i_n)$, 
the first summation is over all non-negative integers $k_{j,m}$, $j,m=1,2,\dots, n$ 
such that 
\[
\sum_{j,m=1}^n k_{j,m} = \frac{|I|}{2}, \qquad \qquad  
\sum_{  m =1}^{ n } (k_{j,m}+k_{m,j}) = i_j \, \textnormal{ for all } \, j=1,2, \dots, n, 
\]
and we set for each $m=1,2, \dots, n,$ 
\[
K_m:= k_{m,m} + \sum_{j=1}^{m-1} ( k_{j,m}+k_{m,j} ) .
\]
\end{prop}

\begin{proof} First observe that, for Brownian bridge integrals in
exponential form, we have the identity
\[
\int \exp \left ( \sqrt{-1} \,\sum_{j=1}^n u_j \alpha(v_j) \right ) \, D[\alpha] 
=
\exp \left ( -\frac{1}{2} \sum_{j, m =1}^n c_{j,m} u_j u_m \right ), 
\]
where the terms $c_{j,m}$ are given by 
\[
c_{j,m} = v_j(1-v_m) \quad \textnormal{if} \quad j \leq m, \qquad 
\textnormal{and}  \qquad c_{j,m} = v_m(1-v_j) \quad \textnormal{if} \quad m \leq j .
\]
This implies that we have 
\[
\frac{(\sqrt{-1})^{i_1+i_2+\cdots+i_n}}{(i_1+i_2+\cdots+i_n)!} 
\binom{i_1+i_2+\cdots + i_n}{i_1, i_2, \dots, i_n}
\int   \alpha(v_1)^{i_1} \alpha(v_2)^{i_2} \cdots \alpha(v_n)^{i_n}  \, D[\alpha]  =
\]
\[
\frac{(-1/2)^{(i_1+i_2 +\cdots + i_n)/2}}{((i_1+i_2 +\cdots + i_n)/2)!} 
\left (\textnormal{Coefficient of } u_1^{i_1} u_2^{i_2} \cdots u_n^{i_n}  \textnormal{ in } 
\big ( \sum_{j,m=1}^n c_{j,m} u_j u_m \big )^{(i_1+i_2 + \cdots +i_n)/2}  \right )
\]
\[
= \frac{(-1/2)^{(i_1+i_2 +\cdots + i_n)/2}}{((i_1+i_2 +\cdots + i_n)/2)!}  
\sum \binom{(i_1+i_2 +\cdots + i_n)/2}{k_{1,1}, k_{1, 2}, \dots, k_{1, n}, k_{2,1}, \dots ,k_{n,n}}
\prod_{j,m=1}^n c_{j,m}^{k_{j,m}}
\]
where the summation is over all non-negative integers $k_{j,m}$, $j,m=1,2,\dots, n$ 
such that 
\[
\sum_{j,m=1}^n k_{j,m} = (i_1+i_2 +\cdots + i_n)/2
\]
and, for any $j=1,2, \dots, n$, 
\[
2k_{j,j} + \sum_{ 1\leq m \leq n , m \neq j} (k_{j,m}+k_{m,j}) = i_j. 
\]
This then gives 
\[
\int   \alpha(v_1)^{i_1} \alpha(v_2)^{i_2} \cdots \alpha(v_n)^{i_n}  \, D[\alpha] 
\]
\begin{eqnarray*}
&=&\binom{|I|}{I}^{-1} \,\,
\frac{|I|!}{(\sqrt{-1})^{|I|}}
\frac{(-1/2)^{|I|/2}}{(|I|/2)!}  
\left ( 
\sum \binom{|I|/2}{k_{m,j}}
\prod_{j,m=1}^n c_{j,m}^{k_{j,m}} 
\right ) 
\end{eqnarray*}
\[
= \binom{|I|}{I}^{-1} \,\,
\frac{|I|!}{(\sqrt{-1})^{|I|}}
\frac{(-1/2)^{|I|/2}}{(|I|/2)!}  
\left ( 
\sum \binom{|I|/2}{k_{m,j}}
\prod_{m=1}^n v_m^{i_m-K_m} (1-v_m)^{K_m}
\right ),  
\]
where $ I=(i_1, i_2, \dots, i_n) $ and for each $m=1,2, \dots, n$, 
\[
K_m:= k_{m,m} + \sum_{j=1}^{m-1} ( k_{j,m}+k_{m,j} )
\]
Therefore we obtain 
\[
\int   \alpha(v_1)^{i_1} \alpha(v_2)^{i_2} \cdots \alpha(v_n)^{i_n}  \, D[\alpha] 
\]
\[
=  \binom{|I|}{I}^{-1} \,\,
\frac{|I|!}{(\sqrt{-1})^{|I|}}
\frac{(-1/2)^{|I|/2}}{(|I|/2)!}  
\left ( 
\sum \binom{|I|/2}{k_{m,j}}
\prod_{m=1}^n \sum_{r_m=0}^{K_m} (-1)^{r_m}\binom{K_m}{r_m} v_m^{i_m-r_m}
\right )
\]
\[
=  \binom{|I|}{I}^{-1} \,\,
\frac{|I|!}{(\sqrt{-1})^{|I|}}
\frac{(-1/2)^{|I|/2}}{(|I|/2)!}  
\left ( 
\sum \binom{|I|/2}{k_{m,j}}
\sum_{r_1=0}^{K_1} \sum_{r_2=0}^{K_2} \cdots \sum_{r_n=0}^{K_n} 
\prod_{p=1}^n (-1)^{r_p}  \binom{K_p}{r_p}v_p^{i_p-r_p}
\right ). 
\]
\end{proof}

We then obtain the following expression for integration over a simplex of 
Brownian bridge monomial integrals.

\begin{lem} \label{monomialbridgesimplexintegral}
For  $i_1, i_2, \dots, i_n \in \mathbb{Z}_{\geq 0}$ such that $i_1+ i_2 + \cdots+ i_n \in 2\mathbb{Z}_{\geq 0}$
we have 
\[
 \int_{\Delta^n}\int   \alpha(v_1)^{i_1} \alpha(v_2)^{i_2} \cdots \alpha(v_n)^{i_n}  \, D[\alpha] 
 \, dv_1 \, dv_2 \, \cdots dv_n
\]
\[ 
=
\binom{|I|}{I}^{-1} \,\,
\frac{|I|!}{(\sqrt{-1})^{|I|}}
\frac{(-1/2)^{|I|/2}}{(|I|/2)!}  
\left ( 
\sum \binom{|I|/2}{k_{m,j}}
\sum_{r_1=0}^{K_1} \sum_{r_2=0}^{K_2} \cdots \sum_{r_n=0}^{K_n} 
\prod_{p=1}^n \frac{(-1)^{r_p} \binom{K_p}{r_p}}{ p+ \sum_{\ell=1}^p (i_\ell - r_\ell)}
\right ), 
\]
where $I=(i_1, i_2, \dots, i_n)$ and the first summation is over all 
non-negative integers $k_{j,m}$, $j,m=1,2,\dots, n$ 
such that 
\[
\sum_{j,m=1}^n k_{j,m} = |I|/2, \qquad \qquad
 \sum_{ m =1 }^n (k_{j,m}+k_{m,j}) = i_j \, \textnormal{ for all } \,   j=1,2, \dots, n,    
\]
and where, for each $m=1,2, \dots, n$, we set
\[
K_m = k_{m,m} + \sum_{j=1}^{m-1} ( k_{j,m}+k_{m,j} ). 
\]
\end{lem}

\proof This follows directly from Lemma~\ref{monointegralsimplex} and
Proposition~\ref{monomialbridgeintegral}. \endproof

\smallskip
\subsection{Shuffle product}\label{shuffleSec}

We introduce the following notation for the integrals described combinatorially in 
Proposition~\ref{monomialbridgeintegral}. 

\begin{defn}\label{Vi1inDef}
For $(v_1, v_2, \dots, v_{n}) \in \Delta^{n}$ and for
$i_1, i_2, \dots, i_n \in \mathbb{Z}_{\geq 0}$ such that $i_1+ i_2 + \cdots+ i_n \in 2\mathbb{Z}_{\geq 0}$ 
we set   
\begin{equation}\label{Vi1in}
V(i_1, i_2, \dots, i_n) :=
\int   \alpha(v_1)^{i_1} \alpha(v_2)^{i_2} \cdots \alpha(v_n)^{i_n}  \, D[\alpha] \, .
\end{equation}
\end{defn}

\smallskip

We view $(i_1, i_2, \dots, i_n)$ as a word constructed with the letters $i_1, i_2, \dots, i_n$ 
and we extend the definition of $V$ linearly to the vector space generated by all such words. 

\smallskip

\begin{defn}\label{shuffledef}
The {\it shuffle product}  of two words $(i_1, i_2, \dots, i_p)$ and $(j_1, j_2, \dots, j_q)$ is defined 
to be the sum of the $\binom{p+q}{p}$ words obtained by interlacing the letters of the two words in such 
a way that in each term the order of the letters of each word  is preserved. The shuffle product is 
denoted by $\shuffle$. 
\end{defn}

\smallskip

\begin{lem} 
\label{BrownianintShuffle}
Assume that $2n=m_1 i_1 + m_2 i_2 +\cdots + m_r i_r$ is an even 
positive integer (where $i_1, i_2, \dots, i_r$ are distinct positive integers and 
$m_1, m_2, \dots, m_r$ are positive integers). Then 
\[
\int x_{i_1}(\alpha)^{m_1}  x_{i_2}(\alpha)^{m_2} \cdots  x_{i_r}(\alpha)^{m_r} \, D[\alpha] 
\]
\[
= m!
\int_{\Delta^{|m|}} 
V\left ( \underbrace{(i_1,  \dots, i_1)}_{m_1} \shuffle \underbrace{ (i_2,  \dots, i_2)}_{m_2} \shuffle \cdots 
\shuffle \underbrace{(i_r,  \dots, i_r)}_{m_r} \right ) 
\, dv_1\, dv_2 \cdots \, dv_{|m|}, 
\]
where 
\[
m!=(m_1!) (m_2!) \cdots  (m_r!),\qquad  |m|=m_1+m_2+\cdots+m_r  \, .
\]
\end{lem}

\proof 
It follows directly from writing 
\[
\int x_{i_1}(\alpha)^{m_1}  x_{i_2}(\alpha)^{m_2} \cdots  x_{i_r}(\alpha)^{m_r} \, D[\alpha] 
\]
\[
= 
\int \Big (\int_0^1 \alpha(v_1)^{i_1} \, dv_1 \Big)^{m_1} \Big (\int_0^1 \alpha(v_2)^{i_2} \, dv_2 \Big)^{m_2} \cdots \Big (\int_0^1 \alpha(v_r)^{i_r} \, dv_r \Big)^{m_r}  \, D[\alpha], 
\]
and considering definitions \ref{Vi1inDef} and  \ref{shuffledef}. 
\endproof

Note that Lemma~\ref{monomialbridgesimplexintegral} provides a formula for 
$\int_{\Delta^n} V(i_1, \dots, i_n) \,  dv_1 \cdots \, dv_n$. Thus, using Lemma~\ref{BrownianintShuffle}
and  Lemma~\ref{monomialbridgesimplexintegral} we have achieved a 
combinatorial description of all the Brownian bridge integrals involved in the calculation of
the spectral action expansion.

\smallskip
\subsection{The  integrals in terms of the Dawson function}\label{DawsonSec}
In Lemma~\ref{monomialbridgesimplexintegral} we gave a combinatorial formula for 
\[
\int_{\Delta^n}\int   \alpha(v_1)^{i_1} \alpha(v_2)^{i_2} \cdots \alpha(v_n)^{i_n}  \, D[\alpha] 
 \, dv_1 \, dv_2 \, \cdots dv_n.  
\]
A crucial fact that we used for deriving the combinatorial formula is that since the Brownian 
bridge is a Gaussian process, for $(v_1, \dots, v_n) \in \Delta^n$ we have: 
\[
\int \exp \left ( \sqrt{-1} \,\sum_{j=1}^n u_j \alpha(v_j) \right ) \, D[\alpha] 
=
\exp \left ( -\frac{1}{2} \sum_{j, m =1}^n c_{j,m} u_j u_m \right ), 
\]
where 
\[
c_{j,m} = v_j(1-v_m) \quad \textnormal{if} \quad j \leq m, \qquad 
\textnormal{and}  \qquad c_{j,m} = v_m(1-v_j) \quad \textnormal{if} \quad m \leq j. 
\]
Therefore, if   
$i_1, i_2, \dots, i_n \in \mathbb{Z}_{\geq 0}$ and $i_1+ i_2 + \cdots+ i_n \in 2\mathbb{Z}_{\geq 0}$, then, 
setting  $I=(i_1, \dots, i_n)$,   
\[
\frac{\left ( \sqrt{-1} \right ) ^{|I|} \binom{|I|}{I}}{|I|!}
\int_{\Delta^n}\int   \alpha(v_1)^{i_1} \alpha(v_2)^{i_2} \cdots \alpha(v_n)^{i_n}  \, D[\alpha] 
 \, dv_1 \, dv_2 \, \cdots dv_n 
 \]
is equal to the coefficient of $u_1^{i_1} u_2^{i_2}\cdots  u_n^{i_n}$ in the Maclaurin series of 
\begin{equation} 
\label{intleadingtoDawson}
\int_{\Delta^n} \exp \left ( -\frac{1}{2} \sum_{j, m =1}^n c_{j,m} u_j u_m \right ) \, dv_1 \, dv_2 \cdots dv_n.  
\end{equation}
By writing the expansion of the integrand in the latter we derived the combinatorial formula 
presented in Lemma~\ref{monomialbridgesimplexintegral}. 

\smallskip

It is natural to ask whether there is a closed formula for the result of the integral 
given by \eqref{intleadingtoDawson}. 
It turns out that it is possible to obtain such closed
expressions in terms of the Dawson function 
\begin{equation}\label{Dawsonfunction}
F(x):=\exp \left(-x^2\right) \int_0^x \exp \left(y^2\right) \, dy. 
\end{equation}
We show the first few cases of the integral \eqref{intleadingtoDawson}
and their explicit form in terms of the function \eqref{Dawsonfunction}.

\smallskip

\begin{ex}\label{Dawsonexamples}{\rm
When $n=1,2,3$ we find the following explicit expressions:
\[
\int_{\Delta^1} \exp \left ( -\frac{1}{2}  c_{1,1} u_1^2 \right ) \, dv_1  =
\frac{2 \sqrt{2} F\left(\frac{u_1}{2 \sqrt{2}}\right)}{u_1}, 
\]
\[
\int_{\Delta^2} \exp \left ( -\frac{1}{2} \sum_{j, m =1}^2 c_{j,m} u_j u_m \right ) \, dv_1 \, dv_2 
=
\frac{4 \sqrt{2} \left(F\left(\frac{u_1}{2 \sqrt{2}}\right)+F\left(\frac{u_2}{2 \sqrt{2}}\right)-F\left(\frac{u_1+u_2}{2 \sqrt{2}}\right)\right)}{u_1 u_2 \left(u_1+u_2\right)},  
\]
\[
\int_{\Delta^3} \exp \left ( -\frac{1}{2} \sum_{j, m =1}^3 c_{j,m} u_j u_m \right ) \, dv_1 \, dv_2 \, dv_3=
\]
\begin{center}
\begin{math}
\frac{8 \sqrt{2} F\left(\frac{u_1}{2 \sqrt{2}}\right)}{u_1 u_2 \left(u_1+u_2\right) \left(u_2+u_3\right) \left(u_1+u_2+u_3\right)}+\frac{8 \sqrt{2} F\left(\frac{u_2}{2 \sqrt{2}}\right)}{u_1 u_2 \left(u_1+u_2\right) \left(u_2+u_3\right) \left(u_1+u_2+u_3\right)}-\frac{8 \sqrt{2} F\left(\frac{u_1+u_2}{2 \sqrt{2}}\right)}{u_1 u_2 \left(u_1+u_2\right) \left(u_2+u_3\right) \left(u_1+u_2+u_3\right)}+\frac{8 \sqrt{2} F\left(\frac{u_2}{2 \sqrt{2}}\right)}{u_1 \left(u_1+u_2\right) u_3 \left(u_2+u_3\right) \left(u_1+u_2+u_3\right)}-\frac{8 \sqrt{2} F\left(\frac{u_1+u_2}{2 \sqrt{2}}\right)}{u_1 \left(u_1+u_2\right) u_3 \left(u_2+u_3\right) \left(u_1+u_2+u_3\right)}-\frac{8 \sqrt{2} F\left(\frac{u_2+u_3}{2 \sqrt{2}}\right)}{u_1 \left(u_1+u_2\right) u_3 \left(u_2+u_3\right) \left(u_1+u_2+u_3\right)}+\frac{8 \sqrt{2} F\left(\frac{u_1+u_2+u_3}{2 \sqrt{2}}\right)}{u_1 \left(u_1+u_2\right) u_3 \left(u_2+u_3\right) \left(u_1+u_2+u_3\right)}+\frac{8 \sqrt{2} F\left(\frac{u_2}{2 \sqrt{2}}\right)}{u_2 \left(u_1+u_2\right) u_3 \left(u_2+u_3\right) \left(u_1+u_2+u_3\right)}+\frac{8 \sqrt{2} F\left(\frac{u_3}{2 \sqrt{2}}\right)}{u_2 \left(u_1+u_2\right) u_3 \left(u_2+u_3\right) \left(u_1+u_2+u_3\right)}-\frac{8 \sqrt{2} F\left(\frac{u_2+u_3}{2 \sqrt{2}}\right)}{u_2 \left(u_1+u_2\right) u_3 \left(u_2+u_3\right) \left(u_1+u_2+u_3\right)}. 
\end{math}
\end{center}
An explicit expression for
\[
\int_{\Delta^4} \exp \left ( -\frac{1}{2} \sum_{j, m =1}^4 c_{j,m} u_j u_m \right ) \, dv_1 \, dv_2 \, dv_3 \, dv_4
\]
in terms of the Dawson function \eqref{Dawsonfunction} is included in the Appendix.
}\end{ex}

\smallskip
\subsection{Combinatorial description of the full spectral action expansion}\label{fullcombSec}

In Proposition~\ref{CrmBellPoly} we showed that
\[
C^{(r,m)}_{2M} =
\sum_{\substack{ 0 \leq k, p \leq2 M \\ 
0 \leq n \leq M \\ 
0 \leq \beta \leq 2M-2n}}
\frac{\binom{-n+r}{k} \binom{2n+m}{p}}{4^n n!}  
 u_0^{-n+r-k} v_0^{2n+m-p} 
k! \, p!
\left ( \sum_{\lambda, \mu}
\prod_{i =1}^\infty \frac{\left ( \frac{u_i}{i!} \right )^{\lambda_i}  \left( \frac{v_i}{i!} \right )^{\mu_i}}{\lambda_i ! \, \mu_i !}  
\right ), 
\]
where, for each fixed $k, p, n, \beta$,  the inner summation is over all sequences $\lambda=$$(\lambda_1, \lambda_2, \dots )$ 
and $\mu=$$(\mu_1, \mu_2, \dots )$ of non-negative integers such that 
\[
|\lambda|' := \sum_{i=1}^\infty i \,\lambda_i= \beta, 
\qquad 
|\lambda|:=\sum_{i=1}^\infty \lambda_i  = k, \qquad
|\mu|' =  2M-2n-\beta, 
\qquad 
|\mu|  = p. 
\]
Note that these conditions imply that only finitely many $\lambda_i$ and $\mu_i$ can be 
non-zero, namely: $\lambda_ {\beta-k+2}=\lambda_ {\beta-k+3}=\cdots=0$ 
and $\mu_{2M-2n-\beta-p+2}=\mu_{2M-2n-\beta-p+3}= \cdots = 0$.

\begin{lem} \label{CrmBrownianint}
For  $r \in \R$ and $m, M \in \Z_{\geq 0}$ the Brownian bridge integral of the 
expressions $C^{(r,m)}_{2M}$ above gives
\[
\int C^{(r,m)}_{2M} \, D[\alpha]=
\]
\[
\sum \Bigg (  
\frac{\binom{-n+r}{k} \binom{2n+m}{p} k! \, p! }{4^n \,2^{n-M} \, n! }   
\int_{\Delta^{k+p}} 
V\big ( \underbrace{(1,  \dots, 1)}_{\lambda_1 +\mu_1} \shuffle \underbrace{ (2,  \dots, 2)}_{\lambda_2+\mu_2} \shuffle \cdots  \big ) 
\, dv_1 \cdots dv_{k+p} \times
\]
\[ 
B(t)^{-n+r-k}\,\left (A'(t) \right)^{2n+m-p}
\prod_{i =1}^\infty \binom{\lambda_i + \mu_i}{\lambda_i}\left ( \frac{B^{(i)}(t)}{i!} \right )^{\lambda_i}  \left( \frac{A^{(i+1)}(t)}{i!} \right )^{\mu_i}
\Bigg ), 
\]
where the summation is over all integers $0 \leq k, p \leq2 M, 
0 \leq n \leq M,$ $0 \leq \beta \leq 2M-2n,$ and over all 
sequences $\lambda=$$(\lambda_1, \lambda_2, \dots )$ 
and $\mu=$$(\mu_1, \mu_2, \dots )$ of non-negative integers for each choice 
of $k, p, n, \beta$, such that
$|\lambda|'=\beta , |\lambda|=k,$ $|\mu|' =  2M-2n-\beta, |\mu|=p$. 
\end{lem}
\begin{proof}
This follows directly from the above fact from Proposition~\ref{CrmBellPoly} using 
Lemma \ref{BrownianintShuffle}. 
\end{proof}

In Theorem~\ref{fullexpansion1} we showed that the coefficients  
appearing in the asymptotic expansion 
\[
\Tr(\exp(-\tau^2 D^2)) 
\sim \sum_{M=0}^\infty  \,\tau^{2M-4}   \int a_{2M}(t) \, dt  \qquad (\textnormal{as } \, \tau \to 0^+)
\]
are given by 
\[
a_0(t) = \frac{1}{2}  C_0^{(-3/2,0)}, 
\]
and 
\[
a_{2M}(t) = 
\int \left ( 
\frac{1}{2} C_{2M}^{(-3/2,0)} 
+ \frac{1}{4} \left ( C_{2M-2}^{(-5/2,2)}  - C_{2M-2}^{(-1/2,0)}  \right ) 
\right ) D[\alpha], 
\qquad M \in \Z_{\geq1}. 
\]

The result of Lemma~\ref{CrmBrownianint} gives a combinatorial 
description of $\int C^{(r,m)}_{2M} \, D[\alpha]$, hence we can write a combinatorial 
formula for an arbitrary coefficient in the full expansion of the spectral action 
for the Robertson--Walker metric with the expansion factor $a(t)$. We use again
the notation $A(t)=1/a(t), B(t)=A(t)^2$ as in \eqref{AandB} and the expressions
$V\left ( (1,  \dots, 1) \shuffle (2,  \dots, 2) \shuffle \cdots  \right )$ as in
Definition~\ref{Vi1inDef}.

\smallskip

\begin{thm}\label{a2Mcomb}
For any $M \in \Z_{\geq1}$ the coefficients of the expansion of the
spectral action of a Robertson--Walker metric are given by
\[
a_{2M}(t)=
\] 
\[
\frac{1}{2}\sum{}^{'} \Bigg (  
\frac{\binom{-n - 3/2}{k} \binom{2n}{p} k! \, p! }{4^n \,2^{n-M} \, n! }   
\int_{\Delta^{k+p}} 
V\big ( \underbrace{(1,  \dots, 1)}_{\lambda_1 +\mu_1} \shuffle \underbrace{ (2,  \dots, 2)}_{\lambda_2+\mu_2} \shuffle \cdots  \big ) 
\, dv_1 \cdots dv_{k+p} \times
\]
\[ 
B(t)^{-n-(3/2)-k}\,\left (A'(t) \right)^{2n-p}
\prod_{i =1}^\infty \binom{\lambda_i + \mu_i}{\lambda_i}\left ( \frac{B^{(i)}(t)}{i!} \right )^{\lambda_i}  \left( \frac{A^{(i+1)}(t)}{i!} \right )^{\mu_i}
\Bigg )
\]
\[
+\frac{1}{4}\sum{}^{''} \Bigg (  
\left (  
\binom{-n-5/2}{k} \binom{2n+2}{p} B(t)^{-5/2} \left ( A'(t) \right )^2 
-
 \binom{-n-1/2}{k} \binom{2n}{p} B(t)^{-1/2}   
\right )  \times 
\]
\[
\frac{ k! \, p! }{4^n \,2^{n-M} \, n! } 
\int_{\Delta^{k+p}} 
V\big ( \underbrace{(1,  \dots, 1)}_{\lambda_1 +\mu_1} \shuffle \underbrace{ (2,  \dots, 2)}_{\lambda_2+\mu_2} \shuffle \cdots  \big ) 
\, dv_1 \cdots dv_{k+p} \times
\]
\[ 
B(t)^{-n-k}\,\left (A'(t) \right)^{2n-p}
\prod_{i =1}^\infty \binom{\lambda_i + \mu_i}{\lambda_i}\left ( \frac{B^{(i)}(t)}{i!} \right )^{\lambda_i}  \left( \frac{A^{(i+1)}(t)}{i!} \right )^{\mu_i}
\Bigg ), 
\]
where the structure of the summations is as follows. 
The first summation $\sum{}^{'}$is over all integers $0 \leq k, p \leq2 M, 
0 \leq n \leq M,$ $0 \leq \beta \leq 2M-2n,$ and over all 
sequences $\lambda=$$(\lambda_1, \lambda_2, \dots )$ 
and $\mu=$$(\mu_1, \mu_2, \dots )$ of non-negative integers (for each choice 
of $k, p, n, \beta$) such that
$|\lambda|'=\beta , |\lambda|=k,$ $|\mu|' =  2M-2n-\beta, |\mu|=p$. Similarly, 
the second summation $\sum{}^{''}$ is over all integers $0 \leq k, p \leq2 M-2, 
0 \leq n \leq M-1,$ $0 \leq \beta \leq 2M-2-2n,$ and over all 
sequences $\lambda=$$(\lambda_1, \lambda_2, \dots )$, $\mu=$ $(\mu_1, \mu_2, \dots )$of non-negative integers such that
$|\lambda|'=\beta , |\lambda|=k,$ $|\mu|' =  2M-2-2n-\beta, |\mu|=p$. 
\end{thm}

\begin{proof}
The result follows directly from Theorem~\ref{fullexpansion1} and Lemma~\ref{CrmBrownianint}.
\end{proof}

\smallskip
\subsection{Explicit form of the coefficients}

Using the combinatorial formula obtained in Theorem~\ref{a2Mcomb}, we can compute explicitly
the coefficients $a_{2M}(t)$ in terms of the expressions $A(t)=1/a(t), B(t)=A(t)^2$. For the first
few coefficients this gives the following expressions: 
\[
a_0(t) = \frac{1}{2 B(t)^{3/2}}, 
\]
\[
a_2(t) = \frac{3 A'(t)^2}{8 B(t)^{5/2}}-\frac{B''(t)}{8 B(t)^{5/2}}+\frac{5 B'(t)^2}{32 B(t)^{7/2}}-\frac{1}{4 B(t)^{1/2}},
\]

\[
a_4(t)=
\]
\begin{center}
\begin{math}
\frac{A''(t)^2}{16 B(t)^{5/2}}-\frac{5 A'(t)^2 B''(t)}{32 B(t)^{7/2}}+\frac{35 A'(t)^2 B'(t)^2}{128 B(t)^{9/2}}+\frac{5 A'(t)^4}{64 B(t)^{7/2}}-\frac{A'(t)^2}{16 B(t)^{3/2}}+\frac{A^{(3)}(t) A'(t)}{8 B(t)^{5/2}}-\frac{5 A'(t) A''(t) B'(t)}{16 B(t)^{7/2}}-\frac{B^{(4)}(t)}{80 B(t)^{5/2}}+\frac{3 B''(t)^2}{64 B(t)^{7/2}}+\frac{B''(t)}{48 B(t)^{3/2}}+\frac{105 B'(t)^4}{1024 B(t)^{11/2}}-\frac{B'(t)^2}{64 B(t)^{5/2}}+\frac{B^{(3)}(t) B'(t)}{16 B(t)^{7/2}}-\frac{77 B'(t)^2 B''(t)}{384 B(t)^{9/2}}. 
\end{math}
\end{center}
Similar explicit expressions for the coefficients $a_6(t)$ and $a_8(t)$ are reported in Appendix B. When written in terms of the scaling factor $a(t)$
through the relation \eqref{AandB}, these expressions agree with those computed in \cite{CC-RW}, \cite{FGK}.

\smallskip
\subsection{The Fa\`a di Bruno Hopf algebra}

The Bell polynomials and the Fa\`a di Bruno formula have a Hopf algebra interpretation, where one considers the group $G^{{\rm diff}}(A)$
of formal diffeomorphisms tangent to the identity, 
$$ f(t) = t + \sum_{n\geq 2} f_n\, t^n \in t A[[t]], $$
with $A$ a unital commutative algebra over a field $\bK$, with the product given by composition.
Viewed as an affine group scheme, $G^{{\rm diff}}$ is dual to the Fa\`a di Bruno Hopf algebra $\cH_{{\rm FdB}}$, 
\begin{equation}\label{GdiffHFdB}
G^{{\rm diff}}(A) =\Hom(\cH_{{\rm FdB}}, A).
\end{equation}
As an algebra $\cH_{{\rm FdB}}$ is a polynomial algebra $\K[x_1,x_2,x_3,\ldots, x_n,\ldots]$ in countably many variables $x_i$,
with coproduct (see \cite{FraMan})
$$ \Delta (x_n)=\sum_{m=0}^n \frac{(m+1)!}{(n+1)!} \, B_{n+1,m+1}(1, 2!x_1, 3!x_2, \ldots, (n-m+1)! x_{n-m}) \otimes x_m $$
$$ = \sum_{m=0}^n \left( \sum_{\substack{k_0+k_1+ k_2+\cdots+k_n =m +1 \\ k_1+2k_2+\cdots +n k_n =n-m}}
\frac{(m+1)!}{k_0! k_1!\cdots k_n!} \prod_{i=1}^n x_i^{k_i} \right) \otimes x_m, $$
with $x_0=1$. The Hopf algebra is graded by $\deg(x_n)=n$ and connected $\cH_{\deg=0}=\K$.
The counit is given by $\epsilon(x_n)=\delta_{n,0}$ and the antipode is determined inductively for graded connected
Hopf algebras. 

\smallskip

It is known that the Fa\`a di Bruno Hopf algebra embeds in the Connes--Kreimer Hopf algebra of planar rooted trees, \cite{FraMan},
$\cH_{{\rm FdB}}\hookrightarrow \cH_{CK}$ and dually the affine group schemes map surjectively $G_{CK}\twoheadrightarrow G^{{\rm diff}}$.

\smallskip

While we will not consider this question in the present paper, it is worth pointing out that the
structure of the asymptotic expansion of the spectral action of the Robertson--Walker metrics
that we obtained here in terms of Bell polynomials suggest the presence of an interesting
Hopf algebra action, similar to the one regulating the renormalization of quantum field 
theories, see \cite{CoMa-book}. Understanding the structure and meaning of the role
of the Fa\`a di Bruno Hopf algebra in the spectral action expansion appears to be an
especially interesting question in view of a better understanding of the spectral action as
a gravity model. Indeed, one usually considers the spectral action functional as an effective
field theory (at energies around or below unification, as indicated by the resulting models
of gravity coupled to matter) and treats it semiclassically using the leading terms of the 
asymptotic expansion as a classical action functional for (modified) gravity. 
The full spectral action expansion provides a series of higher-derivatives correction terms,
which are known to improve renormalizability, 
In particular the role of the full spectral action expansion in renormalizability 
in the case of Yang--Mills models was studied in \cite{WvS1} and for general
almost-commutative geometries in \cite{WvS2}, \cite{WvS3}. The use of the full
expansion of the spectral action functional is crucial in these renormalizability arguments,
see \cite{WvS3}.  The description of the coefficients of the spectral action in terms of
Brownian bridge integrals appears especially suitable for analyzing the spectral action
as a quantum theory and their expression in terms of Bell polynomials, with the
Fa\`a di Bruno Hopf algebra action suggests what the underlying Hopf-algebraic
renormalization structure should be. We hope to return to this question in future work.

\medskip
\section{Multifractal Robertson--Walker cosmologies}\label{MultiSec}

In this second part of the paper we turn to consider the multifractal cases
of Robertson--Walker cosmologies, where the spatial sections are obtained 
as an arrangement of $3$-spheres such as an Apollonian packing,
generalizing the static cases considered in \cite{BaMa}. 

\medskip
\subsection{Packed Swiss Cheese Cosmologies}

The hypothesis of multifractal structures in cosmology was proposed to
justify the observable distribution of clusters of galaxies, see for instance
\cite{Syla}. A particularly interesting model exhibiting fractality and
multifractality is known as the ``packed swiss cheese cosmology", \cite{MuDy}.
These are constructed on the model of the fractality of an Apollonian 
packing of spheres inside a $4$-dimensional spacetime. 

\smallskip

In \cite{BaMa} it was shown that a spectral action model of gravity can be
applied to these fractal cosmologies. Under some regularity assumptions 
on the structure of the fractal, the spectral action can be computed, using
a general method of \cite{CIL}, \cite{CIS} for constructing spectral triples
(the noncommutative analogs of spin geometry) on fractals. 
Unlike the case of an ordinary smooth manifold, in the presence of
a fractal structure the heat kernel of the Dirac operator acquires 
some log periodic terms. These correspond to the presence of
poles off the real line in the zeta function of the Dirac operators.
In turn these poles determine additional terms in the asymptotic
expansion of the spectral action. These are corrections to
the action functional of gravity that detect the presence of fractality.
In particular, the shape of a slow--roll potential for inflation derived 
from the spectral action model was shown in \cite{BaMa} to be 
also affected by these terms, so that corrections due to the
presence of fractality also appear in the slow--roll coefficients,
which in principle are detectable via observational data. 
The regularity assumptions on the fractal geometry used in \cite{BaMa}
were aimed at obtaining sufficiently good analytic properties of the
corresponding zeta functions, in the sense of \cite{LaFra}. 

\smallskip

However, the model of spectral action on multifractal swiss-cheese type
cosmologies considered in \cite{BaMa} is not entirely realistic,
because the relevant spacetime is assumed to be a product
of a fractal packing of spheres (or of spherical manifolds) times
a compactified Euclidean time dimension $S^1$, so that, in
particular, the scaling factor of the spatial sections remains constant.
This corresponds to a static, rather than a more realistic expanding universe.

\smallskip

In order to make the model more realistic and physically interesting,
our goal in the present paper is 
to reformulate the multifractal spectral action model 
in terms of Robertson--Walker metrics.

\smallskip

Although this may at first look like a simple modification, in fact
it requires a completely different set of analytical tools to derive
the spectral action computation from heat kernel and zeta function
information. In particular, the analytic techniques involved are based
on the derivation of the terms of the asymptotic expansion of the
spectral action for Robertson--Walker metrics discussed in the 
previous sections, based on the Feynman-Kac formula and
Brownian bridge integrals as in \cite{CC-RW}. 
In the case of Robertson--Walker metrics that are round $4$-spheres,
the result we obtain can also be obtained using the technique
of \cite{BaMa}, based on the results of \cite{LeVe}
counting the contributions of the different levels in the
fractal structure, and on results on the heat kernel on fractals, 
\cite{Du}, \cite{EIS}, \cite{EZ}. 

\medskip
\subsection{Dirac operator decomposition}\label{DiracFractal}

We consider here a Robertson--Walker geometry on a spacetime 
of the form $\R \times \cP$, where the spatial sections, instead of
being a single $3$-sphere, form an Apollonian packing of $3$-spheres,
as in a Packed Swiss Cheese Cosmology. 

\smallskip

More precisely, here $\cP$ is a packing of $3$-dimensional spheres 
with radii $\{ a_{n,k}: n \in \N, k=1, \dots, 6 \cdot 5^{n-1} \}$, where, in an 
iterative construction of the packing, at each stage $n \in \N,$ a number of 
spheres equal to $6 \cdot 5^{n-1}$ is added to the packing. We denote 
these spheres by $S_{a_{n,k}}$, see \cite{BaMa} and \cite{GraLa} for a
detailed description of the Apollonian packings of higher dimensional
spheres and their iterative construction.

\smallskip

In a Robertson--Walker metric on a spacetime with an Apollonian packing
of spheres in the spatial sections, we can assume 
that each sphere in the packing inflates at time $t$ with the same 
rate $a(t)$. We consider two possible rescalings of the Robertson--Walker
metric by effect of the scaling radii $a_{n,k}$ of the packing.
\begin{enumerate}
\item The first choice is to rescale the whole $4$-dimensional spacetime, that
is, to consider a metric of the form
\begin{equation}\label{RWmetricank2}
ds_{n,k}^2 = a_{n, k}^2 \, (dt^2 + a(t)^2 \,  d \sigma^2),  \qquad n \in \N, k =1, \dots, 6 \cdot 5^{n-1}.
\end{equation}
\item The second choice is to rescale only the spatial sections, that is,
to consider a metric of the form
\begin{equation} \label{RWmetricank}
ds_{n,k}^2 = dt^2 + a(t)^2 \, a_{n, k}^2 \, d \sigma^2,  \qquad n \in \N, k =1, \dots, 6 \cdot 5^{n-1}.
\end{equation}
\end{enumerate}
In both cases we write $D_{n, k}$ for the resulting Dirac operators  
on $\R \times S_{a_{n,k}}$ for the metric  \eqref{RWmetricank2} 
or \eqref{RWmetricank}. 

\smallskip

We then encode the geometry of the inflating sphere packing $\cP$ in the Dirac operator of
a spectral triple associated to the entire (fractal) space $\R \times \cP$. The main advantage of
the spectral triples formalism of noncommutative geometry, \cite{CoS3}, is the fact that
it makes it possible to adapt fundamental properties of Riemannian geometry to spaces
that are not smooth manifolds, including fractals. 

\smallskip

In our case, we follow the construction of spectral triples on fractals obtained in \cite{CIL}, \cite{CIS}.
For a space of the form $\R \times \cP$, with $\cP$ an Apollonian packing of $3$-spheres, the
spectral triple we consider is as in \cite{BaMa}, with $(\cA,\cH,D)$ with
$\cA$ an involutive subalgebra of $C_0(\R \times \mathcal{P})$ of functions $f$ with
bounded commutator $[D,\pi(f)]$, where $\pi$ is the representation of the algebra as
multiplication operators on the Hilbert space $\cH=\oplus_{n,k} \cH_{n,k}$ with
$\cH_{n,k}=L^2(S_{a_{n,k}},\bS)$ the spinor spaces of the individual spheres in the
packing and with Dirac operator $D=D_{\R \times \mathcal{P}}$ of the form
\begin{equation} \label{DiracOpsdirectsum}
D_{\R \times \mathcal{P}} :=  \bigoplus_{n \in \N} \bigoplus_{k=1}^{6 \cdot 5^{n-1} } D_{n, k},
\end{equation}
with the $D_{a_{n,k}}$ the Dirac operators on the individual spaces $\R \times S_{a_{n,k}}$
with the Robertson--Walker metric \eqref{RWmetricank}. 

\smallskip

We will discuss both choices \eqref{RWmetricank2} and \eqref{RWmetricank} in \S \ref{MFsec}.

\medskip
\subsection{Mellin transform and zeta functions}\label{MellinSec}

We first recall the relation between the terms in asymptotic expansion of a function and the poles 
of its Mellin transform, \cite{Flajolet}. Given a meromorphic function $\phi(z)$ with set of poles $\cS\subset \C$
and its Laurent series expansion at a pole $z_0\in \cS$,
$$ \phi(z) = \sum_{-N\leq k} c_k (z-z_0)^k, $$
the singular element $S(\phi,z_0)$ of $\phi$ at $z_0$ is the 
projection onto the polar part of the
Laurent expansion at $z_0$,
$$ S(\phi,z_0):=  \sum_{-N\leq k \leq 0} c_k (z-z_0)^k . $$
The singular expansion of $\phi$ is the formal sum of all the singular elements of $\phi$
at all poles in $\cS$,
$$ S_\phi(z):= \sum_{z\in \cS} S(\phi,z). $$
We write $\phi(z) \asymp S_\phi(z)$ to denote the singular expansion.
For example, the singular expansion of the Gamma function is
$$ \Gamma(z) \asymp \sum_{k\geq 0} \frac{(-1)^k}{k!} \frac{1}{z+k}. $$
Then the relation between the asymptotic expansion at $u\to 0$ of a function
$f(u)$ and the singular expansion of its Mellin transform $\phi(z)=\cM(f)(z)$ is
as follows. The small time asymptotic expansion is of the form
$$ f(u) \sim_{u\to 0^+} \sum_{\alpha \in \cS, k_\alpha} c_{\alpha,k_\alpha} u^\alpha \log(u)^{k_\alpha}, $$
where the coefficients $c_{\alpha,k_\alpha}$ are determined by the singular expansion of the Mellin transform,
$$ \cM(f)(z) \asymp S_{\cM(f)}(z) = \sum_{\alpha \in \cS, k_\alpha} c_{\alpha,k_\alpha}  \frac{(-1)^{k_\alpha} k_\alpha !}{(s+\alpha)^{k_\alpha +1}}, $$
where the index $k_\alpha$ ranges over the terms in the singular element of $\phi(z)=\cM(f)(z)$ at $z=\alpha$, 
up to the order of pole at $\alpha$. A similar expression holds for the asymptotic 
expansion at $u\to \infty$, see \cite{Flajolet} and the appendix to \cite{Zagier}. In the case where there
are no logarithmic terms in the asymptotic expansion, 
$$ f(u) \sim_{u\to 0^+} \sum_{\alpha \in \cS} c_{\alpha} u^\alpha, $$
the Mellin transform $\cM(f)(z)$ has analytic continuation to a meromorphic function on $\C$ with
simple poles at $z=-\alpha$ with residue $c_\alpha$.

\medskip
\subsection{Packings of $4$-spheres}

Before discussing the general structure of the spectral action on Packed Swiss Cheese Cosmologies
based on Robertson--Walker metrics, which we will be discussing in \S \ref{MFsec},
we begin by discussing explicitly a special case: the choice \eqref{RWmetricank2} of the scaled metrics 
in the special setting where the underlying Robertson--Walker metric is a round
$4$-sphere. Thus, we are considering a packing of $4$-spheres,  
where one scales the entire $4$-sphere over each $3$-sphere in the packing and not just
the spatial directions. The resulting spacetimes then have a different scaling of the time
coordinate in each sphere of the Apollonian packing. This case is much simpler
than the general case, because one can directly see the corrections to the
spectral action due to fractality simply by looking at the zeta function of the
Dirac operator as in the cases considered in \cite{BaMa}. Thus, it serves as a good model
case to identify the expected correction terms we will encounter in the general case.

\smallskip

\begin{lem}\label{zetaS4}
In the case of a $4$-dimensional sphere $S^4_r$ of radius $r>0$, 
the zeta function of the Dirac operator is given by
$$ \zeta_D(s)=\Tr(|D_{S^4_r}|^{-s}) =\sum_{\ell,\pm} {\rm m}_{\ell,\pm} |\lambda_{\ell,\pm}|^{-s}=  
\frac{4}{3} \, r^s\, (\zeta(s-3) - \zeta(s-1)), $$
with $\zeta(s)$ the Riemann zeta function. 
\end{lem}

\proof
The spectrum of the Dirac operator on a round $(D-1)$-dimensional sphere $S^{D-1}_r$
of radius $r$ is 
\begin{eqnarray}
\Spec( D_{S^{D-1}_r}) &=&  \bigg \{ \displaystyle \lambda_{\ell,\pm}= \pm r^{-1} \bigg (\frac{D-1}{2}+\ell \bigg)  \mid l \in \mathbb{Z}_{+} \bigg \},
\end{eqnarray}
with multiplicities 
\begin{eqnarray}
{\rm m}_{\ell,\pm} &=& \displaystyle 2^{[\frac{D-1}{2}]}{ \ell+D \choose \ell}.
\end{eqnarray}
As shown in \cite{CC-RW}, this can also be obtained
by regarding the unit $4$-dimensional sphere as a Robertson--Walker metric with $a(t)=\sin(t)$.
In the case of a $4$-sphere of radius $r$ this gives
$$ \zeta_D(s)=\Tr(|D_{S^4_r}|^{-s}) = \sum_{n \geq 0}{\mu(n)\Tr(f(H_n))} $$
with $f(x)=x^{-s/2}$ and
where $H_n=H_n^+\oplus H_n^-$ with
$$ H_n^\pm =- r^{-2}  \left(\frac{d^2}{dt^2}-\frac{(n+\frac{3}{2})^2}{a^2} \pm \frac{(n+\frac{3}{2})a'}{a^2} \right), $$
with multiplicity $\mu(n)=2(n+1)(n+2)$ and with $a(t)=\sin(t)$ with $t\in [0,\pi]$. Here we are scaling both
the factor $a(t)$ and the $t$ direction by the same factor $r$, since the whole $S^4$ is rescaled. 
Indeed one then has
$$ \sum_{n \geq 0}{\mu(n)\Tr(f(H_n))} =\sum_{n\geq 0} \mu(n) \sum_{k\geq n+2} (r^{-2} k^2)^{-s/2} $$
$$ = 4 r^s \sum_{k\geq n+2} \frac{k^2-k}{3} k^{-s} =\frac{4}{3} r^s (\zeta(s-3)-\zeta(s-1)). $$
\endproof

Let $\zeta_{\cL}(s) =\sum_{n,k} a_{n,k}^s$ denote the fractal string zeta function \cite{LaFra} of an
Apollonian packing $\cP$ of $3$-spheres with sequence of radii $\cL=\{ a_{n,k} \}$.

\begin{cor}\label{zetapackS4}
For a packing $\cP$ of $3$-spheres with sequence of radii $\cL=\{ a_{n,k} \}$ and 
the collection of $4$-spheres obtained from it by making each $3$-sphere the equator inside a fixed
hyperplane of a corresponding $4$-sphere, we obtain a Dirac operator $\cD_\cP$ (of the general
form discussed in \S \ref{DiracFractal}) with zeta function $\zeta_{\cD_\cP}(s)=\zeta_{\cL}(s) \zeta_{D_{S^4}}(s)$,
given by  the product of the fractal string zeta function of the sphere packing and
the zeta function of the Dirac operator on the unit $4$-sphere. 
\end{cor}

\proof This follows directly from the previous lemma with the form \eqref{DiracOpsdirectsum}
of the Dirac operator, which gives
$$ \zeta_{\cD_\cP}(s)=\Tr(|\cD_\cP|^{-s})=\sum_{n,k} \frac{4}{3} a_{n,k}^s (\zeta(s-3)-\zeta(s-1)) 
= \zeta_{\cL}(s) \zeta_{D_{S^4}}(s). $$
\endproof

This shows that, for a spacetime geometry constructed in this way the leading terms
in the spectral action expansion have the following form.

\begin{lem}\label{SpS4pack}
The leading terms in the expansion of the spectral action for the Dirac operator $\cD_\cP$
of the $4$-sphere packings described above have the form
\begin{equation}\label{leadSPS4pack}
 \Tr(f(\cD_\cP/\Lambda)) \sim f(0)\zeta_{\cD_\cP}(0) + f_2 \Lambda^2 \frac{\zeta_{\cL}(2)}{2} 
+ f_4 \Lambda^4 \frac{\zeta_{\cL}(4)}{2}
+ \sum_{\sigma \in \cS(\cL)} f_\sigma \Lambda^\sigma
\frac{\zeta_{D_{S^4}}(\sigma)}{2} \cR_\sigma, 
\end{equation}
where $\cS(\cL)$ is the set of poles of the fractal string zeta function $\zeta_{\cL}(s)$ of the sequence
of radii of the packing and $\cR_\sigma={\rm Res}_{s=\sigma}\zeta_{\cL}(s)$ are the
corresponding residues, and the coefficients $f_\alpha$ are the momenta of the test function $f$,
$$ f_\alpha=\int_0^\infty  f(\nu) \nu^{\alpha-1} d\nu. $$
\end{lem}

\proof
The result follows using the Mellin transform relation between the zeta function of the
Dirac operator and the heat-kernel of the Dirac Laplacian 
\begin{equation}\label{heatzetaMellin}
 \Tr(|\cD_\cP|^{-s}) = \frac{1}{\Gamma(s/2)} \int_0^\infty \Tr(e^{-t \cD_\cP^2})\, t^{s/2-1} dt 
\end{equation} 
give the terms in the expansion of the spectral action of the form \eqref{leadSPS4pack},
through the relation between heat kernel coefficients and residues of the zeta function
via Mellin transform, as in \S \ref{MellinSec}. 
\endproof

The terms $\zeta_{\cL}(2)$ and $\zeta_{\cL}(4)$ replace the radii $r^2$ and $r^4$ of
the corresponding terms in the spectral action on a single sphere $S^4_r$ of radius $r$
of Lemma~\ref{zetaS4} with the zeta regularizations of the sums 
$\sum_{n,k} a_{n,k}^2$ and $\sum_{n,k} a_{n,k}^4$.
Since the packing dimension of a $3$-dimensional sphere packing 
is smaller than four but larger than three, by the estimate in \cite{BaMa}, the sum
$\zeta_{\cL}(4)=\sum_{n,k} a_{n,k}^4$ is an actual convergent sum while 
$\zeta_{\cL}(2)$ is a zeta regularized value. However, for more general packings of $4$-spheres,
not obtained from a packing of $3$-spheres, the packing dimension may be larger than $4$, hence
$\zeta_{\cL}(4)$ would also be a regularized value. 

\smallskip

The multifractal nature of the sphere packing is reflected in the fact that the
zeta function $\zeta_{\cL}(s)$ has poles off the real line. 

\smallskip

\smallskip

In \cite{BaMa} the packing $\cP$ of $3$-spheres was assumed to satisfy certain 
strong {\em analyticity} assumptions (listed in \S 3.3 of \cite{BaMa}), requiring that
the fractal string zeta function $\zeta_{\cL}(s)=\sum_{n,k} a_{n,k}^s$ 
of the sequence of radii of the packing would have analytic continuation
to a meromorphic function on a region of the complex plane that 
contains the non-negative real axis; the analytic continuation
would have only one pole on the non-negative real axis, located at the
packing dimension of $\cP$; all the poles of $\zeta_{\cL}(s)$ would be simple;
and $\cL$ would have a good approximation by a family $\cL_n$ of self-similar
fractal strings with the lattice property (see \cite{LaFra}) so that the
complex poles of $\zeta_{\cL}(s)$ are approximated by the poles of
$\zeta_{\cL_n}(s)$. More precisely, this approximation property means
that, for all $\epsilon>0$ there is an
$n\in \N$ and $R=R(\epsilon,n)>0$ such that
within a vertical region of size at most $R$ the complex poles of
$\zeta_\cL(s)$ are within distance at most $\epsilon$ from the poles
of $\zeta_{\cL_n}(s)$. 
Under these assumptions, 
the following result is obtained by applying directly the results of \cite{BaMa} to the packing of $4$-spheres.

\begin{prop}\label{oscterm}
Under the analyticity assumption of \cite{BaMa} on the sphere packing, the term
$$  \sum_{\sigma \in \cS(\cL)} f_\sigma \Lambda^\sigma 
\frac{\zeta_{D_{S^4}}(\sigma)}{2} \cR_\sigma $$
consists of a leading real term
$$ \Lambda^\sigma \, \frac{4 f_\sigma}{3} (\zeta(\sigma-3)-\zeta(\sigma-1)) \cR_\sigma $$
with $\sigma=\sigma(\cP) \in \R_+$ the packing dimension of $\cP$, 
and an oscillatory term $\cS_\cP^{{\rm osc}}(\Lambda)$ involving the contribution 
of the poles of $\zeta_\cL(s)$ off the real line. Its truncation $\cS_\cP^{{\rm osc}}(\Lambda)_{\leq R}$ 
given by counting only the poles in a strip of vertical width $R$ satisfies
$$ \cS_\cP^{{\rm osc}}(\Lambda)_{\leq R} \sim \sum_{j=0}^{N_r} \Lambda^{\sigma_{n,j}} \phi_{\sigma_{n,j}}(\theta_n(\Lambda)), $$
where $s_{n,j} = \sigma_{n,j} + i (\alpha_{n,j}+ \frac{2\pi m}{\log b_n})$, for $j=0,\ldots,N_n$
are the poles of the approximating $\zeta_{\cL_n}(s)$ in the same strip, $\theta_n=\frac{\log \Lambda}{\log b_n}$
and $\phi_{\sigma_{n,j}}(\theta_n) = \sum_m f_{s_{n,j}} e^{2\pi i m \theta_n}$ with
$$ f_{s_{n,j}} = \frac{4 f_\sigma}{3} (\zeta(s_{n,j}-3)-\zeta(s_{n,j}-1)) 
\cR_{s_{n,j}} \int_0^\infty f(u) u^{s_{n,j}-1} du. $$
\end{prop}

In the following subsection we look at a simpler and more explicit lower dimensional example
based on a special case of Apollonian circle packings, the Ford circles. This provides an example
where the correction terms to the spectral action due to fractality can be computed completely
explicitly, although the packing in this case does not satisfy the analyticity assumption since it 
does not have a good approximation by self-similar fractal strings with the lattice property.

\medskip
\subsection{Lower dimensional example: circle packings}

It is useful to consider first a simpler lower dimensional example
where, instead of a $4$-dimensional spacetime $\R\times \cP$,
with $\cP$ an Apollonian packing of $3$-spheres, one considers
the case of a $2$-dimensional spacetime $\R\times \cC$ where
$\cC$ is an Apollonian packing of circles. The reasons for considering
this example, although it is not directly of physical relevance, is that
it simplifies two important features of the more general case we will
be analyzing in the following: the sequence of the radii of the packing
can be described more explicitly, especially in the more interesting
cases with number theoretic structure, \cite{GraLa2}; moreover, the Dirac
spectrum for the Dirac operator on the circle is simpler than the
Dirac spectrum on higher dimensional spheres.

\smallskip
\subsubsection{Ford circles} 

We consider here in particular the lower dimensional example of
a $2$-dimensional spacetime $\R\times \cC$ where the Apollonian
circle packing $\cC$ is given by the
Ford circles. These are circles tangent to the real line
at points $(k/n,0)$ with centers at the points $(k/n,1/(2n^2))$.
The advantage of this case is that the sequence of radii is 
known and given by a simple explicit expression. This example
will also be helpful in showing that the condition mentioned
above on the existence of a good approximation of $\cL$ and the
poles of $\zeta_{\cL}(s)$ by a family $\cL_n$ of self-similar
fractal strings and the poles of their zeta functions $\zeta_{\cL_n}(s)$
is in fact a very delicate property and even
very simple and apparently very regular examples of 
Apollonian packings need not satisfy it. 

\smallskip

\begin{lem}\label{zetaDirichlet}
The fractal string zeta function of the
Apollonian packing of Ford circles is given by
\begin{equation}\label{FordZeta2}
 \zeta_{\cL}(s) = 2^{-s} \,\, \frac{\zeta(2s-1)}{\zeta(2s)}.
\end{equation}
\end{lem}

\proof
In the case of Ford circles, the number of circles of radius
$r_n=(2n^2)^{-1}$ is equal to the number of integers $1\leq k \leq n$ that
are coprime to $n$, $\gcd\{ k, n \}=1$. This means that the multiplicity 
$m(r_n)$ is given by the value of the Euler totient function
$$ m(r_n)=\varphi(n),  $$
where the Euler totient function is equivalently given by
$$ \varphi(n) = n \prod_{p|n} (1-\frac{1}{p} ), $$
with product over the distinct prime numbers dividing $n$. 
Thus, the fractal string zeta function in the sense of \cite{LaFra} of the
Apollonian packing of Ford circles is essentially the Dirichlet series
generating function of the Euler totient function,
\begin{equation}\label{FordZeta}
\zeta_{\cL}(s) = \sum_{n\geq 1} \varphi(n) \, (2n^2)^{-s} =
2^{-s} \sum_{n\geq 1} \varphi(n) \, n^{-2s} =2^{-s} \cD_\varphi(2s).
\end{equation}
The Dirichlet series generating function of the Euler totient function
\begin{equation}\label{TotientD}
\cD_\varphi(s) = \sum_{n\geq 1} \frac{\varphi(n)}{n^s}
\end{equation}
can be computed using the fact that for a prime power $p^k$ the
totient function satisfies $\varphi(p^k)=p^k-p^{k-1}$ so that
$$ 1+\sum_k \varphi(p^k) p^{-sk}=\frac{1-p^{-s}}{1-p^{1-s}} $$
which then gives, using the Euler product formula, 
$$ \cD_\varphi(s) = \frac{\zeta(s-1)}{\zeta(s)}, $$
where $\zeta(s)$ is the Riemann zeta function. 
Thus, the zeta function of the Ford circles packing is given by \eqref{FordZeta2}.
\endproof

\smallskip

\begin{lem}\label{noLn}
The Ford circles packing does not satisfy the approximation
condition by a family $\cL_n$ of self-similar
fractal strings with the lattice property.
\end{lem}

\proof
The set $\cS(\cL)$ of poles of $\zeta_{\cL}(s)$ consists of three
subsets $\cS(\cL)=\cS_1\cup \cS_2 \cup \cS_3$ where
$\cS_1=\{ s=1 \}$, the point where the function $\zeta(2s-1)$ has a pole, 
$\cS_2=\{ s=-k\,:\, k\in \N\}$, the points that are trivial zeros of
$\zeta(2s)$, and $\cS_3=\{ \rho \in \C\smallsetminus \R_-\,:\, \zeta(2\rho)=0 \}$ 
consisting of all the non-trivial zeros of the Riemann zeta function $\zeta(2s)$. 
Assuming that the Riemann hypothesis holds, the poles of $\zeta_{\cL}(s)$ in $\cS_3$
are all on a single vertical line of real part $1/4$. 
Thus, in the case of the Apollonian packing given by Ford circles, 
the question of whether poles of  $\zeta_{\cL}(s)$ that lie off the real line 
have a good approximation by self-similar fractal strings, is in fact the question
of whether the nontrivial Riemann zeros admit such an approximation.
It is known (see Theorem~11.1 of \cite{LaFra}) that the non-trivial Riemann
zeros do not contain any infinite arithmetic progression. The 
possibility of a finite arithmetic progression within a certain vertical strip 
is limited (see Theorem~11.5 of \cite{LaFra}) by an estimate of the following
form: if $\zeta(a+inb)=0$ for some $a\in (0,1)$ and $b>0$ and for all integers 
with $0< |n|< \Lambda$ then 
$$ \Lambda < 60 \log b \left( \frac{b}{2\pi}\right)^{\frac{1}{a}-1} $$
and $\Lambda < 13 b$ when $a=1/2$, see Theorem~11.5 of \cite{LaFra}. 
Also the possibility of having an infinite sequence of non-trivial Riemann zeros 
approximated by an arithmetic progression, namely having $\zeta(a+inb)\to 0$
as $|n|\to \infty$ for some $a\in (0,1)$ and $b>0$ is ruled out (see Theorem~11.16
of \cite{LaFra}). Thus, the explicit example of Ford circles provides a simple
case where one can see that the approximation condition by self-similar
fractal strings is very difficult to satisfy even for very regular packings. 
\endproof

\smallskip

Nonetheless, in the case of the Ford circles, one can explicitly see 
the corrections to the spectral action due to the fractality of the Apollonian
packing. 

\smallskip

Since our main focus here is on $4$-dimensional spacetime geometries,
we can also construct a $4$-dimensional example using the $1$-dimensional
Apollonian circle packing by Ford circles, by increasing dimension to a collection
of $2$-spheres with the Ford circles as equators in a given hyperplane, and
then considering these $2$-spheres as equators of a collection of $3$-spheres
and similarly pass to a collection of $4$-spheres. 

\begin{prop}\label{FordS4pack}
For a packing of $4$-spheres obtained from the Ford packing of circles as above,
the leading terms of the spectral action expansion are of the form
$$ \Tr(f(\cD_\cP/\Lambda)) \sim \frac{11}{140}\, f(0) + \frac{f_1}{\pi^2} \Lambda + \frac{45\, \zeta(3)}{4\pi^4} \, f_2 \Lambda^2 + \frac{4725\, \zeta(7)}{16 \pi^8} \, f_4 \Lambda^4  $$
$$ + \sum_{k\in \N} \frac{2^{k+1} f_{-k}}{3}\, \frac{\zeta(-k-3)-\zeta(-k-1)}{\zeta(-2k-1)}  \Lambda^{-k} $$
$$ + \sum_{\sigma=a+ib} 
\frac{2^{-a}\cos(b \log 2)}{3} \Re(Z_\sigma)\, r(f)_\sigma\, \cos(b \log\Lambda)  \Lambda^a, $$
where $\sigma$ ranges over the nontrivial zeros of $\zeta(2s)$, with 
$r(f)_\sigma=\int_0^\infty f(u) u^{a-1} \cos(bu) \,du$ and  $Z_\sigma =(\zeta(\sigma-3)-\zeta(\sigma-1))\zeta(2\sigma-1)$.
\end{prop}

\proof
The packing is given by a collection of
$4$-spheres of radii the $r_n=(2n^2)^{-1}$ as the Ford circles and with
the same multiplicities given by the Euler totient function $m(r_n)=\varphi(n)$.
Thus, we can apply the expansion obtained previously for an
arrangement of $4$-spheres with a given sequence of radii and 
multiplicities and we obtain
$$ \Tr(f(\cD_\cP/\Lambda)) \sim f(0) \zeta_{\cD_\cP}(0) + f_2 \Lambda^2 \frac{\zeta_{\cL}(2)}{2}
+ f_4 \Lambda^4 \frac{\zeta_{\cL}(4)}{2} +\sum_{\sigma \in \cS(\cL)} f_\sigma \Lambda^\sigma
\frac{\zeta_{D_{S^4}}(\sigma)}{2} \cR_\sigma $$
where here we have
$$ f(0) \zeta_{\cD_\cP}(0) = f(0) \zeta_{D_{S^4}}(0) \, \zeta_{\cL}(0) =
\frac{4 f(0)}{3}\, (\zeta(-3)-\zeta(-1)) \, \frac{\zeta(-1)}{\zeta(0)} = \frac{11}{140}\, f(0), $$
with $\zeta(0)=-1/2$, $\zeta(-1)=-1/12$, and $\zeta(-3)=1/120$, and 
$$  f_2 \Lambda^2 \frac{\zeta_{\cL}(2)}{2} = \frac{45\, \zeta(3)}{4\pi^4} \, f_2 \Lambda^2, $$
$$ f_4 \Lambda^4 \frac{\zeta_{\cL}(4)}{2} = \frac{4725\, \zeta(7)}{16 \pi^8} \, f_4 \Lambda^4.  $$
The sequence of additional terms corresponding to poles of the zeta function
$\zeta_\cL(s)$ can be subdivided into the contributions of the three sets $\cS_i$
described above so that we have a term coming from $\cS_1=\{ \sigma=1 \}$ with
$$ f_1 \Lambda \frac{\zeta_{D_{S^4}}(1)}{2} \cR_1 = \frac{1}{\pi^2} \, f_1 \Lambda, $$
then a series of contributions from the poles in $\cS_2$, at the even negative integers,
$$ f_{-k} \Lambda^{-k} \frac{\zeta_{D_{S^4}}(-k)}{2} \cR_{-k} = \frac{2^{k+1}}{3} \frac{\zeta(-k-3)-\zeta(-k-1)}{\zeta(-2k-1)} f_{-k} \Lambda^{-k}, $$
and a series of contributions from the set $\cS_3$ that involves the non-trivial zeros of the
Riemann zeta function, of the form
$$ f_\sigma \Lambda^\sigma \frac{\zeta_{D_{S^4}}(\sigma)}{2} \cR_\sigma =  f_\sigma \Lambda^\sigma
\frac{2^{-\sigma +1}}{3} (\zeta(\sigma-3)-\zeta(\sigma-1))\zeta(2\sigma-1), $$
at points $\sigma\in \C \smallsetminus \R_-$ with $\zeta(2\sigma)=0$. We can write these equivalently
as a series of terms
$$  \frac{2^{-a}\cos(b \log 2)}{3} \Re(Z_\sigma)\, r(f)_\sigma\, \Lambda^a \cos(b \log\Lambda), $$
 where $\sigma =a+ib$ (with $a=1/4$ under the assumption that the Riemann hypothesis holds) and
 with $r(f)_\sigma=\int_0^\infty f(u) u^{a-1} \cos(bu) \,du$ and with
 $Z_\sigma =(\zeta(\sigma-3)-\zeta(\sigma-1))\zeta(2\sigma-1)$. 
  \endproof
 
 Thus, assuming the
 Riemann hypothesis holds, the correction terms due to the presence of fractality introduce
 a term of order $\Lambda$ in the energy scale and a term of order $\Lambda^{1/4}$, where
 the latter occurs together with a series of log periodic terms $\cos(b \log \Lambda)$ with
 $b$ the imaginary parts of the nontrivial Riemann zeros. 

 \smallskip
 
 \begin{rem}\label{irrational}{\rm 
 Observe also that the usual cosmological term and Einstein--Hilbert term now no longer have
 rational coefficients as in the case of the ordinary Robertson--Walker metrics with a single $S^3$
 as spatial section. An effect of the fractality introduced by the sphere packing is the zeta regularization
 of the sphere radius powers in these terms, which introduce non-rational coefficients like $\zeta(3)$, $\zeta(7)$,
 and powers of $\pi$.  In relation to the results of \cite{FFM3}, \cite{FM}, observe that the 
 coefficients in this example are no longer rational numbers but are still periods of mixed Tate
 motives. One can ask the question, for more general sphere packings with associated zeta
 function $\zeta_{\cL}(z)$, of whether the argument given in \cite{FM} can be modified to
 obtain a motivic description of the coefficients of the spectral action expansion and what
 conditions on the fractal string $\cL=\{ a_{n,k} \}$ of the packing radii will give rise to mixed Tate
 periods.  }
 \end{rem}
 
 \medskip
\section{Feynman--Kac formula on sphere packings} \label{MFsec}

In this section we consider the general case of a Packed Swiss Cheese Cosmology
on $\R\times \cP$, where $\cP$ is an Apollonian packing of $3$-spheres $S^3$,
with radii sequence $\cL=\{ a_{n,k} \}$, where each $\R\times S^3_{n,k}$ is endowed
with a scaled Robertson--Walker metric of either the form 
\eqref{RWmetricank2} or \eqref{RWmetricank}, for a given underlying
Robertson--Walker metric $dt^2 + a(t)^2 d\sigma^2$.  We use the full expansion
of the heat kernel for the underlying Robertson--Walker metric, obtained in the
previous sections using the Brownian bridge and the Feynman--Kac formula,
and an analysis of the effect of the scaling by the radii $a_{n,k}$ to derive via
a Mellin transform argument the full heat kernel expansion for the Dirac operator
on the Packed Swiss Cheese Cosmology $\R\times \cP$.

\smallskip

We first consider the case of $\R\times \cP$ with the scaled Robertson--Walker metrics
of the form $a_{n,k}^2 (dt^2 + a(t)^2 d\sigma^2)$, as in \eqref{RWmetricank2}, which
we refer to as the ``round scaling" case.
We compute the Feynman--Kac formula for the entire sphere packing using 
a Mellin transform with respect to the $s$ variables of the heat kernel, together with
the results on the asymptotic expansion for the underlying Robertson--Walker metric 
to obtain the full heat kernel expansion for the Packed Swiss Cheese Cosmology.

\smallskip

As we have seen in the simpler examples of the previous section,
one finds two series of terms, one that corresponds to the
expansion of the underlying Robertson--Walker metric, with
zeta regularized coefficients, and one additional series that 
corresponds to the poles of the zeta function $\zeta_{\cL}(z)$
of the fractal string of the packing.

\smallskip

We then consider the case of $\R\times \cP$ with the scaled Robertson--Walker metrics
of the form $dt^2+a_{n,k}^2 \cdot a(t)^2 \, d\sigma^2$ as in \eqref{RWmetricank}, or the
``non-round scaling" case.
We illustrate in this case a different argument based on 
the Mellin transform of the function $f_s(x)$ with
respect to the ``multiplicity" variable $x$ and we interpret the integral
$\int_\R f_s(x)\, dx$ as a special value of a combination of Mellin
transforms. This shows the occurrence of zeta regularized sums 
over the radii in the resulting Feynman--Kac formula. We also explain
how one obtains the contributions of the poles of the zeta function
$\zeta_\cL(z)$ to the asymptotic expansion of the spectral action in this case.

\smallskip
\subsection{Zeta regularized series and Mellin transform}\label{zetasSer}

We consider the Packed Swiss Cheese Cosmology $\R \times \cP$ with 
radii $\cL=\{ a_{n,k} \}$ and with the
Robertson--Walker metrics of the form $a_{n,k}^2 (dt^2 + a(t)^2 d\sigma^2)$, 
as in \eqref{RWmetricank2}. We present 
a method for computing the asymptotic expansion of the heat kernel
based on the Mellin transform with respect to the variable $\tau$ of the
heat-kernel expansion $\exp(-\tau^2 \cD^2)$.

\smallskip

We consider a slightly more general form of the series considered in \S 4 of \cite{Zagier}.
In particular, we consider the case of a function $f(\tau)$ with small time asymptotic expansion 
\begin{equation}\label{ftauN}
 f(\tau) \sim \sum_N c_N \tau^N 
\end{equation}
and we consider an associated series of the form
\begin{equation}\label{gRtau}
 g_R(\tau) =\sum_n f(r_n \tau), 
\end{equation} 
where $R=\{ r_n \}$ is an assigned sequence of $r_n\in \R^*_+$ with the property that 
the zeta function $\zeta_R(z)=\sum_n r_n^{-z}$ converges for $\Re(z)>C$ for some $C>0$,
and has an analytic continuation to a meromorphic function in $\C$ for which $z=-N$ for $N\in \N$ 
are regular points. We also assume that $\zeta_R(z)$ has only simple poles and that the
poles of $\zeta_R$ are regular values of the Mellin transform $\cM(f)(z)$. 

\begin{prop}\label{expandgR}
Let $R=\{ r_n \}$ be a sequence as above with $f(\tau)$ and $g_R(\tau)$ 
as in \eqref{ftauN} and \eqref{gRtau}. Then the small time asymptotic expansion of $g_R(\tau)$
is given by
\begin{equation}\label{gRexpand}
g_R(\tau)\sim_{\tau \to 0+} \,\,\sum_N c_N\, \zeta_R(-N)\, \tau^N
+ \sum_{\sigma \in \cS(\zeta_R)} \cR_{R,\sigma} \, \cM(f)(\sigma)\, \tau^{-\sigma} ,
\end{equation}
where $\cS(\zeta_R)$ is the set of poles of $\zeta_R(z)$
with residues $\cR_{R,\sigma}:={\rm Res}_{z=\sigma}\zeta_R(z)$.
\end{prop}

\proof
We can formally write for the series $g_R$ the expansion
\begin{equation}\label{gRtauzeta}
 g_R(\tau) \sim \sum_{N,n} c_N r_n^N \tau^N =\sum_N \zeta_R(-N) \tau^N. 
\end{equation}
One observes then, as in \cite{Zagier}, that if $\cM(f)(z)$ is the Mellin transform
$$ \cM(f)(z) =\int_0^\infty f(\tau) \tau^{z-1} d\tau $$
then the Mellin transform $\cM(g)(z)$ is given by 
\begin{equation}\label{MellinfgR}
 \cM(g)(z)=\zeta_R(z)\cdot \cM(f)(z). 
\end{equation} 
This means that we can obtain the asymptotic expansion of $g_R(\tau)$ when $\tau \to 0+$
by analyzing the singular expansion of the Mellin transform $\cM(g_R)(z)$, as recalled in \S \ref{MellinSec}.
In this case we obtain the singular expansion
$$ S_{\cM(g_R)}(z) = \sum_{\sigma \in \cS(\zeta_R)} \frac{\cR_{R,\sigma} \, \cM(f)(\sigma)}{z-\sigma} 
+ \sum_{\sigma \in \cS(\cM(f))} \frac{\zeta_R(\sigma) c_\sigma}{z-\sigma},$$
where $\cS(\zeta_R)$ is the set of poles of $\zeta_R(z)$, which are assumed to be simple,
with residue $\cR_{R,\sigma}$, and where
$\cS(\cM(f))$ is the set of poles of $\cM(f)$ with singular expansion
$$ S_{\cM(f)}(z) = \sum_{\sigma \in \cS(\cM(f))} \frac{c_\sigma}{z-\sigma}. $$
Since we are assuming that $f(\tau)$ has small time asymptotic expansion \eqref{ftauN}, the relation
to the singular expansion of the Mellin transform gives
$$ S_{\cM(f)}(z) =\sum_N \frac{c_N}{z+N}, $$
hence we obtain the result.
\endproof

\smallskip
\subsection{The Feynman--Kac formula and asymptotic expansion}\label{FKfractal}

We use the method described in Proposition~\ref{expandgR} together with the
Feynman--Kac formula, to obtain the asymptotic expansion of the trace of
the heat-kernel for the square of the Dirac operator on the sphere packing $\cP$
endowed with the Robertson--Walker metric as in \eqref{RWmetricank2}.

\smallskip

\begin{thm}\label{expandD2P}
Let $\cD:=\cD_{\R \times \cP}=\oplus_{n,k} D_{n,k}$ be the Dirac operator on a sphere packing $\cP$, with
sequence of radii $\cL=\{ a_{n,k} \}$, with the Robertson--Walker metrics \eqref{RWmetricank2}. Assume
that the zeta function $\zeta_{\cL}(z)$ of the fractal string $\cL$ has analytic continuation to a meromorphic
function on $\C$ that is regular at the points $z\in \{ M\in \Z\,:\, M\leq 4 \}$ and only has simple poles. 
Then the heat kernel expansion is given by
\begin{equation}\label{heatkerP}
\begin{array}{c}
\Tr(\exp(-\tau^2 \cD^2))\sim \\[3mm] 
\displaystyle{ \sum_{M=0}^\infty \tau^{2M-4} \zeta_{\cL}(-2M+4) 
\int  \left( \frac{1}{2} C_{2M}^{(-3/2,0)}  +\frac{1}{4} \left(
C_{2M-2}^{(-5/2,2)}-C_{2M-2}^{(-1/2,0)}\right) \right) D[\alpha]} \\[3mm]
+ \sum_{\sigma \in \cS_\cL}  \, \tilde f(\sigma) \cdot \, {\rm Res}_{z=\sigma} \zeta_{\cL} \cdot \, \tau^{-\sigma},
\end{array}
\end{equation}
where $\cS_\cL$ is the set of poles of $\zeta_\cL$ and $\tilde f(z)=\cM(f)(z)$ is the Mellin
transform of the function $f(\tau)=\Tr(\exp(-\tau^2 D^2))$ with $D$ the Dirac operator
on $\R\times S^3$ with the Robertson--Walker metric $dt^2 +a(t)^2 d\sigma^2$.
\end{thm}

\proof
We consider again the Feynman--Kac formula, focusing first on a single sphere in
the packing $\cP$, with radius $a_{n,k}$, endowed with a Robertson--Walker metric
of the form \eqref{RWmetricank2}. This means that we consider an operator of the form 
\begin{equation}\label{Hmnk}
H_{m,n,k}=- a_{n,k}^{-2} \frac{d^2}{dt^2} +V_{m,n,k}(t)
\end{equation}
where the potential $V_{m,n,k}$ is as in \eqref{Vmnk}. The Feynman--Kac formula
then reads as
$$ e^{-\tau^2 H_{m,n,k}}(t,t) =e^{-\frac{\tau^2}{a_{n,k}^2} ( \frac{d^2}{dt^2} + a_{n,k}^2 V_{m,n,k} )}(t,t) $$ $$
=\frac{a_{n,k}}{2\sqrt{\pi}\tau} \int \exp(-\tau^2 \int_0^1 V_{m,n,k}(t+\sqrt{2}\frac{\tau}{a_{n,k}} \, \alpha(u)) du) D[\alpha] .$$
Using the same Taylor expansion method described earlier, after replacing as above the sum  
$$ \sum_m \mu(m) e^{-\tau^2 H_{m,n,k}}(t,t) $$
with multiplicities $\mu(m)$ with the integral 
$$ \int_{-\infty}^\infty f_{\tau,\, n,k}(x) \, dx $$
where
$$ f_{\tau,\, n,k}(x) = \left( x^2-\frac{1}{4} \right) e^{-x^2 a_{n,k}^{-2} U - x a_{n,k}^{-1} V}, $$
with $U$ and $V$ as in \eqref{Uvar} and \eqref{Vvar}, we obtain 
$$ \sum_m \mu(m) e^{-\tau^2 H_{m,n,k}}(t,t) = \int \frac{a_{n,k}}{\tau} \left( \frac{e^{\frac{V^2}{4U}} }{4}(-a_{n,k} U^{-1/2} + 2 a_{n,k}^3 U^{-3/2} + a_{n,k}^3 V^2 U^{-5/2} ) \right) D[\alpha] $$
$$ = \int \frac{1}{\tau} \left( \frac{e^{\frac{V^2}{4U}}}{4} 
(-a^2_{n,k} U^{-1/2} + a_{n,k}^4( 2 U^{-3/2} + V^2 U^{-5/2}) ) \right) D[\alpha] $$
with the Taylor expansion
$$ e^{\frac{V^2}{4U}} U^r V^\ell =\tau^{2(r+\ell)} \sum_{M=0}^\infty a_{n,k}^{-M-2(r+\ell)}\, C^{(r,\ell)}_M \tau^M $$
with $C^{(r,\ell)}_M$ as in \eqref{CrmM} and the resulting expansion as in \eqref{intfsexpand}, which after
scaling appropriately by the factors $a_{n,k}$ becomes
$$
\begin{array}{c}
\displaystyle{ \frac{1}{\tau} \left( \frac{e^{\frac{V^2}{4U}}}{4} 
(-a^2_{n,k} U^{-1/2} + a_{n,k}^4( 2 U^{-3/2} + V^2 U^{-5/2}) )\right)  
=  }\\[3mm]
\displaystyle{  \frac{1}{4} \sum_{M=0}^\infty  \left ( C_M^{(-5/2,2)} - C_M^{(-1/2,0)}  \right ) \zeta_{\cL}(-M+2) \tau^{M-2}
+
\frac{1}{2} \sum_{M=0}^\infty   C_M^{(-3/2,0)}  \zeta_{\cL}(-M+4)  \tau^{M-4}. }
\end{array}
$$
Thus, we can write the Feynman--Kac formula for the whole $\cP$
$$ \sum_{n,k} \sum_m \mu(m) e^{-\tau^2 H_{m,n,k}}(t,t) = $$
$$ \sum_{M=0}^\infty \tau^{2M-4} \zeta_{\cL}(-2M+4) \int \left( \frac{1}{2} C_{2M}^{(-3/2,0)}  +\frac{1}{4} (
C_{2M-2}^{(-5/2,2)}-C_{2M-2}^{(-1/2,0)})\right) D[\alpha], $$
with only the term $\frac{1}{2} C_0^{(-3/2,0)}$ when $M=0$, 
as in Theorem~\ref{fullexpansion1}. This series should be interpreted in the sense discussed in
\S \ref{zetasSer} above, as a series
$$ g_{\cL}(\tau)=\sum_{n,k} f(a_{n,k}^{-1} \tau) $$
$$ f(\tau)\sim \sum_M \tau^{2M-4} \int \left( \frac{1}{2} C_{2M}^{(-3/2,0)}  +\frac{1}{4} (
C_{2M-2}^{(-5/2,2)}-C_{2M-2}^{(-1/2,0)})\right) D[\alpha]. $$
Thus, applying the method of Proposition~\ref{expandgR}, we find that $g_{\cL}(\tau)$ has
an asymptotic expansion for $\tau\to 0+$ of the form \eqref{heatkerP}.
\endproof 

Theorem~\ref{expandD2P} determines the full expansion of the spectral action on
a multifractal Packed Swiss Cheese Cosmology $\R \times \cP$ with the Robertson--Walker 
metrics \eqref{RWmetricank2}.

\begin{cor}\label{SpActPfull}
Under the same hypotheses as Theorem~\ref{expandD2P},
the full expansion of the spectral action on $\R \times \cP$ with the Robertson--Walker 
metrics \eqref{RWmetricank2} is of the form
\begin{equation}\label{SAfullP}
\begin{array}{c}
\Tr(f(\cD/\Lambda))\sim \\[3mm]
\displaystyle{ \sum_{M=0}^\infty \Lambda^{4-2M} f_{4-2M} \, \zeta_{\cL}(-2M+4) 
\int  \left( \frac{1}{2} C_{2M}^{(-3/2,0)}  +\frac{1}{4} (
C_{2M-2}^{(-5/2,2)}-C_{2M-2}^{(-1/2,0)})\right) D[\alpha]} \\[3mm]
+ \sum_{\sigma \in \cS_\cL}  \, \tilde f(\sigma) \cdot f_{\sigma} \, \cdot {\rm Res}_{z=\sigma} \zeta_{\cL} \cdot \, \Lambda^\sigma.
\end{array}
\end{equation}
\end{cor}

\proof As in \cite{CCact}, \cite{CCuncanny}, and \cite{WvS}, we relate the coefficients of the
spectral action expansion to the coefficients of the heat kernel as in \eqref{SAexpand} and
\eqref{SAcoeffs}, by computing the
spectral action $\Tr(f(\cD/\Lambda))$ with respect to a test 
function of the form $f(x)=\int_0^\infty e^{-\tau^2 x^2} d\mu(\tau)$
for some measure $\mu$, with $\int_0^\infty d\mu(\tau)=f(0)$. 
Assuming that the full expansion of the heat kernel is of the form
$$ \Tr(e^{-\tau^2 \cD^2}) \sim \sum_\alpha \tau^{2\alpha} c_{2\alpha}, $$
we obtain from \eqref{SAexpand} and \eqref{SAcoeffs}
$$ \Tr(f(\cD/\Lambda)) \sim \sum_{\alpha<0} f_{2\alpha} c_{2\alpha} \Lambda^{-2\alpha} +
a_0\, f(0) + \sum_{\alpha>0} f_{2\alpha} c_{2\alpha} \Lambda^{-2\alpha}, $$
where for $\alpha<0$
$$ f_{2\alpha}=\int_0^\infty f(v)\, v^{-2\alpha -1} \, dv, $$
and when $\alpha = M>0$
$$ f_{2M}= \int_0^\infty \tau^{2M} d\mu(\tau) =(-1)^M f^{(2M)}(0) . $$
Thus, we obtain a series of the form
$$ \Tr(f(\cD/\Lambda)) \sim $$ $$ \sum_{M=0}^\infty \Lambda^{-2M+4} f_{-2M+4} \, \zeta_{\cL}(-2M+4) 
\int  \left( \frac{1}{2} C_{2M}^{(-3/2,0)}  +\frac{1}{4}  (
C_{2M-2}^{(-5/2,2)}-C_{2M-2}^{(-1/2,0)} )\right) D[\alpha] $$
$$ + \sum_{\sigma \in \cS_\cL}  \, \tilde f(\sigma) \cdot f_{\sigma} \, \cdot {\rm Res}_{z=\sigma} \zeta_{\cL} \cdot \, \Lambda^{\sigma}, $$
where the Mellin transform relation \eqref{heatzetaMellin} between heat kernel and
zeta function
$$ \Tr(|\cD|^{-z})= \frac{2}{\Gamma(z/2)} \int_0^\infty e^{-\tau^2 \cD^2} \tau^{z-1} d\tau $$
gives 
$$ \tilde f(z) = \cM(\Tr(e^{-\tau^2\cD^2}))(z) = \frac{\Gamma(z/2)}{2}\, \zeta_{\cD}(z). $$ 
\endproof

\smallskip
\subsection{Scaling properties for non-round scaling}

We now consider the effect of rescaling the Robertson--Walker metric $dt^2 + a(t)^2 d\sigma^2$ 
to metrics of the form $dt^2+a_{n,k}^2 \cdot a(t)^2 \, d\sigma^2$, as in \eqref{RWmetricank}.

\begin{lem}\label{scaleUVlem}
Let $U$ and $V$ be as in \eqref{Uvar}, \eqref{Vvar}, for a given Robertson--Walker metric 
of the form $ds^2=dt^2 + a(t)^2 d\sigma^2$ on $\R\times S^3$. For a rescaled metric of the form
$ds_a^2=dt^2 + a^2 \cdot a(t)^2 \, d\sigma^2$, with a constant scaling factor $a>0$, 
as in \eqref{RWmetricank}, we have
\begin{equation}\label{scaleUVcase1}
U \mapsto a^{-2}\, U , \ \ \ \ \  V \mapsto a^{-1}\, V\, .
\end{equation}
\end{lem}

\proof In the first case, 
the expressions \eqref{AandB}, \eqref{Uvar}, \eqref{Vvar} show that the
$U$ and $V$ functions defined by \eqref{Uvar} and \eqref{Vvar} change like
$U \mapsto a^{-2}\, U$ and  $V \mapsto a^{-1}\, V$, 
under the rescaling $a(t) \mapsto a\cdot a(t)$ of the scaling factor.
\endproof

\begin{rem}\label{scale2case}{\rm
Compare this with the case of 
a rescaled metric of the form $ds_a^2=a\cdot (dt^2 + a(t)^2 d\sigma^2)$ as in \eqref{RWmetricank2},
that we discussed in  \S \ref{zetasSer}, where the
the operator $d^2/dt^2$ also scales by $a^{-2}$,
hence in the Feynman--Kac formula, one modifies the term 
$$ \exp\left( -s \frac{d^2}{dt^2} \right) \mapsto \exp\left( -\frac{s}{a^2}\, \frac{d^2}{dt^2}\right) . $$
This has the effect
of scaling the variable $s\mapsto a^{-2} s$, so that the variables $U$ and $V$
are replaced by new variables $U',V'$ with 
$$ U'=a^{-4} \, s \int_0^1 A^2(t +\frac{\sqrt{2s}}{a}\, \alpha(v)) dv \ \ \ \text{ and } \ \ \ 
V'= a^{-3} \, s \int_0^1 A'(t + \frac{\sqrt{2s}}{a}\, \alpha(v)) dv. $$
with the effect of the rescaling $\sqrt{2s}/a$
on the expansion and on the resulting heat-kernel asymptotics shown in \S \ref{zetasSer} . }
\end{rem}

\begin{lem}\label{rescale1CM}
The rescaling $U \mapsto a^{-2}\, U$ and  $V \mapsto a^{-1}\, V$ gives the
rescaled expression
\begin{equation}\label{scaleCMrm1}
\frac{1}{4} \sum_{M=0}^\infty (a^3 \, C^{(-5/2,2)}_M - a\, C^{(-1/2,0)}_M)\, \tau^{M-2} +
\frac{1}{2} \sum_{M=0}^\infty a^3 \, C^{(-3/2,0)}_M\, \tau^{M-4}.
\end{equation}
\end{lem}

\proof
By Lemma~\ref{scaleUVlem}, n this case the scaling takes the form
$$ e^{V^2/4U} U^r V^m  \mapsto a^{-2r -m} \,\, e^{V^2/4U} U^r V^m, $$
which means that the coefficients $C^{r,m}_M$ are rescaled by a factor $a^{-2r -m}$
and we obtain the rescaled expression \eqref{scaleCMrm1}.
\endproof

\smallskip

Clearly the expression \eqref{scaleCMrm1} suggests that, in this case, we should expect an
asymptotic expansion with zeta regularized terms of the form
\begin{equation}\label{zetaCMrm1}
\frac{1}{4} \sum_{M=0}^\infty (\zeta_{\cL}(3) \, C^{(-5/2,2)}_M - \zeta_{\cL}(1)\, C^{(-1/2,0)}_M)\, \tau^{M-2} +
\frac{1}{2} \sum_{M=0}^\infty \zeta_{\cL}(3) \, C^{(-3/2,0)}_M\, \tau^{M-4},
\end{equation}
with the zeta function $\zeta_\cL(s)=\sum_{n,k} a_{n,k}^s$ of the sphere packing radii.  
This will be justified more precisely in the following subsections.

\smallskip
\subsection{Mellin transform and hypergeometric function}

We provide an argument for the presence of the zeta regularized terms \eqref{zetaCMrm1}
based on taking a Mellin transform with respect to the ``multiplicity variable" $x$ in the function
$f_s(x)$ of \eqref{fsx}.

\smallskip

With the notation $a^{(n)}=a(a+1)\cdots (a+n-1)$ and $a^{(0)}=1$, the
Kummer confluent hypergeometric function is defined by the series
$$ {}_1 F_1 (a,b,t)= \sum_{n=0}^\infty \frac{a^{(n)} t^n}{b^{(n)} n!} $$
and is a solution of the Kummer equation
$$ t \frac{d^2 f}{dt^2} + (b-t) \frac{df}{dt} - a f =0. $$

The function $f_s(x)$ of \eqref{fsx} has a Mellin transform that can be
computed explicitly in terms of the Kummer confluent hypergeometric function.

\begin{lem}\label{MellinHyp}
The Mellin transform in the $x$-variable of the function $$f_{s,-}(x):=f_s(x)=(x^2-\frac{1}{4}) e^{-x^2 U -x V}$$ is
given in terms of the Kummer confluent hypergeometric function $_1 F_1$ by the expression
$$ \cM ((x^2-\frac{1}{4}) e^{-x^2 U -x V})(z) = \frac{1}{8} U^{-(z+3)/2}  \times $$ $$  \Big(  U^{1/2}\,  \Gamma(\frac{z}{2}) ( -U\, {}_1 F_1 (\frac{z}{2},\frac{1}{2},\frac{V^2}{4 U})
+ 2z \, {}_1 F_1( \frac{z+2}{2}, \frac{1}{2} , \frac{V^2}{4U}))  $$
$$ + V\,  \Gamma(\frac{z+1}{2}) (U\,  {}_1 F_1 (\frac{z+1}{2}, \frac{3}{2}, \frac{V^2}{4U})) -2 (z+1)\, 
 {}_1 F_1 (\frac{z+3}{2},\frac{3}{2}, \frac{V^2}{4U}))  \Big)\, . $$
 Similarly, the Mellin transform in the $x$-variable of the function
 $$f_{s,+}(x):=(x^2-\frac{1}{4}) e^{-x^2 U +x V}$$ is also given in terms of the Kummer confluent hypergeometric function as
 $$ \cM((x^2-\frac{1}{4}) e^{-x^2 U +x V})(z) = \frac{1}{8} U^{-(z+3)/2}  \times $$ $$
   \Big(  U^{1/2}\, \Gamma(z/2) ( -U\, {}_1 F_1 (\frac{z}{2},\frac{1}{2},\frac{V^2}{4 U})
 + 2z \, {}_1 F_1( \frac{z+2}{2}, \frac{1}{2} , \frac{V^2}{4U}))  $$
 $$ + V \,  \Gamma(\frac{z+1}{2})  (- U\, {}_1 F_1( \frac{z+1}{2}, \frac{3}{2}, \frac{V^2}{4U})) + 2 (z+1)\, {}_1 F_1 (\frac{z+3}{2},\frac{3}{2}, \frac{V^2}{4U}))  \Big)\, . $$
\end{lem}

\smallskip

As discussed earlier, the real variable $x$ is a continuous variable replacing the
discrete $m+\frac{3}{2}$ in the expression for the potential $V_m(t)$ of \eqref{Vnt}.

\begin{lem}\label{xintMellin}
The multiplicity integral is a special value at $z=1$ 
\begin{equation}\label{intfszMellin}
 \begin{array}{c}
\displaystyle{ \int_{-\infty}^\infty f_s(x)\, dx } = \\[3mm]
 \displaystyle{
\left(  -\frac{1}{4}\, U^{-(1+\frac{z}{2})}\, \Gamma(\frac{z}{2}) ( U\, {}_1 F_1(\frac{z}{2},\frac{1}{2}, \frac{V^2}{4U}) -
2z\, {}_1 F_1(1+\frac{z}{2},\frac{1}{2},\frac{V^2}{4U}))  \right)|_{z=1} } \\[3mm]
\displaystyle{ = e^{\frac{V^2}{4U}} \frac{\sqrt{\pi}}{4} (-U^{-1/2} +2 U^{-3/2} + V^2 U^{-5/2}) }\, . 
\end{array}
\end{equation}
\end{lem}

\proof
We consider again the integral \eqref{intfsx}, which we write in the form
$$ \int_{-\infty}^\infty f_s(x)\, dx = \int_0^\infty f_{s,-}(x)\, dx + \int_0^\infty f_{s,+}(x)\, dx, $$
where 
$$  f_{s,\pm}(x)=(x^2-\frac{1}{4}) e^{-x^2 U \pm x V}.  $$
In turn, we can write these integrals in terms of Mellin transforms as
\begin{equation}\label{intpmMellin}
 \int_{-\infty}^\infty f_s(x)\, dx = \cM (f_{s,-})(z)|_{z=1} + \cM (f_{s,+})(z)|_{z=1} 
\end{equation} 
The Mellin transform on the right-hand-side is given by 
\begin{equation}\label{Mellinpm}
 \cM(f_{s,-})(z) + \cM(f_{s,+})(z) =-\frac{1}{4}\, U^{-1-\frac{z}{2}}\, \Gamma(\frac{z}{2}) ( U\, {}_1 F_1(\frac{z}{2},\frac{1}{2}, \frac{V^2}{4U}) -
2z\, {}_1 F_1(1+\frac{z}{2},\frac{1}{2},\frac{V^2}{4U})),  
\end{equation}
and the evaluation at $z=1$ of this expression gives back the expression we used before
\begin{equation}\label{ev1Mellinpm}
\begin{array}{c}
\displaystyle{ 
 \left( -\frac{1}{4}\, U^{-(1+\frac{z}{2})}\, \Gamma(\frac{z}{2}) ( U\, {}_1 F_1(\frac{z}{2},\frac{1}{2}, \frac{V^2}{4U}) -
2z\, {}_1 F_1(1+\frac{z}{2},\frac{1}{2},\frac{V^2}{4U})) \right)|_{z=1}} = \\[3mm] 
\displaystyle{ e^{\frac{V^2}{4U}} \frac{\sqrt{\pi}}{4} (-U^{-1/2} +2 U^{-3/2} + V^2 U^{-5/2}) }. 
\end{array}
\end{equation}
\endproof

\smallskip
\subsection{Scaling and hypergeometric functions}

To simplify some of the following expressions, we introduce the notation
\begin{equation}\label{HUVaz}
\begin{array}{rl} 
H_\lambda (\tau,z):= & \displaystyle{ U^{-z/2} \, \Gamma(z/2)\, {}_1 F_1(\frac{z}{2},\lambda,\frac{V^2}{4U})}, \\[3mm]
H(\tau,z):= & \displaystyle{ H_{1/2}(\tau,z)=U^{-z/2} \, \Gamma(z/2)\, {}_1 F_1(\frac{z}{2},\frac{1}{2},\frac{V^2}{4U}) },
\end{array}
\end{equation}
where, as above, the variables $s$ and $\tau$ are related by $s=\tau^2$.
We also introduce, for later use, the notation
\begin{equation}\label{zetaHUVaz}
H_{\cL}(\tau,z):=U^{-z/2} \,\zeta_{\cL}(z)\, \Gamma(z/2)\, {}_1 F_1(\frac{z}{2},\frac{1}{2},\frac{V^2}{4U})=
\zeta_{\cL}(z)\, H(\tau,z) .
\end{equation}

\smallskip

\begin{cor}\label{scaleMfcase1}
Consider the scaling of the $S^3$ spatial sections by a factor $a_{n,k}$
taken from the series $\cL=\{ a_{n,k} \}$ of radii of a given sphere packing, as in \eqref{RWmetricank}.
\begin{equation}\label{MellinHUVacase1}
 \begin{array}{rl} \cM(f_{s,n,k})(z)= \displaystyle{ \frac{1}{8} a_{n,k}^z } & \displaystyle{
\left( H_{\frac{3}{2}}(\tau,z+1) \, V - H_{\frac{1}{2}} (\tau, z) \right) } \\[3mm] $$
\displaystyle{  - \frac{1}{2} a_{n,k}^{z+2}} &\displaystyle{  \left( H_{\frac{3}{2}}(\tau,z+3)\, V  -  H_{\frac{1}{2}}(\tau, z+2) \right). }
\end{array} 
\end{equation}
\end{cor}

\proof Scaling the Robertson--Walker metric by $a_{n,k}$ as in  \eqref{RWmetricank},
we find that the Mellin transform of $f_s(x)$ of \eqref{fsx} satisfies
$$ \cM(f_s)(z) \mapsto \cM(f_{s,n,k})(z):=  \frac{1}{8} a_{n,k}^{(z+3)} \, U^{-(z+3)/2}  \times $$ $$  \Big( a_{n,k}^{-1} \, U^{1/2}\,  \Gamma(\frac{z}{2}) ( - a_{n,k}^{-2} \, U\, {}_1 F_1 (\frac{z}{2},\frac{1}{2},\frac{V^2}{4 U})
+ 2z \, {}_1 F_1( \frac{z+2}{2}, \frac{1}{2} , \frac{V^2}{4U}))  $$
$$ + a_{n,k}^{-1}\, V\,  \Gamma(\frac{z+1}{2}) (a_{n,k}^{-2}\, U\,  {}_1 F_1 (\frac{z+1}{2}, \frac{3}{2}, \frac{V^2}{4U})) -2 (z+1)\, 
 {}_1 F_1 (\frac{z+3}{2},\frac{3}{2}, \frac{V^2}{4U}))  \Big)  $$
 $$ = \frac{1}{8} a_{n,k}^z \, \left( U^{-(z+1)/2}\,  V\, \Gamma(\frac{z+1}{2}) \,  {}_1 F_1 (\frac{z+1}{2}, \frac{3}{2}, \frac{V^2}{4U}))  -U^{-z/2} \Gamma(\frac{z}{2}) {}_1 F_1(\frac{z}{2},\frac{1}{2},\frac{V^2}{4 U}) \right) $$
 $$ +\frac{1}{4} a_{n,k}^{z+2}\, \left( U^{-(\frac{z}{2}+1)} \Gamma(\frac{z}{2})\, z \, \, {}_1 F_1( \frac{z+2}{2}, \frac{1}{2} , \frac{V^2}{4U}))  - U^{-\frac{z+3}{2}} \, V\, \Gamma(\frac{z+1}{2}) \, (z+1) {}_1 F_1 (\frac{z+3}{2},\frac{3}{2}, \frac{V^2}{4U}) \right) . $$
 We then write the above as \eqref{MellinHUVacase1},
where we used $\Gamma(z/2) z = 2 \Gamma((z+2)/2)$ and $\Gamma((z+1)/2) (z+1) =2 \Gamma(z+3)/2)$.
 \endproof
 
 Similarly, we obtain the scaling of the integral $\int_{-\infty}^\infty f_s(x) \, dx$, viewed in
 terms of sums of Mellin transforms as above. 

\begin{cor}\label{scaleMfpmcase1}
For a scaled metric of the form  $dt^2 + a_{n,k}^2 a(t)^2 d\sigma^2$ as in \eqref{RWmetricank} 
we have
\begin{equation}\label{scaleMfpm}
\begin{array}{c} 
\cM (f_{s,n,k,-})(z) + \cM (f_{s,n,k,+})(z) =  \\[3mm]
\displaystyle{ 
 -\frac{1}{4}\, a_{n,k}^z\, U^{-\frac{z}{2}}\, \Gamma(\frac{z}{2}) {}_1 F_1(\frac{z}{2},\frac{1}{2}, \frac{V^2}{4U}) 
 + a_{n,k}^{z+2} U^{1-\frac{z}{2}}\, \Gamma(1+\frac{z}{2}) \, {}_1 F_1(1+\frac{z}{2},\frac{1}{2},\frac{V^2}{4U})) 
} \\[3mm]
\displaystyle{ -\frac{1}{4} a_{n,k}^z \, H(\tau,z) + a_{n,k}^{z+2} \, H(\tau,z+2).  }
\end{array}
\end{equation}
\end{cor}

\smallskip
\subsection{Sphere packing and Mellin transform}

We consider then the full sphere packing $\cP$ with sequence of radii $\cL=\{ a_{n,k} \}$ and with
the Robertson--Walker metrics of the form as in \eqref{RWmetricank}. The
potential $V_m(t)$, $m\in \N$, of \eqref{Vnt} is replaced by a sequence
\begin{equation}\label{Vmnk}
 V_{m,n,k}= \frac{(m+\frac{3}{2})}{a_{n,k}^2 \cdot a(t)^2} (( m+\frac{3}{2} ) - a_{n,k}\cdot a^\prime(t)). 
\end{equation} 
Each potential in this sequence corresponds to a scaled choice $a_{n,k}^{-1} A(t)$
and $a_{n,k}^{-2} B(t)$ of the variables of \eqref{AandB}, and corresponding scaled variables 
$a_{n,k}^{-2}\, U$ and $a_{n,k}^{-1}\, V$ in \eqref{Uvar} and \eqref{Vvar}, as discussed above.
We then consider a function $f_{\cP,s}$ associated to the full sphere packing $\cP$ of the form
\begin{equation}\label{fPs}
f_{\cP,s}(x) = \left( x^2 - \frac{1}{4} \right) \sum_{n,k} e^{-x^2 a_{n,k}^{-2}\, U- x a_{n,k}^{-1}\, V}.
\end{equation}
As above, we write
$$ \int_{-\infty}^\infty f_{\cP,s}(x) \, dx = \cM(f_{\cP,s,-})(z) |_{z=1} + \cM(f_{\cP,s,+})(z) |_{z=1}, $$
with $f_{\cP,s,\pm} = (x^2 -1/4) \sum_{n,k} \exp(-x^2 a_{n,k}^{-2}\, U \pm x a_{n,k}^{-1}\, V)$.

\begin{lem}\label{MellinP}
For $\R\times \cP$ with sequence of radii $\cL=\{ a_{n,k} \}$ and with
the Robertson--Walker metrics of the form as in \eqref{RWmetricank},
the Mellin transform of $f_{\cP,s,-} + f_{\cP,s,+}$ satisfies, with $s=\tau^2$,
\begin{equation}\label{MellinF1pmP}
\cM(f_{\cP,s,-})(z) + \cM(f_{\cP,s,+})(z) = -\frac{1}{4} H_\cL(\tau,z) + H_\cL(\tau,z+2) .
\end{equation}
\end{lem}

\proof This follows directly from \eqref{intfszMellin} since we have
$$ \begin{array}{rl} \cM(f_{\cP,s})(z) = & \displaystyle{ \frac{1}{8}\left( H_{\frac{3}{2}}(\tau,z+1) \, V - H_{\frac{1}{2}} (\tau,z) \right) \sum_{n,k} a_{n,k}^z } \\[3mm]
& \displaystyle{ -\frac{1}{2} \left( H_{\frac{3}{2}}(\tau,z+3)\, V  -  H_{\frac{1}{2}}(\tau,z+2) \right) \sum_{n,k} a_{n,k}^{z+2} }. \\[3mm]
= & \displaystyle{ \frac{1}{8}\left( H_{\frac{3}{2}}(\tau,z+1) \, V - H_{\frac{1}{2}} (\tau,z) \right) \zeta_{\cL}(z) }
\\[3mm] & \displaystyle{ -\frac{1}{2} \left( H_{\frac{3}{2}}(\tau,z+3)\, V  -  H_{\frac{1}{2}}(\tau,z+2) \right) 
\zeta_{\cL}(z+2)  }, 
\end{array} $$
and 
$$ \cM(f_{\cP,s,-})(z) + \cM(f_{\cP,s,+})(z) = - \frac{1}{4} (\sum_{n,k} a_{n,k}^z) \, H(\tau,z) + (\sum_{n,k} a_{n,k}^{z+2}) \, H(\tau,z+2), $$
which gives \eqref{MellinF1pmP}.
\endproof

\smallskip

This shows that, indeed, we obtain the zeta regularized coefficients $\zeta_\cL(3)$ and $\zeta_{\cL}(1)$ as in
\eqref{zetaCMrm1}, when we evaluate at $z=1$ the expression \eqref{MellinF1pmP}. This argument, however,
does not suffice to identify all the modified terms in the asymptotic expansion of the spectral action due to
the fractality of the sphere packing, as we expect also in this case to see contributions from the poles of
the zeta function $\zeta_\cL(z)$. We cannot apply the same argument used for the round scaling here,
since the fact that the $\tau$ variable is not rescaled prevents us from applying the same argument
of \S \ref{zetasSer} based on \S 4 of \cite{Zagier}. However, we show in the next subsection that one
can still extract the information on the contribution of the poles of $\zeta_\cL(z)$ from a further analysis 
of these Mellin transforms.

\smallskip
\subsection{Pole contributions}

The discussion above shows why one obtains zeta regularized coefficients
in the spectral action expansion in the case of the non-round scaling \eqref{RWmetricank} 
of the Robertson--Walker metrics. However, it does not explain why in the asymptotic 
expansion for $\tau\to 0$ of the heat kernel one should also
find contributions associated to the poles of the zeta function $\zeta_{\cL}(z)$, as in
the case of the round scaling discussed before. 
In fact, one can see that such terms will occur in this case too,
when one applies the $\tau$ expansion \eqref{UVtausum} 
$$ U = \tau^2 \sum_{n=0}^\infty \frac{u_n}{n!} \,\tau^n , \ \ \  V = \tau^2 \sum_{n=0}^\infty \frac{v_n}{n!} \,\tau^n $$
to the function  $H_\cL(\tau,z)=\zeta_\cL(z) \Gamma(z/2) U^{-z/2} \, {}_1 F_1(z/2,1/2,V^2/4U)$
as in \eqref{zetaHUVaz}, equivalently written as 
$$ H_\cL(\tau,z)= \zeta_\cL(z) \Gamma(z/2) U^{-z/2} \, \sum_{n=0}^\infty \frac{(z/2)_n}{4^n \, n!\, (1/2)_n} \,
V^{2n} U^{-n}, $$
by expanding the confluent hypergeometric function, with $(a)_n=a(a+1)\cdots (a+n-1)$ denoting the rising
factorial. One sees from \eqref{UVtausum} that the term $U^{-z/2}$ will contribute a term 
with $\tau^z$ times a power series in $\tau$, while the confluent hypergeometric function will
contribute power series in $\tau$. 

\smallskip

Rather than giving a complete computation, we simply explain here why the presence of the
term with $\tau^z$ will generate the pole contributions, by illustrating the same
phenomenon in a simplified case. 

\smallskip

The product of the Mellin transforms $\cM(f_1)(z)\cdot \cM(f_2)(z)$ corresponds to the
transform of the Mellin convolution
$$ \cM(f_1)(z)\cdot \cM(f_2)(z) = \cM(f_1 \star f_2)(z), $$
$$ (f_1 \star f_2)(x) = \int_0^\infty f_1(\frac{x}{u}) f_2(u) \, \frac{du}{u}. $$
Moreover, Mellin transform can be applied to distributions \cite{Kang}, \cite{Misra}, by considering the
space of test functions $\cD(\R_+)$ and the space $\cQ=\cM(\cD(\R_+))$ of 
their Mellin transforms. Let $\cD'_+$ and $\cQ'$ denote the dual spaces. 
Denoting by $\langle\cdot,\cdot\rangle_{\cM}$ the duality
pairing between $\cQ'$ and $\cQ$ and by $\langle\cdot,\cdot\rangle$ the duality pairing
between $\cD'_+$ and $\cD(\R_+)$, for a distribution $\Lambda\in \cD'_+$ one defines
$\cM(\Lambda)$ by the property that 
$\langle \cM(\Lambda), \cM(\phi) \rangle =\langle \Lambda,\phi \rangle$.
Then the Mellin transform of a distribution in $\cD'_+$ belongs to the space $\cQ'$, which is
a space of analytic functions.  In particular, the Mellin transform of a delta distribution is
given by
\begin{equation}\label{Mellindelta}
\tau^{z-1} = \cM( \delta(x-\tau) ).
\end{equation}
Similarly, we can write as Mellin transform of a distribution 
\begin{equation}\label{zetaMellindelta}
\tau^z \, \zeta_{\cL}(z) = \cM(  \sum_{n,k} \tau\, a_{n,k}\,\, \delta(x-\tau \cdot a_{n,k}) ).
\end{equation}
This is the distribution acting as
$$ \langle \sum_{n,k} \tau\, a_{n,k}\,\, \delta(x-\tau \cdot a_{n,k}), \phi(x) \rangle = 
\sum_{n,k} \tau\, a_{n,k}\,\, \phi(\tau\, a_{n,k}). $$
We write this distribution as
$$ \Lambda_{\cP,\tau} :=\sum_{n,k} \tau\, a_{n,k}\,\, \delta(x-\tau \cdot a_{n,k}). $$

\smallskip

Consider then a given function $g(x)$. In particular, for our application we should
think of the function $g_\gamma(x) :=\cM^{-1}(\Gamma(z/2)\, {}_1 F_1(z/2,1/2,\gamma))$. 
The product of Mellin transforms is then
$$ \cM(\Lambda_{\cP,\tau})(z) \cdot \cM(g)(z) = \cM (\Lambda_{\cP,\tau} \star g)(z) $$
$$ =\cM( \sum_{n,k} \tau a_{n,k}\, \int_0^\infty \delta(u- \tau a_{n,k}) \, g(\frac{x}{u})\, \frac{du}{u} )
 =\sum_{n,k}  \cM( g(\frac{x}{\tau \cdot a_{n,k}}) ). $$
 We then let $h_z(\tau) := \cM(g(\frac{x}{\tau})$ and we write the above as
 $$ L_z(\tau):= \sum_{n,k} h_z( \tau \cdot a_{n,k}). $$
 One can then obtain the asymptotic expansion for this function by using the same
 technique that we used in the case of the round scaling, by considering now the
 variable $z$ fixed (it will be evaluated at $z=1$ in the end) and taking a Mellin
 transform with respect to the variable $\tau$. To avoid confusing notation we
 write $\cM_\tau$ for the Mellin transform taken with respect to the variable $\tau$,
 and we write this Mellin transform as a function of a complex variable $\beta$.
 Arguing as in the round scaling case, we have
 $$ \cM_\tau(L_z(\tau))(\beta) = \zeta_{\cL}(\beta)\cdot \cM(h_z(\tau))(\beta). $$
It then follows that the terms in the asymptotic expansion for $\tau \to 0$ of
$L_z(\tau)$ are determined by the terms in the singular expansion of the
Mellin transform $\zeta_{\cL}(\beta)\cdot \cM(h_z(\tau))(\beta)$. These contain
a series of terms that correspond to the poles $\sigma \in \cS_\cL$ of the
zeta function $\zeta_{\cL}(\beta)$ with coefficient given by the product of the
residue $\cR_\sigma={\rm Res}_{\beta=\sigma} \zeta_{\cL}(\beta)$ and
the value $\cM(h_z(\tau))(\sigma)$ with a power $\tau^{-\sigma}$, as well as
terms that correspond to the poles of $\cM(h_z(\tau))(\beta)$ (which are
terms of the asymptotic expansion of $h_z(\tau)$. When we apply this
argument to $g_\gamma(x) :=\cM^{-1}(\Gamma(z/2)\, {}_1 F_1(z/2,1/2,\gamma))$,
the resulting asymptotic expansion then needs to be modified by 
replacing $\gamma = V^2/4U$ and expanding $U$ and $V$ in powers of
$\tau$ according to \eqref{UVtausum}. This makes writing out in full
the explicit computation lengthy, but it does not change the fundamental
structure of the expansion, which will still have a series of terms arising from the
poles of $\zeta_{\cL}(\beta)$. Thus, even without carrying out
a full explicit computation of all these terms, we then see that the structure of
the asymptotic expansion of the heat kernel (hence of the spectral action)
is similar to the case of the round scaling, with a series of terms generated
by the poles of $\zeta_{\cL}(\beta)$ and a series of terms coming from the
asymptotic expansion of the underlying (unscaled) Robertson--Walker metric,
appearing with zeta regularized coefficients as in \eqref{zetaCMrm1}.

\medskip
\section*{Appendix A}

We include here the explicit expression in terms of Dawson functions for the integrals
\eqref{intleadingtoDawson} in the case of the $4$-dimensional simplex. We have
\[
\int_{\Delta^4} \exp \left ( -\frac{1}{2} \sum_{j, m =1}^4 c_{j,m} u_j u_m \right ) \, dv_1 \, dv_2 \, dv_3 \, dv_4=
\]
\begin{center}
\begin{math}
\frac{32 F\left(\frac{u_2}{2 \sqrt{2}}\right) \sqrt{2}}{u_1 \left(u_1+u_2\right) \left(u_2+u_3\right) \left(u_1+u_2+u_3\right) \left(u_3+u_4\right) \left(u_2+u_3+u_4\right) \left(u_1+u_2+u_3+u_4\right)}-\frac{32 F\left(\frac{u_1+u_2}{2 \sqrt{2}}\right) \sqrt{2}}{u_1 \left(u_1+u_2\right) \left(u_2+u_3\right) \left(u_1+u_2+u_3\right) \left(u_3+u_4\right) \left(u_2+u_3+u_4\right) \left(u_1+u_2+u_3+u_4\right)}-\frac{32 F\left(\frac{u_2+u_3}{2 \sqrt{2}}\right) \sqrt{2}}{u_1 \left(u_1+u_2\right) \left(u_2+u_3\right) \left(u_1+u_2+u_3\right) \left(u_3+u_4\right) \left(u_2+u_3+u_4\right) \left(u_1+u_2+u_3+u_4\right)}+\frac{32 F\left(\frac{u_1+u_2+u_3}{2 \sqrt{2}}\right) \sqrt{2}}{u_1 \left(u_1+u_2\right) \left(u_2+u_3\right) \left(u_1+u_2+u_3\right) \left(u_3+u_4\right) \left(u_2+u_3+u_4\right) \left(u_1+u_2+u_3+u_4\right)}+\frac{32 F\left(\frac{u_2}{2 \sqrt{2}}\right) \sqrt{2}}{u_2 \left(u_1+u_2\right) \left(u_2+u_3\right) \left(u_1+u_2+u_3\right) \left(u_3+u_4\right) \left(u_2+u_3+u_4\right) \left(u_1+u_2+u_3+u_4\right)}+\frac{32 F\left(\frac{u_3}{2 \sqrt{2}}\right) \sqrt{2}}{u_2 \left(u_1+u_2\right) \left(u_2+u_3\right) \left(u_1+u_2+u_3\right) \left(u_3+u_4\right) \left(u_2+u_3+u_4\right) \left(u_1+u_2+u_3+u_4\right)}-\frac{32 F\left(\frac{u_2+u_3}{2 \sqrt{2}}\right) \sqrt{2}}{u_2 \left(u_1+u_2\right) \left(u_2+u_3\right) \left(u_1+u_2+u_3\right) \left(u_3+u_4\right) \left(u_2+u_3+u_4\right) \left(u_1+u_2+u_3+u_4\right)}+\frac{16 F\left(\frac{u_1}{2 \sqrt{2}}\right) u_3 \sqrt{2}}{u_1 u_2 \left(u_1+u_2\right) \left(u_2+u_3\right) \left(u_1+u_2+u_3\right) \left(u_3+u_4\right) \left(u_2+u_3+u_4\right) \left(u_1+u_2+u_3+u_4\right)}+\frac{16 F\left(\frac{u_2}{2 \sqrt{2}}\right) u_3 \sqrt{2}}{u_1 u_2 \left(u_1+u_2\right) \left(u_2+u_3\right) \left(u_1+u_2+u_3\right) \left(u_3+u_4\right) \left(u_2+u_3+u_4\right) \left(u_1+u_2+u_3+u_4\right)}-\frac{16 F\left(\frac{u_1+u_2}{2 \sqrt{2}}\right) u_3 \sqrt{2}}{u_1 u_2 \left(u_1+u_2\right) \left(u_2+u_3\right) \left(u_1+u_2+u_3\right) \left(u_3+u_4\right) \left(u_2+u_3+u_4\right) \left(u_1+u_2+u_3+u_4\right)}+\frac{16 F\left(\frac{u_1}{2 \sqrt{2}}\right) u_4 \sqrt{2}}{u_1 u_2 \left(u_1+u_2\right) \left(u_2+u_3\right) \left(u_1+u_2+u_3\right) \left(u_3+u_4\right) \left(u_2+u_3+u_4\right) \left(u_1+u_2+u_3+u_4\right)}+\frac{16 F\left(\frac{u_2}{2 \sqrt{2}}\right) u_4 \sqrt{2}}{u_1 u_2 \left(u_1+u_2\right) \left(u_2+u_3\right) \left(u_1+u_2+u_3\right) \left(u_3+u_4\right) \left(u_2+u_3+u_4\right) \left(u_1+u_2+u_3+u_4\right)}-\frac{16 F\left(\frac{u_1+u_2}{2 \sqrt{2}}\right) u_4 \sqrt{2}}{u_1 u_2 \left(u_1+u_2\right) \left(u_2+u_3\right) \left(u_1+u_2+u_3\right) \left(u_3+u_4\right) \left(u_2+u_3+u_4\right) \left(u_1+u_2+u_3+u_4\right)}+\frac{32 F\left(\frac{u_2}{2 \sqrt{2}}\right) \sqrt{2}}{\left(u_1+u_2\right) u_3 \left(u_2+u_3\right) \left(u_1+u_2+u_3\right) \left(u_3+u_4\right) \left(u_2+u_3+u_4\right) \left(u_1+u_2+u_3+u_4\right)}+\frac{32 F\left(\frac{u_3}{2 \sqrt{2}}\right) \sqrt{2}}{\left(u_1+u_2\right) u_3 \left(u_2+u_3\right) \left(u_1+u_2+u_3\right) \left(u_3+u_4\right) \left(u_2+u_3+u_4\right) \left(u_1+u_2+u_3+u_4\right)}-\frac{32 F\left(\frac{u_2+u_3}{2 \sqrt{2}}\right) \sqrt{2}}{\left(u_1+u_2\right) u_3 \left(u_2+u_3\right) \left(u_1+u_2+u_3\right) \left(u_3+u_4\right) \left(u_2+u_3+u_4\right) \left(u_1+u_2+u_3+u_4\right)}+\frac{16 F\left(\frac{u_2}{2 \sqrt{2}}\right) u_2 \sqrt{2}}{u_1 \left(u_1+u_2\right) u_3 \left(u_2+u_3\right) \left(u_1+u_2+u_3\right) \left(u_3+u_4\right) \left(u_2+u_3+u_4\right) \left(u_1+u_2+u_3+u_4\right)}-\frac{16 F\left(\frac{u_1+u_2}{2 \sqrt{2}}\right) u_2 \sqrt{2}}{u_1 \left(u_1+u_2\right) u_3 \left(u_2+u_3\right) \left(u_1+u_2+u_3\right) \left(u_3+u_4\right) \left(u_2+u_3+u_4\right) \left(u_1+u_2+u_3+u_4\right)}-\frac{16 F\left(\frac{u_2+u_3}{2 \sqrt{2}}\right) u_2 \sqrt{2}}{u_1 \left(u_1+u_2\right) u_3 \left(u_2+u_3\right) \left(u_1+u_2+u_3\right) \left(u_3+u_4\right) \left(u_2+u_3+u_4\right) \left(u_1+u_2+u_3+u_4\right)}+\frac{16 F\left(\frac{u_1+u_2+u_3}{2 \sqrt{2}}\right) u_2 \sqrt{2}}{u_1 \left(u_1+u_2\right) u_3 \left(u_2+u_3\right) \left(u_1+u_2+u_3\right) \left(u_3+u_4\right) \left(u_2+u_3+u_4\right) \left(u_1+u_2+u_3+u_4\right)}+\frac{16 F\left(\frac{u_2}{2 \sqrt{2}}\right) u_4 \sqrt{2}}{u_1 \left(u_1+u_2\right) u_3 \left(u_2+u_3\right) \left(u_1+u_2+u_3\right) \left(u_3+u_4\right) \left(u_2+u_3+u_4\right) \left(u_1+u_2+u_3+u_4\right)}-\frac{16 F\left(\frac{u_1+u_2}{2 \sqrt{2}}\right) u_4 \sqrt{2}}{u_1 \left(u_1+u_2\right) u_3 \left(u_2+u_3\right) \left(u_1+u_2+u_3\right) \left(u_3+u_4\right) \left(u_2+u_3+u_4\right) \left(u_1+u_2+u_3+u_4\right)}-\frac{16 F\left(\frac{u_2+u_3}{2 \sqrt{2}}\right) u_4 \sqrt{2}}{u_1 \left(u_1+u_2\right) u_3 \left(u_2+u_3\right) \left(u_1+u_2+u_3\right) \left(u_3+u_4\right) \left(u_2+u_3+u_4\right) \left(u_1+u_2+u_3+u_4\right)}+\frac{16 F\left(\frac{u_1+u_2+u_3}{2 \sqrt{2}}\right) u_4 \sqrt{2}}{u_1 \left(u_1+u_2\right) u_3 \left(u_2+u_3\right) \left(u_1+u_2+u_3\right) \left(u_3+u_4\right) \left(u_2+u_3+u_4\right) \left(u_1+u_2+u_3+u_4\right)}+\frac{16 F\left(\frac{u_2}{2 \sqrt{2}}\right) u_1 \sqrt{2}}{u_2 \left(u_1+u_2\right) u_3 \left(u_2+u_3\right) \left(u_1+u_2+u_3\right) \left(u_3+u_4\right) \left(u_2+u_3+u_4\right) \left(u_1+u_2+u_3+u_4\right)}+\frac{16 F\left(\frac{u_3}{2 \sqrt{2}}\right) u_1 \sqrt{2}}{u_2 \left(u_1+u_2\right) u_3 \left(u_2+u_3\right) \left(u_1+u_2+u_3\right) \left(u_3+u_4\right) \left(u_2+u_3+u_4\right) \left(u_1+u_2+u_3+u_4\right)}-\frac{16 F\left(\frac{u_2+u_3}{2 \sqrt{2}}\right) u_1 \sqrt{2}}{u_2 \left(u_1+u_2\right) u_3 \left(u_2+u_3\right) \left(u_1+u_2+u_3\right) \left(u_3+u_4\right) \left(u_2+u_3+u_4\right) \left(u_1+u_2+u_3+u_4\right)}+\frac{16 F\left(\frac{u_2}{2 \sqrt{2}}\right) u_4 \sqrt{2}}{u_2 \left(u_1+u_2\right) u_3 \left(u_2+u_3\right) \left(u_1+u_2+u_3\right) \left(u_3+u_4\right) \left(u_2+u_3+u_4\right) \left(u_1+u_2+u_3+u_4\right)}+\frac{16 F\left(\frac{u_3}{2 \sqrt{2}}\right) u_4 \sqrt{2}}{u_2 \left(u_1+u_2\right) u_3 \left(u_2+u_3\right) \left(u_1+u_2+u_3\right) \left(u_3+u_4\right) \left(u_2+u_3+u_4\right) \left(u_1+u_2+u_3+u_4\right)}-\frac{16 F\left(\frac{u_2+u_3}{2 \sqrt{2}}\right) u_4 \sqrt{2}}{u_2 \left(u_1+u_2\right) u_3 \left(u_2+u_3\right) \left(u_1+u_2+u_3\right) \left(u_3+u_4\right) \left(u_2+u_3+u_4\right) \left(u_1+u_2+u_3+u_4\right)}+\frac{32 F\left(\frac{u_3}{2 \sqrt{2}}\right) \sqrt{2}}{\left(u_1+u_2\right) \left(u_2+u_3\right) \left(u_1+u_2+u_3\right) u_4 \left(u_3+u_4\right) \left(u_2+u_3+u_4\right) \left(u_1+u_2+u_3+u_4\right)}-\frac{32 F\left(\frac{u_2+u_3}{2 \sqrt{2}}\right) \sqrt{2}}{\left(u_1+u_2\right) \left(u_2+u_3\right) \left(u_1+u_2+u_3\right) u_4 \left(u_3+u_4\right) \left(u_2+u_3+u_4\right) \left(u_1+u_2+u_3+u_4\right)}-\frac{32 F\left(\frac{u_3+u_4}{2 \sqrt{2}}\right) \sqrt{2}}{\left(u_1+u_2\right) \left(u_2+u_3\right) \left(u_1+u_2+u_3\right) u_4 \left(u_3+u_4\right) \left(u_2+u_3+u_4\right) \left(u_1+u_2+u_3+u_4\right)}+\frac{32 F\left(\frac{u_2+u_3+u_4}{2 \sqrt{2}}\right) \sqrt{2}}{\left(u_1+u_2\right) \left(u_2+u_3\right) \left(u_1+u_2+u_3\right) u_4 \left(u_3+u_4\right) \left(u_2+u_3+u_4\right) \left(u_1+u_2+u_3+u_4\right)}-\frac{16 F\left(\frac{u_2+u_3}{2 \sqrt{2}}\right) u_2 \sqrt{2}}{u_1 \left(u_1+u_2\right) \left(u_2+u_3\right) \left(u_1+u_2+u_3\right) u_4 \left(u_3+u_4\right) \left(u_2+u_3+u_4\right) \left(u_1+u_2+u_3+u_4\right)}+\frac{16 F\left(\frac{u_1+u_2+u_3}{2 \sqrt{2}}\right) u_2 \sqrt{2}}{u_1 \left(u_1+u_2\right) \left(u_2+u_3\right) \left(u_1+u_2+u_3\right) u_4 \left(u_3+u_4\right) \left(u_2+u_3+u_4\right) \left(u_1+u_2+u_3+u_4\right)}+\frac{16 F\left(\frac{u_2+u_3+u_4}{2 \sqrt{2}}\right) u_2 \sqrt{2}}{u_1 \left(u_1+u_2\right) \left(u_2+u_3\right) \left(u_1+u_2+u_3\right) u_4 \left(u_3+u_4\right) \left(u_2+u_3+u_4\right) \left(u_1+u_2+u_3+u_4\right)}-\frac{16 F\left(\frac{u_1+u_2+u_3+u_4}{2 \sqrt{2}}\right) u_2 \sqrt{2}}{u_1 \left(u_1+u_2\right) \left(u_2+u_3\right) \left(u_1+u_2+u_3\right) u_4 \left(u_3+u_4\right) \left(u_2+u_3+u_4\right) \left(u_1+u_2+u_3+u_4\right)}-\frac{16 F\left(\frac{u_2+u_3}{2 \sqrt{2}}\right) u_3 \sqrt{2}}{u_1 \left(u_1+u_2\right) \left(u_2+u_3\right) \left(u_1+u_2+u_3\right) u_4 \left(u_3+u_4\right) \left(u_2+u_3+u_4\right) \left(u_1+u_2+u_3+u_4\right)}+\frac{16 F\left(\frac{u_1+u_2+u_3}{2 \sqrt{2}}\right) u_3 \sqrt{2}}{u_1 \left(u_1+u_2\right) \left(u_2+u_3\right) \left(u_1+u_2+u_3\right) u_4 \left(u_3+u_4\right) \left(u_2+u_3+u_4\right) \left(u_1+u_2+u_3+u_4\right)}+\frac{16 F\left(\frac{u_2+u_3+u_4}{2 \sqrt{2}}\right) u_3 \sqrt{2}}{u_1 \left(u_1+u_2\right) \left(u_2+u_3\right) \left(u_1+u_2+u_3\right) u_4 \left(u_3+u_4\right) \left(u_2+u_3+u_4\right) \left(u_1+u_2+u_3+u_4\right)}-\frac{16 F\left(\frac{u_1+u_2+u_3+u_4}{2 \sqrt{2}}\right) u_3 \sqrt{2}}{u_1 \left(u_1+u_2\right) \left(u_2+u_3\right) \left(u_1+u_2+u_3\right) u_4 \left(u_3+u_4\right) \left(u_2+u_3+u_4\right) \left(u_1+u_2+u_3+u_4\right)}+\frac{16 F\left(\frac{u_3}{2 \sqrt{2}}\right) u_1 \sqrt{2}}{u_2 \left(u_1+u_2\right) \left(u_2+u_3\right) \left(u_1+u_2+u_3\right) u_4 \left(u_3+u_4\right) \left(u_2+u_3+u_4\right) \left(u_1+u_2+u_3+u_4\right)}-\frac{16 F\left(\frac{u_2+u_3}{2 \sqrt{2}}\right) u_1 \sqrt{2}}{u_2 \left(u_1+u_2\right) \left(u_2+u_3\right) \left(u_1+u_2+u_3\right) u_4 \left(u_3+u_4\right) \left(u_2+u_3+u_4\right) \left(u_1+u_2+u_3+u_4\right)}-\frac{16 F\left(\frac{u_3+u_4}{2 \sqrt{2}}\right) u_1 \sqrt{2}}{u_2 \left(u_1+u_2\right) \left(u_2+u_3\right) \left(u_1+u_2+u_3\right) u_4 \left(u_3+u_4\right) \left(u_2+u_3+u_4\right) \left(u_1+u_2+u_3+u_4\right)}+\frac{16 F\left(\frac{u_2+u_3+u_4}{2 \sqrt{2}}\right) u_1 \sqrt{2}}{u_2 \left(u_1+u_2\right) \left(u_2+u_3\right) \left(u_1+u_2+u_3\right) u_4 \left(u_3+u_4\right) \left(u_2+u_3+u_4\right) \left(u_1+u_2+u_3+u_4\right)}+\frac{16 F\left(\frac{u_3}{2 \sqrt{2}}\right) u_3 \sqrt{2}}{u_2 \left(u_1+u_2\right) \left(u_2+u_3\right) \left(u_1+u_2+u_3\right) u_4 \left(u_3+u_4\right) \left(u_2+u_3+u_4\right) \left(u_1+u_2+u_3+u_4\right)}-\frac{16 F\left(\frac{u_2+u_3}{2 \sqrt{2}}\right) u_3 \sqrt{2}}{u_2 \left(u_1+u_2\right) \left(u_2+u_3\right) \left(u_1+u_2+u_3\right) u_4 \left(u_3+u_4\right) \left(u_2+u_3+u_4\right) \left(u_1+u_2+u_3+u_4\right)}-\frac{16 F\left(\frac{u_3+u_4}{2 \sqrt{2}}\right) u_3 \sqrt{2}}{u_2 \left(u_1+u_2\right) \left(u_2+u_3\right) \left(u_1+u_2+u_3\right) u_4 \left(u_3+u_4\right) \left(u_2+u_3+u_4\right) \left(u_1+u_2+u_3+u_4\right)}+\frac{16 F\left(\frac{u_2+u_3+u_4}{2 \sqrt{2}}\right) u_3 \sqrt{2}}{u_2 \left(u_1+u_2\right) \left(u_2+u_3\right) \left(u_1+u_2+u_3\right) u_4 \left(u_3+u_4\right) \left(u_2+u_3+u_4\right) \left(u_1+u_2+u_3+u_4\right)}+\frac{16 F\left(\frac{u_3}{2 \sqrt{2}}\right) u_1 \sqrt{2}}{\left(u_1+u_2\right) u_3 \left(u_2+u_3\right) \left(u_1+u_2+u_3\right) u_4 \left(u_3+u_4\right) \left(u_2+u_3+u_4\right) \left(u_1+u_2+u_3+u_4\right)}+\frac{16 F\left(\frac{u_4}{2 \sqrt{2}}\right) u_1 \sqrt{2}}{\left(u_1+u_2\right) u_3 \left(u_2+u_3\right) \left(u_1+u_2+u_3\right) u_4 \left(u_3+u_4\right) \left(u_2+u_3+u_4\right) \left(u_1+u_2+u_3+u_4\right)}-\frac{16 F\left(\frac{u_3+u_4}{2 \sqrt{2}}\right) u_1 \sqrt{2}}{\left(u_1+u_2\right) u_3 \left(u_2+u_3\right) \left(u_1+u_2+u_3\right) u_4 \left(u_3+u_4\right) \left(u_2+u_3+u_4\right) \left(u_1+u_2+u_3+u_4\right)}+\frac{16 F\left(\frac{u_3}{2 \sqrt{2}}\right) u_2 \sqrt{2}}{\left(u_1+u_2\right) u_3 \left(u_2+u_3\right) \left(u_1+u_2+u_3\right) u_4 \left(u_3+u_4\right) \left(u_2+u_3+u_4\right) \left(u_1+u_2+u_3+u_4\right)}+\frac{16 F\left(\frac{u_4}{2 \sqrt{2}}\right) u_2 \sqrt{2}}{\left(u_1+u_2\right) u_3 \left(u_2+u_3\right) \left(u_1+u_2+u_3\right) u_4 \left(u_3+u_4\right) \left(u_2+u_3+u_4\right) \left(u_1+u_2+u_3+u_4\right)}-\frac{16 F\left(\frac{u_3+u_4}{2 \sqrt{2}}\right) u_2 \sqrt{2}}{\left(u_1+u_2\right) u_3 \left(u_2+u_3\right) \left(u_1+u_2+u_3\right) u_4 \left(u_3+u_4\right) \left(u_2+u_3+u_4\right) \left(u_1+u_2+u_3+u_4\right)}. 
\end{math}
\end{center}

\medskip
\section*{Appendix B}

We report here the explicit expressions for the coefficients $a_6(t)$ and $a_8(t)$ of
the spectral action as a function of $A(t)=1/a(t)$ and $B(t)=A(t)^2$. The expressions
we obtain for $a_{2M}(t)$ for $M=0,\ldots,4$ match those computed in 
\cite{CC-RW}, \cite{FGK} when expressed in terms of the scaling factor $a(t)$. We have: 

\begin{center}
\[
a_6(t)=
\]
\begin{math}
\frac{7 A'(t)^6}{768 B(t)^{9/2}}+\frac{105 B'(t)^2 A'(t)^4}{1024 B(t)^{11/2}}-\frac{A'(t)^4}{128 B(t)^{5/2}}-\frac{35 B''(t) A'(t)^4}{768 B(t)^{9/2}}+\frac{5 A^{(3)}(t) A'(t)^3}{96 B(t)^{7/2}}-\frac{35 B'(t) A''(t) A'(t)^3}{192 B(t)^{9/2}}+\frac{1155 B'(t)^4 A'(t)^2}{4096 B(t)^{13/2}}+\frac{5 A''(t)^2 A'(t)^2}{64 B(t)^{7/2}}+\frac{21 B''(t)^2 A'(t)^2}{256 B(t)^{9/2}}+\frac{B''(t) A'(t)^2}{64 B(t)^{5/2}}+\frac{7 B'(t) B^{(3)}(t) A'(t)^2}{64 B(t)^{9/2}}-\frac{B^{(4)}(t) A'(t)^2}{64 B(t)^{7/2}}-\frac{5 B'(t)^2 A'(t)^2}{256 B(t)^{7/2}}-\frac{231 B'(t)^2 B''(t) A'(t)^2}{512 B(t)^{11/2}}+\frac{B'(t) A''(t) A'(t)}{32 B(t)^{5/2}}+\frac{77 B'(t) A''(t) B''(t) A'(t)}{192 B(t)^{9/2}}+\frac{77 B'(t)^2 A^{(3)}(t) A'(t)}{384 B(t)^{9/2}}+\frac{A^{(5)}(t) A'(t)}{80 B(t)^{5/2}}-\frac{A^{(3)}(t) A'(t)}{48 B(t)^{3/2}}-\frac{A''(t) B^{(3)}(t) A'(t)}{16 B(t)^{7/2}}-\frac{B'(t) A^{(4)}(t) A'(t)}{16 B(t)^{7/2}}-\frac{3 B''(t) A^{(3)}(t) A'(t)}{32 B(t)^{7/2}}-\frac{105 B'(t)^3 A''(t) A'(t)}{256 B(t)^{11/2}}+\frac{5005 B'(t)^6}{49152 B(t)^{15/2}}+\frac{35 B'(t)^2 A''(t)^2}{256 B(t)^{9/2}}+\frac{249 B'(t)^2 B''(t)^2}{1024 B(t)^{11/2}}+\frac{3 A^{(3)}(t)^2}{160 B(t)^{5/2}}+\frac{23 B^{(3)}(t)^2}{2688 B(t)^{7/2}}+\frac{11 B'(t)^2 B''(t)}{768 B(t)^{7/2}}+\frac{15 B'(t)^3 B^{(3)}(t)}{128 B(t)^{11/2}}+\frac{A''(t) A^{(4)}(t)}{40 B(t)^{5/2}}+\frac{19 B''(t) B^{(4)}(t)}{1344 B(t)^{7/2}}+\frac{B^{(4)}(t)}{480 B(t)^{3/2}}+\frac{3 B'(t) B^{(5)}(t)}{448 B(t)^{7/2}}-\frac{A''(t)^2}{96 B(t)^{3/2}}-\frac{B'(t) B^{(3)}(t)}{160 B(t)^{5/2}}-\frac{3 B''(t)^2}{640 B(t)^{5/2}}-\frac{B^{(6)}(t)}{1120 B(t)^{5/2}}-\frac{11 B'(t) A''(t) A^{(3)}(t)}{96 B(t)^{7/2}}-\frac{11 A''(t)^2 B''(t)}{192 B(t)^{7/2}}-\frac{43 B'(t) B''(t) B^{(3)}(t)}{384 B(t)^{9/2}}-\frac{25 B'(t)^2 B^{(4)}(t)}{768 B(t)^{9/2}}-\frac{61 B''(t)^3}{2304 B(t)^{9/2}}-\frac{35 B'(t)^4}{6144 B(t)^{9/2}}-\frac{1309 B'(t)^4 B''(t)}{4096 B(t)^{13/2}}, 
\end{math}
\end{center}

\[
a_8(t)=
\]
\begin{center}
\begin{math}
\frac{3 A'(t)^8}{4096 B(t)^{11/2}}+\frac{77 B'(t)^2 A'(t)^6}{4096 B(t)^{13/2}}-\frac{A'(t)^6}{1536 B(t)^{7/2}}-\frac{7 B''(t) A'(t)^6}{1024 B(t)^{11/2}}+\frac{7 A^{(3)}(t) A'(t)^5}{768 B(t)^{9/2}}-\frac{21 B'(t) A''(t) A'(t)^5}{512 B(t)^{11/2}}+\frac{5005 B'(t)^4 A'(t)^4}{32768 B(t)^{15/2}}+\frac{35 A''(t)^2 A'(t)^4}{1536 B(t)^{9/2}}+\frac{63 B''(t)^2 A'(t)^4}{2048 B(t)^{11/2}}+\frac{5 B''(t) A'(t)^4}{1536 B(t)^{7/2}}+\frac{21 B'(t) B^{(3)}(t) A'(t)^4}{512 B(t)^{11/2}}-\frac{7 B^{(4)}(t) A'(t)^4}{1536 B(t)^{9/2}}-\frac{35 B'(t)^2 A'(t)^4}{6144 B(t)^{9/2}}-\frac{847 B'(t)^2 B''(t) A'(t)^4}{4096 B(t)^{13/2}}+\frac{5 B'(t) A''(t) A'(t)^3}{384 B(t)^{7/2}}+\frac{77 B'(t) A''(t) B''(t) A'(t)^3}{256 B(t)^{11/2}}+\frac{77 B'(t)^2 A^{(3)}(t) A'(t)^3}{512 B(t)^{11/2}}+\frac{A^{(5)}(t) A'(t)^3}{192 B(t)^{7/2}}-\frac{A^{(3)}(t) A'(t)^3}{192 B(t)^{5/2}}-\frac{7 B''(t) A^{(3)}(t) A'(t)^3}{128 B(t)^{9/2}}-\frac{7 A''(t) B^{(3)}(t) A'(t)^3}{192 B(t)^{9/2}}-\frac{7 B'(t) A^{(4)}(t) A'(t)^3}{192 B(t)^{9/2}}-\frac{385 B'(t)^3 A''(t) A'(t)^3}{1024 B(t)^{13/2}}+\frac{25025 B'(t)^6 A'(t)^2}{65536 B(t)^{17/2}}+\frac{315 B'(t)^2 A''(t)^2 A'(t)^2}{1024 B(t)^{11/2}}+\frac{2739 B'(t)^2 B''(t)^2 A'(t)^2}{4096 B(t)^{13/2}}+\frac{3 A^{(3)}(t)^2 A'(t)^2}{128 B(t)^{7/2}}+\frac{23 B^{(3)}(t)^2 A'(t)^2}{1536 B(t)^{9/2}}+\frac{77 B'(t)^2 B''(t) A'(t)^2}{3072 B(t)^{9/2}}+\frac{165 B'(t)^3 B^{(3)}(t) A'(t)^2}{512 B(t)^{13/2}}+\frac{A''(t) A^{(4)}(t) A'(t)^2}{32 B(t)^{7/2}}+\frac{19 B''(t) B^{(4)}(t) A'(t)^2}{768 B(t)^{9/2}}+\frac{B^{(4)}(t) A'(t)^2}{640 B(t)^{5/2}}+\frac{3 B'(t) B^{(5)}(t) A'(t)^2}{256 B(t)^{9/2}}-\frac{A''(t)^2 A'(t)^2}{128 B(t)^{5/2}}-\frac{B'(t) B^{(3)}(t) A'(t)^2}{128 B(t)^{7/2}}-\frac{3 B''(t)^2 A'(t)^2}{512 B(t)^{7/2}}-\frac{B^{(6)}(t) A'(t)^2}{896 B(t)^{7/2}}-\frac{77 B'(t) A''(t) A^{(3)}(t) A'(t)^2}{384 B(t)^{9/2}}-\frac{77 A''(t)^2 B''(t) A'(t)^2}{768 B(t)^{9/2}}-\frac{129 B'(t) B''(t) B^{(3)}(t) A'(t)^2}{512 B(t)^{11/2}}-\frac{61 B''(t)^3 A'(t)^2}{1024 B(t)^{11/2}}-\frac{75 B'(t)^2 B^{(4)}(t) A'(t)^2}{1024 B(t)^{11/2}}-\frac{105 B'(t)^4 A'(t)^2}{8192 B(t)^{11/2}}-\frac{17017 B'(t)^4 B''(t) A'(t)^2}{16384 B(t)^{15/2}}+\frac{35 B'(t)^3 A''(t) A'(t)}{1536 B(t)^{9/2}}+\frac{1309 B'(t)^3 A''(t) B''(t) A'(t)}{1024 B(t)^{13/2}}+\frac{1309 B'(t)^4 A^{(3)}(t) A'(t)}{4096 B(t)^{13/2}}+\frac{11 A''(t)^2 A^{(3)}(t) A'(t)}{192 B(t)^{7/2}}+\frac{61 B''(t)^2 A^{(3)}(t) A'(t)}{768 B(t)^{9/2}}+\frac{3 B''(t) A^{(3)}(t) A'(t)}{320 B(t)^{5/2}}+\frac{A''(t) B^{(3)}(t) A'(t)}{160 B(t)^{5/2}}+\frac{43 A''(t) B''(t) B^{(3)}(t) A'(t)}{384 B(t)^{9/2}}+\frac{43 B'(t) A^{(3)}(t) B^{(3)}(t) A'(t)}{384 B(t)^{9/2}}+\frac{B'(t) A^{(4)}(t) A'(t)}{160 B(t)^{5/2}}+\frac{43 B'(t) B''(t) A^{(4)}(t) A'(t)}{384 B(t)^{9/2}}+\frac{25 B'(t) A''(t) B^{(4)}(t) A'(t)}{384 B(t)^{9/2}}+\frac{25 B'(t)^2 A^{(5)}(t) A'(t)}{768 B(t)^{9/2}}+\frac{A^{(7)}(t) A'(t)}{1120 B(t)^{5/2}}-\frac{A^{(5)}(t) A'(t)}{480 B(t)^{3/2}}-\frac{11 B'(t) A''(t) B''(t) A'(t)}{384 B(t)^{7/2}}-\frac{3 A''(t) B^{(5)}(t) A'(t)}{448 B(t)^{7/2}}-\frac{3 B'(t) A^{(6)}(t) A'(t)}{448 B(t)^{7/2}}-\frac{11 B'(t)^2 A^{(3)}(t) A'(t)}{768 B(t)^{7/2}}-\frac{23 B^{(3)}(t) A^{(4)}(t) A'(t)}{1344 B(t)^{7/2}}-\frac{19 A^{(3)}(t) B^{(4)}(t) A'(t)}{1344 B(t)^{7/2}}-\frac{19 B''(t) A^{(5)}(t) A'(t)}{1344 B(t)^{7/2}}-\frac{35 B'(t) A''(t)^3 A'(t)}{384 B(t)^{9/2}}-\frac{45 B'(t)^2 A''(t) B^{(3)}(t) A'(t)}{128 B(t)^{11/2}}-\frac{15 B'(t)^3 A^{(4)}(t) A'(t)}{128 B(t)^{11/2}}-\frac{249 B'(t) A''(t) B''(t)^2 A'(t)}{512 B(t)^{11/2}}-\frac{249 B'(t)^2 B''(t) A^{(3)}(t) A'(t)}{512 B(t)^{11/2}}-\frac{5005 B'(t)^5 A''(t) A'(t)}{8192 B(t)^{15/2}}+\frac{425425 B'(t)^8}{3145728 B(t)^{19/2}}+\frac{5 A''(t)^4}{768 B(t)^{7/2}}+\frac{1261 B''(t)^4}{61440 B(t)^{11/2}}+\frac{61 B''(t)^3}{32256 B(t)^{7/2}}+\frac{1925 B'(t)^4 A''(t)^2}{8192 B(t)^{13/2}}+\frac{127699 B'(t)^4 B''(t)^2}{163840 B(t)^{15/2}}+\frac{83 A''(t)^2 B''(t)^2}{1536 B(t)^{9/2}}+\frac{83 B'(t)^2 A^{(3)}(t)^2}{1536 B(t)^{9/2}}+\frac{659 B'(t)^2 B^{(3)}(t)^2}{10240 B(t)^{11/2}}+\frac{23 A^{(4)}(t)^2}{6720 B(t)^{5/2}}+\frac{23 B^{(4)}(t)^2}{16128 B(t)^{7/2}}+\frac{119 B'(t)^4 B''(t)}{8192 B(t)^{11/2}}+\frac{11 A''(t)^2 B''(t)}{1920 B(t)^{5/2}}+\frac{11 B'(t) A''(t) A^{(3)}(t)}{960 B(t)^{5/2}}+\frac{83 B'(t) A''(t) B''(t) A^{(3)}(t)}{384 B(t)^{9/2}}+\frac{1859 B'(t)^5 B^{(3)}(t)}{8192 B(t)^{15/2}}+\frac{5 B'(t) A''(t)^2 B^{(3)}(t)}{64 B(t)^{9/2}}+\frac{227 B'(t) B''(t)^2 B^{(3)}(t)}{1280 B(t)^{11/2}}+\frac{43 B'(t) B''(t) B^{(3)}(t)}{5376 B(t)^{7/2}}+\frac{5 B'(t)^2 A''(t) A^{(4)}(t)}{64 B(t)^{9/2}}+\frac{25 B'(t)^2 B^{(4)}(t)}{10752 B(t)^{7/2}}+\frac{527 B'(t)^2 B''(t) B^{(4)}(t)}{5120 B(t)^{11/2}}+\frac{19 A^{(3)}(t) A^{(5)}(t)}{3360 B(t)^{5/2}}+\frac{17 B'(t)^3 B^{(5)}(t)}{1024 B(t)^{11/2}}+\frac{19 B^{(3)}(t) B^{(5)}(t)}{8064 B(t)^{7/2}}+\frac{3 A''(t) A^{(6)}(t)}{1120 B(t)^{5/2}}+\frac{11 B''(t) B^{(6)}(t)}{8064 B(t)^{7/2}}+\frac{B^{(6)}(t)}{6720 B(t)^{3/2}}+\frac{B'(t) B^{(7)}(t)}{2016 B(t)^{7/2}}-\frac{A''(t) A^{(4)}(t)}{240 B(t)^{3/2}}-\frac{A^{(3)}(t)^2}{320 B(t)^{3/2}}-\frac{3 B'(t) B^{(5)}(t)}{4480 B(t)^{5/2}}-\frac{19 B''(t) B^{(4)}(t)}{13440 B(t)^{5/2}}-\frac{B^{(8)}(t)}{20160 B(t)^{5/2}}-\frac{23 B^{(3)}(t)^2}{26880 B(t)^{5/2}}-\frac{5 B'(t)^2 A''(t)^2}{512 B(t)^{7/2}}-\frac{43 A''(t) A^{(3)}(t) B^{(3)}(t)}{1344 B(t)^{7/2}}-\frac{43 A''(t) B''(t) A^{(4)}(t)}{1344 B(t)^{7/2}}-\frac{43 B'(t) A^{(3)}(t) A^{(4)}(t)}{1344 B(t)^{7/2}}-\frac{25 B'(t) A''(t) A^{(5)}(t)}{1344 B(t)^{7/2}}-\frac{61 B''(t) A^{(3)}(t)^2}{2688 B(t)^{7/2}}-\frac{25 A''(t)^2 B^{(4)}(t)}{2688 B(t)^{7/2}}-\frac{B'(t) B''(t) B^{(5)}(t)}{64 B(t)^{9/2}}-\frac{3 B'(t) B^{(3)}(t) B^{(4)}(t)}{128 B(t)^{9/2}}-\frac{5 B'(t)^3 B^{(3)}(t)}{768 B(t)^{9/2}}-\frac{31 B''(t) B^{(3)}(t)^2}{1536 B(t)^{9/2}}-\frac{25 B''(t)^2 B^{(4)}(t)}{1536 B(t)^{9/2}}-\frac{5 B'(t)^2 B^{(6)}(t)}{1536 B(t)^{9/2}}-\frac{83 B'(t)^2 B''(t)^2}{6144 B(t)^{9/2}}-\frac{119 B'(t)^3 A''(t) A^{(3)}(t)}{512 B(t)^{11/2}}-\frac{357 B'(t)^2 A''(t)^2 B''(t)}{1024 B(t)^{11/2}}-\frac{561 B'(t)^4 B^{(4)}(t)}{8192 B(t)^{13/2}}-\frac{4829 B'(t)^3 B''(t) B^{(3)}(t)}{10240 B(t)^{13/2}}-\frac{19943 B'(t)^2 B''(t)^3}{61440 B(t)^{13/2}}-\frac{385 B'(t)^6}{98304 B(t)^{13/2}}-\frac{115115 B'(t)^6 B''(t)}{196608 B(t)^{17/2}}. 
\end{math}
\end{center}

\end{document}